\documentclass[a4paper,onecolumn,11pt,accepted=2022-07-06]{quantumarticle}
\pdfoutput=1
\usepackage[letterpaper, portrait, margin=1in]{geometry}
\usepackage{graphicx, caption}
\usepackage{amsfonts}
\usepackage{amsmath,bm}
\usepackage{xcolor}
\usepackage{cite}
\usepackage{IEEEtrantools}
\usepackage{amssymb}
\usepackage{amsthm}
\usepackage{authblk}
\usepackage{color}
\usepackage{mathtools}
\usepackage{braket}
\usepackage{multirow}
\usepackage[ruled,vlined]{algorithm2e}
\usepackage[shortlabels]{enumitem}

\usepackage{tikz}
\usetikzlibrary{arrows,shapes,patterns,calc,matrix}
\pgfdeclarelayer{background}
\pgfsetlayers{background,main}
\tikzstyle{dot}=[draw=black, fill=blue!50!purple, circle, font=\color{white}, minimum width=1.75em]
\tikzstyle{smalldot}=[draw=black, fill=blue!50!purple, circle, font=\color{white}, minimum width=1.25em]

\theoremstyle{plain}
\newtheorem{thm}{Theorem}

\newtheorem{defn}[thm]{Definition} 

\newtheorem{lemma}[thm]{Lemma}

\let\vec\mathbf
\def\trans{^{\!\top}\!\!\;}
\def\herm{^\dagger}
\def\parit#1{\textup{(\textit{#1})}}
\newcommand\algname[1]{\textsf{\upshape #1}}
\makeatletter
\def\setalgdesc#1#2{%
    \expandafter\gdef\csname #1@algdesc\endcsname{#2}}
\def\algdesc#1{\csname #1@algdesc\endcsname}

\newcommand\blfootnote[1]{%
	\begingroup
	\renewcommand\thefootnote{}\footnote{#1}%
	\addtocounter{footnote}{-1}%
	\endgroup
}

\title{Fast Stabiliser Simulation with Quadratic Form Expansions}
\date{ }
\author[1]{Niel de Beaudrap}
\author[2,3]{Steven Herbert}
\affil[1]{\textit{Department of Informatics, University of Sussex, UK}}
\affil[2]{\textit{Quantinuum (Cambridge Quantum), Terrington House, 13-15 Hills Rd, Cambridge, CB2 1NL, UK}}
\affil[3]{\textit{Department of Computer Science and Technology, University of Cambridge, UK}}

\begin{document}

\maketitle

\sloppy
\begin{abstract}%
  \blfootnote{$^\dagger$
    Both authors contributed equally to the results of this article; the author order is merely alphabetical.
  }%
  \blfootnote{$^\ddagger$ 
    Contact: niel.debeaudrap@gmail.com, Steven.Herbert@Quantinuum.com
  }%
  This paper builds on the idea of simulating stabiliser circuits through transformations of \textit{quadratic form expansions}. 
  This is a representation of a quantum state which specifies a formula for the expansion in the standard basis, describing real and imaginary relative phases using a degree-2 polynomial over the integers.
  We show how, with deft management of the quadratic form expansion representation, we may simulate individual stabiliser operations in $\mathcal{O}(n^2)$ time matching the overall complexity of other simulation techniques~\cite{Aaronson2004,Anders2006,Bravyi2016}.
  Our techniques provide economies of scale in the time to simulate simultaneous measurements of all (or nearly all) qubits in the standard basis.
  Our techniques also allow single-qubit measurements with deterministic outcomes to be simulated in constant time.
  We also describe throughout how these bounds may be tightened when the expansion of the state in the standard basis has relatively few terms (has low `rank'), or can be specified by sparse matrices.
  Specifically, this allows us to simulate a `local' stabiliser syndrome measurement in time $\mathcal{O}(n)$, for a stabiliser code subject to Pauli noise --- matching what is possible using techniques developed by Gidney~\cite{gidney2021stim} without the need to store which operations have thus far been simulated.
\end{abstract}

\section{Introduction}
\label{intro}

Quantum computation, in general, is expected to be impossible to efficiently simulate with conventional (`classical') computers.
That is: for an idealised quantum circuit of one- and two-qubit gates and single-qubit measurements, it is expected that there is no randomised classical algorithm which can (either exactly or with even a modest margin of error) sample from the output distribution of the measurement outcomes, in time which scales polynomially in both the number of qubits and the number of operations.
This, together with quantum algorithms providing speed-ups over the best known classical algorithms~\cite{ShorFactoring, GroverSearch} raises the prospect of significant computational advantage through building quantum computers.
However, this also makes it difficult to test prototype quantum computers.

Fortunately, there is an important subclass of quantum circuit which can be efficiently simulated: stabiliser circuits.
These are circuits in which the measurements are restricted to projective single-qubit measurements onto the eigenstates of the Pauli matrices,
\begin{equation}
\label{eqn:PauliOperators}
\begin{aligned}
    I &=
    \text{\footnotesize$
    \begin{bmatrix}
        1 & 0 \\ 0 & 1
    \end{bmatrix}
    $},
&\quad
    X &=
    \text{\footnotesize$
    \begin{bmatrix}
        0 & 1 \\ 1 & 0
    \end{bmatrix}
    $},
&\quad
    Y &=
    \text{\footnotesize$
    \begin{bmatrix}
        0 \!&\! -i \\ i \!&\! 0
    \end{bmatrix}
    $},
&\quad
    Z &=
    \text{\footnotesize$
    \begin{bmatrix}
        1 \!&\! 0 \\ 0 \!&\! -1
    \end{bmatrix}
    $},
\end{aligned}
\end{equation}
and the unitary operations are restricted to the \emph{Clifford group} --- those unitary operators $U$, for which $UPU^\dagger$ is a `Pauli operator' (a tensor product of $\pm I$, $\pm X$, $\pm Y$, and $\pm Z$) if $P$ is also a Pauli operator.
Such a Pauli operator $P$ expresses a measurable property of an $n$-qubit state $\ket{\psi}$\,, \emph{e.g.},~expressing some correlations between hypothetical Pauli measurements on different qubits if $\ket{\psi}$ is a $+1$-eigenstate of $P$.
In this case, $U P U^\dagger$ represents a similar measurable property of the state $U \ket{\psi}$.
This observation forms the basis of the `stabiliser formalism'~\cite{GK}, which has proven extremely fruitful for the development of techniques for quantum error correction \cite{Devitt2013, Terhal2015, Roffe2019}. It also lay behind the original proof of the \emph{Gottesman--Knill theorem}~\cite{GK}, which is that a stabiliser circuit can be simulated by a classical algorithm in polynomial time, when acting on a standard basis state as input (or any other \emph{stabiliser state} $\ket{\psi}$, that is,~a~state which can be characterised as a $+1$-eigenvector of $n$ independent commuting Pauli operators).
This involves maintaining a \emph{stabiliser tableau}: an array which describes a list of Pauli operators $S_1$, $S_2$, \ldots, $S_n$, characterising $\ket{\psi}$ as the unique state which is a $+1$-eigenstate of each $S_k$.

Implicit in the original result of Ref.~\cite{GK} is that, for a stabiliser circuit on $n$ qubits, \textbf{(a)}~each single-qubit or two-qubit Clifford gate can be simulated in time $\mathcal{O}(n)$ by transformations of individual columns of the tableau, and \textbf{(b)}~measurements may be simulated in time $\mathcal{O}(n^3)$ by reduction to Gaussian elimination.
Simulation of stabiliser circuits, with an eye to improving performance, remains an active field of research.
Approaches to doing so, apart from the stabiliser formalism, do exist. For instance, Anders and Briegel~\cite{Anders2006} use the fact that every stabiliser state is equivalent up to a single-qubit Clifford operation to a `graph state' as the basis for techniques to simulate stabiliser operations in time $\mathcal{O}(n^2)$.
Bravyi~\emph{et al}.~\cite[Section 4.1]{Bravyi2019} also provide  techniques to simulate stabiliser operations, using a format which slightly generalises the stabiliser formalism, requiring $\mathcal{O}(n^2)$ time to simulate a stabiliser operation in the worst case.
Another way in which we may represent stabiliser states (and some other states apart from stabiliser states) is through an expansion
\begin{equation}
    \label{eqn:QFE}
        \ket{\psi}
    \,=\;
        C \!\!\!
        \sum_{\vec x \in \{0,1\}^r} \!\!\!
            \mathrm e^{2 \pi i \:\!\mathbf Q(\vec x)} \ket{f(\vec x)},
\end{equation}
for some integer $0 \le r \le n$, a function $f: \{0,1\}^r \to \{0,1\}^n$ which may be interpreted as a linear or affine transformation mod $2$, $\mathbf Q$ a polynomial of degree at most $2$, and $C$ a normalising constant.
Extending the terminology of Ref.~\cite{dBQuadratic} slightly, we call a form such as in Eqn.~\eqref{eqn:QFE} a \emph{quadratic form expansion}.
Examples of these expansions are essentially as old as the stabiliser formalism in the context of error correction~\cite{CalderbankGoodQEC}; and if generalised to express unitary operators as well as states, have proven useful in the study of circuits of Clifford operations~\cite{dehaene, VDNClassical, JBVClassical}, measurement based quantum computation~\cite{dBQuadratic}, and equality testing for unitary circuits~\cite{Amy_2019}.
In particular, implicit in the observations of van~den~Nest~\cite{VDNClassical}, and the techniques of Amy~\cite{Amy_2019} are polynomial-time algorithms to simulate stabiliser circuits.

Despite the availability of other simulation methods, the stabiliser formalism (including minor variations) remains the \emph{de facto} standard for reasoning about stabiliser circuits.
This is partly due to an elaboration described by Aaronson and Gottesman~\cite{Aaronson2004}, who described improved simulation techniques resulting in an $\mathcal{O}(n^2)$ bound to simulate measurements.
Note, however, that these methods do not track global phase factors, which may be important for applications which leverage the simulation of stabiliser circuits to perform more difficult simulation tasks, as in the work of Bravyi \textit{et al}~\cite{Bravyi2019}.

As an arbitrary stabiliser state on $n$~qubits requires at least $\tfrac{1}{2}n^2$ bits to represent~\cite[Corollary~21]{Gross2006}, a careful choice of data structure is needed to simulate arbitrary stabiliser operations in time asymptotically less than $\mathcal{O}(n^2)$.
Failing this, we may consider techniques which permit better performance for certain operations or under certain conditions.
For instance, the result of Anders and Briegel~\cite{Anders2006} is motivated by the ability to represent single-qubit operations and measurements on product states in $\mathcal{O}(1)$ time, and representing other operations in time $\mathcal{O}(d^2)$, for a parameter which may only be bounded by $d \in \mathcal{O}(n)$ in the worst case but which in some cases may be substantially smaller.
Similarly, previous work of ours~\cite{beaudrap2019quantum} describes techniques to simulate specialised Clifford circuits for concurrent entanglement swapping in limited quantum architectures --- involving only Pauli and controlled-NOT gates, but $X$-eigenstate preparations and measurements --- with a representation of the state as a quadratic form expansion as in Eqn.~\eqref{eqn:QFE}.\footnote{%
        Note that our terminology here is different from that of Ref.~\cite{beaudrap2019quantum}; note also that in that work, the corresponding polynomial $Q$ happens to have degree~1.}
In that work, the unitary gates could be performed in $\mathcal{O}(r)$ time, and measurements in time $\mathcal{O}(nr^2)$, where $r \le n$ governs the sum index as in Eqn.~\eqref{eqn:QFE}.
Guan and Regan~\cite{GuanRegan2019} describe techniques to evaluate outcome probabilities for total measurements following a unitary Clifford circuit, which would be advantageous when either the Clifford circuit contains very few Hadamard gates, or the number $M$ of Clifford gates in the circuit scales is bounded by $M \in \mathcal{O}(n^{2/(\omega-1)})$.%
    \footnote{%
        Here, $2 \le \omega < 2.3729$ is the exponent of the optimal algorithm for $n \times n$ matrix multiplication.
        We discuss this interpretation of the results of Ref.~\cite{GuanRegan2019} in Section~\ref{disc}.
    }
However, a more natural specialisation is presented by Gidney~\cite{gidney2021stim}, motivated by the use case of simulating error correction procedures.
There, he describes a refinement of the algorithm by Aaronson and Gottesman~\cite{Aaronson2004} to simulate any `deterministic measurement' (of a multi-qubit observable $P$, for which $\ket{\psi}$ happens to be an eigenvector) in time $\mathcal{O}(n)$.

In this article, we describe explicit techniques to simulate stabiliser operations in $\mathcal{O}(nr) \subseteq \mathcal{O}(n^2)$ time, on $n$-qubit stabiliser states represented by quadratic form expansions, where $r$ (the `rank' of the quadratic form expansion) again governs the sum index as in Eqn.~\eqref{eqn:QFE}.
This meets the same general asymptotic bound as the stabiliser formalism, and may be made more efficient for states that happen to have fewer terms when expanded over the standard basis, as in our earlier work~\cite{beaudrap2019quantum}.
(Our results mainly concern `weak' simulation, which is to say sampling from the distribution of measurement outcomes: see the remarks towards the end of Section~\ref{sec:elementaryStabiliserOpns}.)
Furthermore, our techniques allow for significant improvements in run-time, in particular cases where the function $f$ and the quadratic form $Q(\vec x)$ may be represented by sparse data structures.
Remarkably, our techniques allow us to compute the outcomes of deterministic single-qubit measurements (in the $X$-, $Y$-, or $Z$-basis) in constant time.
Together with partial use of our tighter run-time bounds for sparse data structures, this allows us to simulate deterministic measurements of `local' multi-qubit Pauli operators\footnote{%
    A multi-qubit operator is `$k$-local' if it acts (non-trivially) on at most $k$ qubits.
    We simply say that it is `local', if the number of qubits it acts on is bounded by some constant.
}
in $\mathcal{O}(n)$ time, a bound which again may be tightened for particular Pauli observables to be measured or when the function $f$ is represented by a sparse data structure.
This matches the performance of techniques presented by Gidney~\cite{gidney2021stim} for deterministic Pauli measurements, and suggests that our techniques may be well-suited to simulation of encoded stabiliser circuits.

The structure of the article is as follows.
In Section~\ref{sec:preliminaries} we briefly introduce the stabiliser operations of interest, describe quadratic form expansions, and present the format for quadratic form expansions which is used by our techniques.
Section~\ref{sec:computing-w-QFEs} describes data structures and helpful techniques and subroutines to manage quadratic form expansions (with technical details deferred to the Appendices), and the computational model which we assume for bounding the run-time complexity.
Section~\ref{sec:simulCliffOpns} then describes the procedures to simulate stabiliser operations on quadratic form expansions together with their run-time bounds.
In particular, Section~\ref{runtimes} summarises the run-times of our simulation techniques for a number of individual stabiliser operations.
Our main results concerning the asymptotic run-time to simulate circuits and other procedures involving multiple stabiliser operations are presented in Section~\ref{sec:complexity-procedures}.
(While our results mainly concern `weak' simulation, Sections~\ref{sec:generalCircuitSim} and~\ref{sec:parallelMeasurements} do contain some results which bear on `strong' simulation, \emph{i.e.},~computing the probability of specific outcomes for a given set of measurements.)
Finally, we conclude in Section~\ref{disc} with a discussion of how our results relate  to other simulation methods or results connected to path-sums, and the potential applications of our techniques.

\section{Preliminaries}
\label{sec:preliminaries}

In this Section, we broadly describe the stabiliser circuit model, and introduce quadratic form expansions as a simple approach to simulating them.
This will serve to show how quadratic form expansions can be used pedagogically, as an alternative to the stabiliser formalism, to demonstrate the Gottesmann--Knill Theorem.
It also serves as the starting point to describe the more involved technical contributions of this article.

We commonly use bold-face letters such as $\vec x$ to represent column vectors, whose transpose $\vec x\trans = {[\,x_1\;\;x_2\;\;\cdots\;]}$ is a row-vector.
In particular, we use $\vec e_j$ to denote some vector whose only non-zero entry is a $1$ in position $j$; the length of the vector will be clear from context.
We use $\mathrm{diag}(\vec a)$ to denote a matrix with coefficients $a_1$, $a_2$, \emph{etc}. on the diagonal and $0$ elsewhere.

\subsection{Stabiliser circuits}
\label{sec:elementaryStabiliserOpns}

We consider stabiliser circuits consisting of one- or two-qubit Clifford gates and Pauli measurements.
The `Clifford gates' consist of operations such as
\begin{equation}
    \label{eqn:simpleCliffordGates}
    H = \tfrac{1}{\sqrt 2}
    \text{\setlength{\arraycolsep}{1ex}\footnotesize
    $\begin{bmatrix} 1 & 1 \\ 1 & -1 \end{bmatrix}$},
    \,\,\,\,\,\,\,\,
    S = \text{\footnotesize$\begin{bmatrix} 1 & 0 \\ 0 & i \end{bmatrix}$},
    \,\,\,\,\,\,\,\,
    \mathrm CZ = \text{\setlength{\arraycolsep}{1ex}
    \footnotesize
    $\begin{bmatrix} 1 & 0 & 0 & 0 \\ 0 & 1 & 0 & 0 \\ 0 & 0 & 1 & 0 \\ 0 & 0 & 0 & -1 \end{bmatrix}$},
    \,\,\,\,\,\,\,\,
    \mathrm CX = \text{\setlength{\arraycolsep}{1ex}
    \footnotesize
    $\begin{bmatrix} 1 & 0 & 0 & 0 \\ 0 & 1 & 0 & 0 \\ 0 & 0 & 0 & 1 \\ 0 & 0 & 1 & 0 \end{bmatrix}$},
\end{equation}
and the Pauli operators of Eqn.~\eqref{eqn:PauliOperators}.%
    \footnote{%
        As usual, we implicitly extend each of these to $n$-qubit unitary transformations for $n \ge 1$, by taking tensor products with the identity operator as required.
    }
These gates are so-called as they are elements of the Clifford group: those unitary operations such that $UPU^\dagger$ is a Pauli operator, whenever $P$ is itself a Pauli operator.
(For an introductory treatment of the Clifford group, see Ref.~\cite[Chapter 10]{NandC}.)
Furthermore, the Clifford group may be generated (up to global phases) by either of the sets $\{S, H, \mathrm CZ\}$ or $\{S,H,\mathrm CX\}$, together with tensor products with the identity operator.
The `Pauli measurements' consist of single-qubit measurements in any of the following three bases:
\begin{equation}
    \begin{aligned}
        \ket{\texttt0}, \\
        \ket{\texttt1},
    \end{aligned}
    \qquad
    \qquad
    \begin{aligned}
        \ket{\texttt+} = \tfrac{1}{\sqrt 2}\ket{\texttt0} + \tfrac{1}{\sqrt 2}\ket{\texttt1}, \\
        \ket{\texttt-} = \tfrac{1}{\sqrt 2}\ket{\texttt0} - \tfrac{1}{\sqrt 2}\ket{\texttt1},
    \end{aligned}
    \qquad
    \qquad
    \begin{aligned}
        \ket{\texttt{+i}} = \tfrac{1}{\sqrt 2}\ket{\texttt0} +  \tfrac{i}{\sqrt 2}\ket{\texttt1}, \\
        \ket{\texttt{-i}} = \tfrac{1}{\sqrt 2}\ket{\texttt0} -  \tfrac{i}{\sqrt 2}\ket{\texttt1}.
    \end{aligned}
\end{equation}
These are the eigenstates of the $Z$, $X$, and $Y$ operators, and we call them the $Z$-basis, the $X$-basis, and the $Y$-basis, respectively.
These states may all be mapped to the standard basis $\{\ket{\texttt0},\ket{\texttt1}\}$ by single-qubit operations, so that single-qubit $Z$-basis measurements and Clifford gates suffice to simulate all Pauli measurements.
However, we will prefer to simulate these measurements without such reduction to $Z$-basis measurements.

We consider stabiliser circuits where measurements may occur in the middle of the circuit, and where operations that are performed after any measurement may depend on the measurement outcome.
The problem in which we are mainly interested is \emph{weak simulation}~\cite{VDNClassical, JozsaClassical} of such stabiliser circuits.
This is the problem of sampling from the output distribution of the measurements from a stabiliser circuit --- or, in other words, to simulate the \emph{behaviour} of a stabiliser circuit with measurements.
This is to be set against `strong simulation', which is the task of evaluating the probability of some of some \emph{particular} result for a subset of the measurement outcomes.
While we are interested in principle in strong simulation as well (and presenting in Sections~\ref{sec:generalCircuitSim} and~\ref{sec:parallelMeasurements} some results concerning strong simulation), our main objective is to describe efficient algorithms to sample from the output distribution of measurements of stabiliser circuits.

\subsection{Quadratic form expansions}
\label{sec:prelims-QFEs}

It does not require great mathematical sophistication to use the stabiliser formalism.
However, the meaning behind stabiliser tableaus and transformations of them is slightly indirect for people first encountering them.
It is also less easy to motivate to those who will not also need them for work in quantum error correction or related topics.
As an alternative, we present quadratic form expansions as a means of more directly representing stabiliser states, and simulating the effects of stabiliser operations on them.

\medskip
\noindent
We consider normalised state-vectors which can be represented as follows:
\begin{equation}
\label{eqn:QFE-roughly}
        \ket{\psi}
    \;=\;
        C
            \!\!
            \sum_{\mathbf x \in \{0,1\}^r} 
                \!\!\!
                i^{\;\!\mathbf Q(\vec x)}
                \,
                \ket{A \mathbf{x} \oplus \mathbf{b}}    \;,
\end{equation}
where  $C \in \mathbb C^\ast$ combines a normalising factor and possibly a global phase, $r \ge 0$ is an integer, $A$ is an $n \times r$ matrix (the `expansion matrix'), $\vec b$ an $n \times 1$ vector, and $\mathbf Q: \mathbb Z^r \to \mathbb Z$ is a function which determines either a real or an imaginary relative phase.
(For the `degenerate' case $r = 0$, the sum is just a single term consisting of $\ket{\vec b}$.)
Both $A$ and $\vec b$ have coefficients over $\{0,1\}$, and the expression $A \vec x \oplus \vec b$ represents the reduction of $A \vec x + \vec b$ modulo~2.
(Naturally, the expression in the exponent of $i$ may be evaluated modulo $4$; our techniques allow us to work gracefully with both forms of modular arithmetic at the same time.)
This describes $\ket{\psi}$ as a superposition, where the terms of the superposition are either purely real or purely imaginary functions of a vector $\vec x$ which determines the term of the superposition.

When the expansion matrix $A$ has rank $r$ as a matrix modulo~2, then this representation serves to describe the correlations (or lack of correlations) in the outcomes of potential single-qubit standard basis measurements.
For instance:
\begin{itemize}
\item
    The outcomes of a $Z$-basis measurement on three qubits in the state $\smash{\tfrac{1}{\sqrt{2}}(\ket{\texttt{000}} + \ket{\texttt{111}})}$ would yield the same result.
    This state has a quadratic form expansion with normalisaion $C = 2^{-1/2}$, rank $r=1$, a matrix $A = {[\,1\;\;1\;\;1\,]\:\!\trans}$ of shape $3 \times 1$, and $\mathbf Q(\vec x) = 0$, $\vec b = \vec 0$.
    This describes how the $Z$-basis measurement outcomes for all of the qubits are characterised by a single bit $x \in \{0,1\}$.
\item
    We may contrast this with the outcomes of $Z$-basis measurements on the product state $\smash{\tfrac{1}{2\sqrt{2}}(\ket{\texttt0} + \ket{\texttt1})(\ket{\texttt0} + \ket{\texttt1})(\ket{\texttt0} + \ket{\texttt1})}$, which would be independent and maximally random.
    This state has a quadratic form expansion with $C = 2^{-3/2}$, $r = 3$, $A = I_3$, and again $\mathbf Q(\vec x) = 0$ and $\vec b = \vec 0$.
    This describes how three independent bits $x_1, x_2, x_3 \in \{0,1\}$ are necessary to characterise $Z$-basis measurement outcomes on the state.
\end{itemize}
Thus the expansion matrix $A$ describes, usually in a non-unique way, a coherent superposition over different standard basis terms, each obtained from some bit-string $\vec x \in \{0,1\}^r$.

If we were to replace the relative phase $i^{\mathbf Q(\vec x)}$ by a more general function $\exp(i \pi \mathbf Q(\vec x))$ for $\mathbf Q: \mathbb Z^r \to \mathbb R$, where $\mathbf Q$ is a polynomial with real coefficients and degree at most $2$, this would be a slight generalisation of what is called a \emph{quadratic form expansion}~\cite{dBQuadratic}
for $\ket{\psi}$.%
    \footnote{%
        This terminology arises from representing $\mathbf Q$ by a polynomial in which every term is quadratic, using the fact that $x_j^2 = x_j$ for $x_j \in \{0,1\}$, so that $\mathbf Q$ is a `quadratic form'.
        One might similarly consider Eqn.~\eqref{eqn:QFE-roughly} to be an example of a `path-sum' representation~\cite{Amy_2019} of a state.
    }
In principle, if one does not impose any bound on the length $r$ of the summation index, one may represent an arbitrary quantum state in this way: see Refs.~\cite{dBQuadratic,Amy_2019} for more details.
Even if the relative phases are restricted to $i^{\mathbf Q(\vec x)}$ for an arbitrary function $\mathbf Q: \mathbb Z^r \to \mathbb Z$, this representation could be used to express any state generated by a Clifford+T circuit (see for instance the discussion on page~\pageref{discn:Clifford+T} in Section~\ref{disc}).
Instead, we consider the case where $\mathbf Q: \mathbb Z^r \to \mathbb Z$ arises from a symmetric integer matrix $Q$ by the equation
\begin{equation}
    \mathbf Q(\vec x) = \vec x\trans \!\!\; Q \!\; \vec x.
\end{equation}
We then refer to $Q$ as a \emph{Gram matrix} for $\mathbf Q$.
Essentially as a corollary of any one of Refs.~\cite{dehaene,VDNClassical,JBVClassical,dBQuadratic}, in the case that $0 \le r \le n$, where $A$ has rank $r$ as a matrix modulo~2, and where $C = 2^{-r/2}$, the result of the expression Eqn.~\eqref{eqn:QFE-roughly} is a normalised stabiliser state; and conversely, any normalised stabiliser state can be represented under such constraints.
We introduce the convention of using a symmetric matrix $Q$ to govern the relative phases, which allows us to represent the imaginary phases from $S$ gates and the sign phases from $\mathrm CZ$ gates on the same footing.
For instance, we may represent the state $\ket{\psi_1} = (I \otimes S) \ket{\texttt{++}} = \tfrac{1}{2} \bigl( \ket{\texttt{00}} + i \ket{\texttt{01}} + \ket{\texttt{10}} + i \ket{\texttt{11}})$ and the state $\ket{\psi_2} = \mathrm CZ \ket{\texttt{++}} = \tfrac{1}{2} \bigl( \ket{\texttt{00}} + \ket{\texttt{01}} + \ket{\texttt{10}} - \ket{\texttt{11}})$ both by quadratic form expansions, as
\begin{equation}
    \begin{aligned}[b]
    \ket{\psi_1} &= \tfrac{1}{2}
                    \!\!\!\sum_{\vec x \in \{0,1\}^2}\!\!\!
                        i^{\;\!
                            [\begin{smallmatrix} x_1 & x_2 \end{smallmatrix}]
                            \bigl[\begin{smallmatrix} 0 & 0 \\ 0 & 1 \end{smallmatrix}\bigr]
                            \bigl[\begin{smallmatrix} x_1 \\ x_2 \end{smallmatrix}\bigr]}
                    \,
                    \ket{x_1, x_2}
                    ,
    &\,\,
    \ket{\psi_2} &= \tfrac{1}{2}
                    \!\!\!\sum_{\vec x \in \{0,1\}^2}\!\!\!
                        i^{\;\!
                            [\begin{smallmatrix} x_1 & x_2 \end{smallmatrix}]
                            \bigl[\begin{smallmatrix} 0 & 1 \\ 1 & 0 \end{smallmatrix}\bigr]
                            \bigl[\begin{smallmatrix} x_1 \\ x_2 \end{smallmatrix}\bigr]}
                    \,
                    \ket{x_1, x_2}
                    .
\end{aligned}
\end{equation}
In each case, we take $C = \tfrac{1}{2}$, $A = I_2$, and $\vec b = \vec 0$, and set $Q$ to the appropriate $2 \times 2$ matrix in the imaginary exponent.
Note that, as $Q$ is symmetric, relative phases which involve distinct variables $x_j$ and $x_k$ must arise from two contributions $x_j x_k \,+\, x_k x_j$ in $\mathbf Q(\vec x)$ and thereby represent a phase $(-1)^{x_j x_k}$ rather than $i^{\;\!x_j x_k}$.

In the discussion above, the case where the expansion matrix $A$ has rank $r$ is significant: it implies in particular that distinct terms of the quadratic form expansion are orthogonal.
In this case, a quadratic form expansion representing a normalised state would satisfy $C = \omega \cdot 2^{-r/2}$, for $\omega \in \mathbb C$ a phase factor.
(We discuss the importance of the rank of $A$ to simulation techniques below.)
Our techniques also allow us to constrain the value of this global phase to a power of $\tau = \sqrt{i\:\!} = \exp(i \pi/4)$, for states obtained by simulating stabiliser operations on states for which $\omega = 1$.
In order to simplify discussion of simulation techniques for stabiliser circuits, except where otherwise specified, any reference to a ``quadratic form expansion'' from this point on should be understood to refer to an expression
\begin{equation}
\label{eqn:QFE-specific}
        \ket{\psi}
    \;=\;
        \frac{\tau^g}{\sqrt{2^r}}
            \!
            \sum_{\mathbf x \in \{0,1\}^r} 
                \!\!\!
                i^{\;\!\vec x\trans \!\!\: Q \:\! \vec x}
                \,
                \ket{A \mathbf{x} \oplus \mathbf{b}}    \;,
\end{equation}
which represents a normalised stabiliser state, for integer matrices $A$ and $Q$ where furthermore $Q$ is symmetric, and where $A \vec x \oplus \vec b$ denotes the reduction modulo~2 of the integer vector $A \vec x + \vec b$.

\subsection{On the role of rank in simulation}
\label{sec:roleOfRank}

The number $r \ge 0$ of columns required for the expansion matrix is related to the number of bits required to generate a particular term of a quadratic form expansion.
To this point, we have not imposed any restrictions on $r$, though we have alluded to the importance of the case where $r$ is the rank of $A$ considered as a matrix modulo~2.

The importance of this property is as follows.
Let $\mathrm{null}\,A$ represent the nullspace of $A$ considered as a matrix modulo~$2$ (\emph{i.e.},~the set of vectors $\vec x$ for which $A \vec x \equiv \vec 0 \!\mod{2})$.
Then, let $\mathrm{rank}\,A = r - \dim(\mathrm{null}\,A)$ be its rank as a matrix modulo~2.
For each $\vec y \equiv A \vec x \!\mod{2}$, we also have $\vec y \equiv A (\vec x + \boldsymbol \delta) \!\mod{2}$ for each $\boldsymbol \delta \in \mathrm{null}\,A$.
As a result, each term $\ket{\vec y}$ in a quadratic form expansion is one out of $2^{\dim(\mathrm{null}\,A)}$ colinear terms, which in some cases might interfere destructively.
If $r = \mathrm{rank}\,A$, each term $\ket{\vec y}$ in a quadratic form expansion is in fact orthogonal to the others, so that there is no destructive interference of outcomes.
As a result, to determine whether a measurement outcome on all qubits is possible, it is sufficient to determine whether it is present.
However, if $r > \mathrm{rank}\,A$, then it is in principle necessary to consider how the terms constructively or destructively interfere to determine whether a measurement outcome is possible.

For the above reason, it is practically motivated to consider quadratic form expansions in which $r = \mathrm{rank}\,A$ --- and we will often refer to $r$ as the \emph{rank} of the quadratic form expansion for this reason.%
    \footnote{%
        This of course should not be confused with the \emph{stabiliser rank} of a state, as described in Refs.~\cite{Bravyi2019, Bravyi2016, Bravyi2016b}; one could refer to our notion of `rank' as the `expansion matrix rank' if it became necessary to avoid confusion.
    }
This raises the question of how the rank $r$ may change under simulation of different operations.

Note that because the Pauli operations of Eqn.~\eqref{eqn:PauliOperators} and the Clifford gates $\mathrm CX$, $\mathrm CZ$, and $S$ of Eqn.~\eqref{eqn:simpleCliffordGates} are `monomial' (having at most one non-zero entry per row/column), it is not difficult to show that they may be simulated by a transformation of a quadratic form expansion which does not change the number of columns $r$ of the expansion matrix.
In particular, the Pauli operators and the diagonal operations $S$ and $\mathrm CZ$ may be simulated with no changes to the expansion matrix at all; and $\mathrm CX$ may be simulated by a simple row-operation on $A$.
To contrast, as the Hadamard gate does not map standard basis states to other standard basis states, simulating it will require changes to the value of $r$, either to increase or decrease it, to ensure that it corresponds to the rank of the expansion matrix.
(Similar remarks apply for $X$- and $Y$-basis measurements.)

As a result of the fact that $H$ is self-inverse, and the way in which quadratic form expansions emphasise the role of the standard basis, this requires an analysis of when a Hadamard gate leads to destructive interference of terms.
This ultimately requires that we solve a system of equations to determine whether such destructive interference takes place.
If we do not constrain the value of $A$ in any way, this will in practise require Gaussian elimination, at a cost of $\mathcal{O}(nr^2) \subseteq \mathcal{O}(n^3)$.
As this does not compare favourably to the $\mathcal{O}(n^2)$ time to simulate any elementary stabiliser operation using the stabiliser techniques of Aaronson and Gottesman~\cite{Aaronson2004}, one may consider what techniques would allow one to avoid the worst-case performance of Gaussian elimination.

\subsection{Our contribution}

In this article, we describe techniques to simulate stabiliser operations using quadratic form expansions where $r = \mathrm{rank}\,A$, in time $\mathcal{O}(nr) \subseteq \mathcal{O}(n^2)$.
This is done by maintaining $A$ in a particular form which makes it easy to certify that $A$ has rank $r$, and in so doing essentially amortise the cost of Gaussian elimination.
We then describe techniques to simulate stabiliser operations using quadratic form expansions, more explicitly than has been done in the related literature~\cite{dehaene,VDNClassical,JBVClassical,Amy_2019,dBQuadratic} and with a clear asymptotic analysis.
The summary of the asymptotic run-times of each of our subroutines is provided in Section~\ref{runtimes}.

The headline complexity of $\mathcal{O}(n^2)$ to simulate stabiliser operations with quadratic form expansions, matches that of existing techniques~\cite{Aaronson2004,Anders2006,Bravyi2019}.
This obscures important differences in how efficiently individual operations may be simulated.
For instance, stabiliser-based techniques~\cite{Aaronson2004,Bravyi2019} can simulate each of the Clifford gates of Eqn.~\eqref{eqn:simpleCliffordGates} in time $\mathcal{O}(n)$, while our techniques variously require $\mathcal{O}(r^2)$ or $\mathcal{O}(nr)$.
(When $r \in \mathcal{O}(\sqrt{n})$, at least those operations which may be simulated in time $\mathcal{O}(r^2)$ may be asymptotically as efficient or more efficient than in the stabiliser formalism; though we do not expect that it will be easy in general to determine that $r$ is bounded in this way for a given stabiliser circuit.)
The principal advantage provided by quadratic form expansions is that they lend themselves to sparse matrix representations.
Thus, the bounds $\mathcal{O}(r^2)$ and $\mathcal{O}(nr)$ themselves obscure substantially tighter bounds, that hold when the expansion matrix $A$ and the Gram matrix $Q$ are sparse.
These tighter bounds allow us to describe procedures to simulate deterministic measurements of local stabiliser operators in time $\mathcal{O}(n)$.

We give a careful account of the complexity of each of our subroutines, in terms of the number of non-zero coefficients in the rows and columns of $A$ and $Q$.
This, in combination with the fact that our techniques for maintaining rank involves maintaining at least $r$ columns in $A$ with precisely one non-zero coefficient, may be expected to provide a significant advantage for simulations when the stabiliser circuit in question has enough structure to certify that the states may be represented by such `sparse' quadratic form expansions.

\section{Managing quadratic form expansions}
\label{sec:computing-w-QFEs}

In this Section, we describe a number of techniques and procedures which will serve to simplify our account of our simulation techniques for quadratic form expansions, and of the run-time analysis for those simulation techniques when the quadratic form expansion has a sparse expansion matrix $A$ and Gram matrix $Q$.
We do so by describing certain constraints on quadratic form expansions, and ways of transforming quadratic form expansions which may be involved in simulating certain  stabiliser operations.
We then describe a list of helper subroutines which encapsulate these results (and also Lemmas~\ref{lemma:Gram-matrix-mixed-moduli} and~\ref{lemma:QFE-change-of-variables}) for use in procedures to simulate stabiliser operations.

\subsection{Principal row forms}
\label{sec:principal-row-form}

Following the observations of Section~\ref{sec:roleOfRank}, we wish for Eqn.~\eqref{eqn:QFE-specific} to represent the decomposition of $\ket{\psi}$ without repeating any standard basis terms, and in particular, so that it does not contain terms which \emph{cancel}.
To this end we impose the condition that the expansion matrix $A$ has rank $r$.

As we simulate operations by modifying the matrix $A$, we must determine how to maintain the invariant of $A$ having rank $r$.
In particular, we must bound the number of columns of $A$ above by $n$ at the end of each simulated operation.
At the same time, we wish to avoid performing Gaussian elimination when the rank of $A$ is in question, with a worst-case performance of $\mathcal{O}(nr^2) \subseteq \mathcal{O}(n^3)$.
Avoiding this cost is one of the main results of Aaronson and Gottesman~\cite{Aaronson2004}, who do so by roughly doubling the amount of stored data.
For quadratic form expansions, we may instead avoid Gaussian elimination by imposing constraints on the form of $A$.

\begin{defn}
    A matrix $A$ with shape $n \times r$ is in \emph{principal row form} if each $\vec e_k\trans$ occurs at least once as a row of $A$; that is, if there is a map $p: \{1,\ldots,r\} \to \{1,\ldots,n\}$ such that $\vec e_{p(k)}\trans A = \vec e_k\trans$ for each $1 \le k \le r$.
    We call such a map $p$, a \emph{principal index map} for $A$; a row $j = p(k)$ for $1 \le k \le r$ is a \emph{principal row} of $A$.
\end{defn}
\noindent
For $A$ in principal row form, each column $A \vec e_c$ of $A$ is the only column which has a non-zero entry in row $p(c)$, by construction.
It then follows that the columns are linearly independent, so that $A$ has rank $r$.
The choice of principal row $p(c)$ which stores a given vector $\vec e_c\trans$ may be non-unique for a given $A$; it is enough for our purposes to indicate one such row for each $1 \le c \le r$.
The role of the principal index map is important enough that we include it in the data describing a quadratic form expansion, as described in Section~\ref{sec:data-structures}.

In simulating operations on a quadratic form expansion, we must occasionally perform operations to the matrix $A$, which could take it out of principal row form.
We must then do additional work to prevent $A$ from being put out of principal row form, or to put it back into  principal row form, by performing suitable column operations and changes to the index map $p$.
In Section~\ref{sec:QFE-subroutines}, we describe some procedures to assist with maintaining $A$ in principal row form.

\subsection{Data structures and programming model}
\label{sec:data-structures}

A procedure to simulate a stabiliser operation on an $n$-qubit state $\ket{\psi}$, which is given as a quadratic form expansion, will in practice mean a procedure which acts on a tuple $\mathcal E = (n,r,g,Q,A,\vec b,p)$ specifying that quadratic form expansion.
These consist of the parameters required to specify a quadratic form expansion as in Eqn.~\eqref{eqn:QFE-specific}, together with a principal index map $p$ for the expansion matrix $A$.

We are particularly interested in the case of quadratic form expansions, where the matrices $A$ and $Q$ may be sparse.
That is, we suppose that there are integers $s,t,w \ge 0$ which bound from above, respectively, the number of non-zero coefficients in each row of $A$, in each column of $A$, and in each row/column of $Q$.

The vector $\vec b$ may be stored straightforwardly as an array or buffer of $n$ bits, and the principal index map we suppose to be represented as an integer array.%
    \footnote{%
        The number of columns in $A$ may temporarily increase to $n+1$ during a simulated operation.
        For this reason, it would be simplest to define $p$ as an array of length $n+1$.
    }
However, motivated by the setting where $s,t,w \ll n$, our techniques rely on the matrices $A$ and $Q$ being stored using a sparse matrix structure.
In particular, for either of these matrices:
\begin{itemize}
\item 
    We suppose that a record of the number of non-zero entries in any row or column $j$ is maintained, and that it is possible to iterate over the non-zero entries in such a row or column (in order), omitting any zero coefficients in doing so.
\item
    We also suppose that the data structure allows constant-time insertion or deletion of non-zero entries in between two given non-zero entries.
\end{itemize}
(This could be achieved with an array of pointers to nodes, which themselves form a list-like structure along the rows, columns, and diagonal.)
We suppose also that an explicit entry is stored for each element on the diagonal of $Q$, whether or not that entry is non-zero, allowing constant-time access to and modification of diagonal entries.
A schematic representation of these data structures is shown in Figure~\ref{fig:QFE-representation-schematic}.

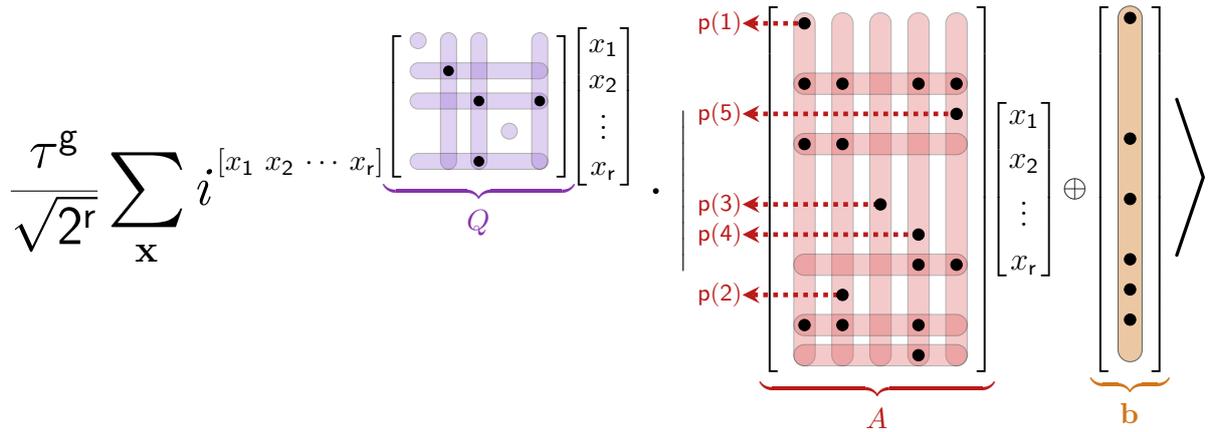
\begin{figure}[t]
    \begin{huge}
    \begin{gather*}
        \mspace{-18mu}
        \mathsf
        {%
        \frac{\tau^g}{\sqrt{2^r}}  \!\!\;
            \mathop{\raisebox{-1ex}{\scalebox{2.82}{$\sum$}}}_{\phantom\vert \mathbf x \phantom\vert}
                \!\!\;
                \mathit{i}^{\raisebox{.5ex}{\normalsize$\;\!
                [\begin{matrix}
                    x_{\mathsf 1} \!\!\!\:&\!\!\!\: x_{\mathsf 2} \!\!\!\:&\!\!\!\: \cdots \!\!\!\:&\!\!\!\: x_{\mathsf r}
                \end{matrix} ]
                \raisebox{7.5mm}{%
                    $\begin{color}{blue!60!red!80!white}\smash{\underbrace{\color{black}
                        \left[\;
                            \mathclap{\begin{matrix} \\[9mm] \end{matrix}}
                            \begin{aligned}
                            \begin{tikzpicture}
                                \foreach \r in {1,...,5} {%
                                    \foreach \c in {1,...,5} {%
                                        \coordinate (Q-\r-\c) at ({\c*.4}, {-\r*.4});
                                        \foreach \d/\a in {%
                                            E/0, NE/45, N/90, NW/135, W/180, SW/225, S/270, SE/315} {
                                                \coordinate (Q-\r-\c-\d) at ($(Q-\r-\c) + (\a:3pt)$);
                                        }
                                    }
                                }
                                \foreach \r in {2,3,5} {
                                    \filldraw [draw=black, fill=blue!70!red!60!white, opacity=0.3]
                                            (Q-\r-1-W) arc (180:270:3pt)
                                                -- (Q-\r-5-S) arc (-90:90:3pt)
                                                -- (Q-\r-1-N) arc (90:180:3pt) -- cycle;
                                    \filldraw [draw=black, fill=blue!70!red!60!white, opacity=0.3]
                                            (Q-1-\r-N) arc (90:180:3pt)
                                                -- (Q-5-\r-W) arc (-180:0:3pt)
                                                -- (Q-1-\r-E) arc (0:90:3pt) -- cycle;
                                }
                                \foreach \r in {1,4} {
                                    \filldraw [draw=black, fill=blue!70!red!60!white, opacity=0.3]
                                        (Q-\r-\r.center) circle (3pt);
                                }
                                \foreach \r/\c in {2/2, 3/3, 3/5, 5/3} {
                                    \filldraw [black] (Q-\r-\c) circle (1.8pt);
                                }
                            \end{tikzpicture}
                            \end{aligned}
                        \;\right]}%
                        _{\text{\normalsize $Q^{\mathclap{\phantom|}}$}}}
                    \end{color}
                    \!\!
                    \left[\:\!
                        \mathclap{\begin{matrix} \\[9mm] \end{matrix}}
                        \begin{matrix}
                            x_{\mathsf 1} \\ x_{\mathsf 2} \\ \vdots \\ x_{\mathsf r}
                        \end{matrix}
                    \:\!\right]
                $}
                $}}
                \;\!\!
                \cdot
                \:\!
                \Biggl\lvert\;
                \raisebox{1mm}{\large%
                $\mspace{32mu}
                \begin{color}{red!70!black!90!white}\underbrace{\color{black}
                    \left[\;
                        \mspace{-56mu}
                        \mathclap{\begin{matrix} \\[20.5ex] \end{matrix}}
                        \begin{aligned}
                        \begin{tikzpicture}
                            \foreach \r in {1,...,12} {%
                                \foreach \c in {0,...,5} {%
                                    \coordinate (A-\r-\c) at ({\c*.5}, {-\r*.4});
                                    \foreach \d/\a in {%
                                        E/0, NE/45, N/90, NW/135, W/180, SW/225, S/270, SE/315} {
                                            \coordinate (A-\r-\c-\d) at ($(A-\r-\c) + (\a:4pt)$);
                                    }
                                }
                            }
                            \foreach \c in {1,...,5} {
                                \filldraw [draw=black, fill=red!80!black!70!white, opacity=0.3]
                                        (A-1-\c-N) arc (90:180:4pt)
                                            -- (A-12-\c-W) arc (-180:0:4pt)
                                            -- (A-1-\c-E) arc (0:90:4pt) -- cycle;
                            }
                            \foreach \r in {3,5,9,11,12} {
                                \filldraw [draw=black, fill=red!80!black!70!white, opacity=0.3]
                                        (A-\r-1-W) arc (180:270:4pt)
                                            -- (A-\r-5-S) arc (-90:90:4pt)
                                            -- (A-\r-1-N) arc (90:180:4pt) -- cycle;
                            }
                            \foreach \c/\r in {1/1, 2/10, 3/7, 4/8, 5/4} {%
                                \draw
                                    [red!70!black!90!white, dotted, line width=1.75pt, stealth-]
                                    ($(A-\r-0-W) + (-5pt,0)$)
                                        node [anchor=east]
                                        {\footnotesize$\mathsf p(\textsf\c)$\!\!\!} -- (A-\r-\c);
                            }
                            \foreach \r/\c in {%
                                1/1, 4/5, 7/3, 8/4, 10/2,
                                3/1,3/2, 3/4, 3/5,
                                5/1, 5/2,
                                9/4, 9/5,
                                11/1, 11/2, 11/4,
                                12/4%
                                } {
                                \filldraw [black] (A-\r-\c) circle (2.125pt);
                            }
                        \end{tikzpicture}
                        \end{aligned}
                    \;\right]}
                    _{\text{\normalsize $A^{\mathclap{\phantom|}}$}}
                \end{color}
                \!\!
                \left[\:\!
                    \mathclap{\begin{matrix} \\[10ex] \end{matrix}}
                    \begin{matrix}
                        x_{\mathsf 1} \\[0.5ex] x_{\mathsf 2} \\[0.5ex] \vdots \\[0.5ex] x_{\mathsf r}
                    \end{matrix}
                \:\!\right]
                \!\!\:\oplus\!\!\:
                \begin{color}{orange!80!black!90!white}\underbrace{\color{black}
                    \left[\;
                        \mathclap{\begin{matrix}  \\[20.5ex]  \end{matrix}}
                        \begin{aligned}
                        \begin{tikzpicture}
                            \foreach \r in {1,...,12} {%
                                \coordinate (b-\r) at (0, {-\r*.4});
                                \foreach \d/\a in {%
                                    E/0, NE/45, N/90, NW/135, W/180, SW/225, S/270, SE/315} {
                                        \coordinate (b-\r-\d) at ($(b-\r) + (\a:4.5pt)$);
                                }
                            }
                            \filldraw [draw=black, fill=orange!80!black!60!white, opacity=0.6]
                                (b-1-N) arc (90:180:4.5pt)
                                    -- (b-12-W) arc (-180:0:4.5pt)
                                    -- (b-1-E) arc (0:90:4.5pt) -- cycle;
                            \foreach \r in {1,5,7,9,10,11} {
                                \filldraw [black] (b-\r) circle (2.125pt);
                            }
                        \end{tikzpicture}
                        \end{aligned}
                    \;\right]}
                    _{\text{\normalsize$\vec b^{\mathclap{\phantom|}}$}}
                \end{color}                
                $}
                \scalebox{2}{$\Biggr\rangle$}
            }
        \mspace{-18mu}
    \end{gather*}
    \end{huge}
    \vspace*{-2ex}
    
    \caption{%
        \label{fig:QFE-representation-schematic}%
        A schematic representation of some quadratic form expansion, in which the expansion matrix $A$ and Gram matrix $Q$ are stored with sparse data structures, and where $A$ is in principal row form with some index map $p$.
        Black dots represent the locations of non-zero entries of $Q$, $A$, and $\vec b$.
        Each row/column of $A$ and of $Q$ are represented in a way which stores only the non-zero entries of each row and column (and all of the diagonal elements of $Q$).
        Some rows of $A$ may be zero, in which case the corresponding qubit is in some fixed standard basis state $\ket{b_j}$; in contrast, each column of $A$ is non-zero.
        In particular, each column $1 \le c \le r$ of $A$ has a designated `principal row' $j = p(c)$ which stores the row-vector $\vec e_c\trans$, and is therefore non-zero \emph{only} in column $c$.
        (While we depict these rows differently above for the sake of clarity, the principal rows would be stored in the same way as the other rows.)
        The management of the `principal index map' $p$ is central to our results.
        In some cases, the choice of principal row may be arbitrary: in the example above, one could equally well select $p(4)$ to be the final row of $A$ instead of the row indicated.
        In other cases, there is only one possible choice of principal row for some column $c$, as when there is only one row which stores $\vec e_c\trans$ (as with $c=1$ above); sometimes this row is also the only row in which column $c$ itself is non-zero (as with $c=3$ above).
        In the course of simulating operations on a quadratic form expansion, it may be necessary to transform a row which is designated as a principal row $j = p(c)$ for some column $c$.
        When this happens, we may attempt to select a new row to act as the principal row for column $c$; in some cases this may require a transformation $A \mapsto A'$ through column operations, to produce a row $j'$ of $A'$ which is equal to $\vec e_c\trans$.
    }
\end{figure}

Our results pre-suppose a machine with random access memory (RAM), where memory accesses, comparisons of bit-values, and comparisons of integers can all be performed in constant time.
This model accurately reflects the complexity of the comparison and arithmetic operations which are performed on much of the data, which may be represented by integers which are bounded above by a constant.
While our model neglects theoretical $\mathrm{poly}\;\! \log(n)$ factors for memory accesses and arithmetic on integers in general, these poly-log factors may in practise also be bounded by constants (\emph{e.g.},~when simulating circuits on $n < 2^{64}$ qubits), and would in any case be swamped by overheads arising from realising these operations on realistic computer architecture (\emph{e.g.},~cache misses and memory word size).

For the sake of brevity, our procedures do not act on the tuple $\mathcal E = (n,r,g,Q,A,\vec b,p)$ explicitly.
Instead, we suppose that $\mathcal E$ is taken implictly an argument to any procedure to modify the quadratic form expansion which it represents, and may be modified as a side effect.
In particular, our results concern the complexity of simulating operations on a quadratic form expansion \emph{in place}, which is in effect to say operating on the quadratic form expansion without making a copy of any of its data.
(In particular, it should be considered to be passed by reference rather than by value.)
As a consequence, the original values stored in any data structure are not available after modification, unless it has been explicitly copied elsewhere.

\subsection{Mixed modulus arithmetic}
\label{sec:mixed-modulus-arithmetic}

To maintain the expansion matrix $A$ in principal row form, it will occasionally be necessary to perform column operations on $A$, of a sort that correspond to a change in variables of the index $\vec x$.
Such a change of variables will involve a corresponding transformation of the Gram matrix $Q$ --- somehow taking into account the fact that while $A\vec x \oplus \vec b$ may be evaluated modulo~2, we cannot treat $Q$ as merely being defined modulo~2 (as this would obscure the difference between relative phases of $i$ and $-i$).

First, we note that while $Q$ cannot be reduced modulo~2, the majority of its coefficients \emph{can} be reduced modulo~2, precisely because $Q$ is also symmetric:
\begin{lemma}
    \label{lemma:Gram-matrix-mixed-moduli}
    Let $Q, Q'$ be symmetric $r \times r$ matrices over $\mathbb Z$.
    If $Q' - Q \equiv \Delta \pmod{4}$, for a matrix $\Delta$ with only even coefficients and which is zero on its diagonal, then $\vec x\trans Q' \;\!\vec x \equiv \vec x \trans Q \;\!\vec x \pmod{4}$; and conversely.
\end{lemma}

\noindent
We prove this (easy) result in the Appendices (Lemma~\ref{lemma:congruence-of-Gram-matrices-mod-4}, p.~\pageref{lemma:congruence-of-Gram-matrices-mod-4}).
Then, in an imaginary exponent, we may evaluate the matrix Gram matrix $Q$ modulo~$4$, but in particular also reduce any \emph{off-diagonal} coefficient of a Gram matrix in an imaginary exponent mod~$2$.
This fact may be used to reduce the number of non-zero coefficients in sparse matrix representations of $Q$.

It will also prove useful to occasionally perform a change of variables on the vector $\vec x$, corresponding to some particular transformation of the expansion matrix $A$.
This is important enough to warrant an explicit result concerning such changes of variables, requiring careful treatment of the mixed-modulus arithmetic:
\begin{lemma}
    \label{lemma:QFE-change-of-variables}
    Let $0 \le r \le n$ be integers, $A$ an $n \times r$ integer matrix, $Q$ a symmetric $r \times r$ integer matrix, $\vec b \in \{0,1\}^n$, and $g \in \mathbb Z$.
    Let $E$ be a unimodular integer matrix,%
        \footnote{%
            An integer matrix $E$ is \emph{unimodular} if $E^{\;\!-1}$ exists and is also an integer matrix; examples include permutation matrices, and either upper or lower triangular matrices in which every coefficient on the main diagonal is $1$.
        }
    and $A' \equiv AE \pmod{2}$ and $Q' \equiv {E\!\:}\trans \!\!\: Q E \pmod{4}$.
    Then
    \vspace*{-.25ex}
    \begin{equation}
    \label{eqn:QFE-change-of-variables}
        \frac{\tau^g}{\sqrt{2^r}}
            \!
            \sum_{\mathbf x \in \{0,1\}^r} 
                \!\!\!
                i^{\;\!\vec x\trans \! Q \vec x}
                \,
                \ket{A \mathbf{x} \oplus \mathbf{b}}
    \;=\;
        \frac{\tau^g}{\sqrt{2^r}}
            \!
            \sum_{\mathbf x \in \{0,1\}^r} 
    \!\!\!
                i^{\;\!\vec x\trans \! Q' \vec x}
                \,
                \ket{A' \mathbf{x} \oplus \mathbf{b}} \;,
    \end{equation}\\[-4ex]
    
    \noindent
    where in particular the index of summation on each side of the equation ranges over $\{0,1\}^r \subseteq \mathbb Z^r$.
\end{lemma}

\begin{proof}
    This result rests on some number-theoretic results which we defer to the Appendix.
    Apart from this, we may prove the result as follows.
    If we let $C = E^{\!\:-1}$, we then have
    \begin{align}{}
        \frac{\tau^g}{\sqrt{2^r}}
            \!
            \sum_{\mathbf x \in \{0,1\}^r} 
                \!\!\!
                i^{\;\!\vec x\trans \! Q \vec x}
                \,
                \ket{A \mathbf{x} \oplus \mathbf{b}}
    & \;=\;
        \frac{\tau^g}{\sqrt{2^r}}
            \!
            \sum_{\mathbf x \in \{0,1\}^r} 
                \!\!\!
                i^{\;\!\vec x\trans \! C\trans\! E\trans \! Q E C \vec x}
                \,
                \ket{AEC \mathbf{x} \oplus \mathbf{b}} \nonumber \\
    & \;=\;
        \frac{\tau^g}{\sqrt{2^r}}
            \!\!\!\!
            \sum_{\substack{\mathbf y = C \vec x \\ \mathbf x \in \{0,1\}^r}} 
                \!\!\!
                i^{\;\!\vec y\trans \! Q' \vec y}
                \,
                \ket{A' \mathbf{y} \oplus \mathbf{b}} .
    \end{align}~\\[-4ex]

    \noindent
    For $\vec y = C \vec x \in \mathbb Z^r$, let $\vec{\tilde y}$ be its residue modulo~2.
    While $\vec y$ is only equivalent to $\vec{\tilde y}$ modulo~2, this is in fact enough (Lemma~\ref{lemma:consistency-mod-2}, p.~\pageref{lemma:consistency-mod-2}) to show that $\vec{\tilde y}\trans\!\!\; Q' \, \vec{\tilde y} \equiv \vec y\trans \!\!\;Q' \, \vec y \pmod{4}$.
    Furthermore, as $E$ and $C$ are invertible integer matrices, they are invertible modulo~2 as well, so that $\bigl\{ \vec{\tilde y} \in \{0,1\}^r \,\big\vert\, \exists \vec x \in \{0,1\}^r : \vec{\tilde y} \equiv C  \vec x  \pmod{2} \bigr\}$ is equal to the set $\{0,1\}^r$ itself.
    Thus, we have
    \vspace*{-1ex}
    \begin{equation}{}
        \frac{\tau^g}{\sqrt{2^r}}
            \!\!\!\!
            \sum_{\substack{\mathbf y = C \vec x \\ \mathbf x \in \{0,1\}^r}} 
                \!\!\!
                i^{\;\!\vec y\trans \! Q' \vec y}
                \,
                \ket{A' \mathbf{y} \oplus \mathbf{b}}
    \;=\;
        \frac{\tau^g}{\sqrt{2^r}}
            \!\!
            \sum_{\vec{\tilde y} \in \{0,1\}^r} 
                \!\!\!
                i^{\;\!\vec{\tilde y}\trans \! Q' \vec{\tilde  y}}
                \,
                \ket{A' \mathbf{\tilde y} \oplus \mathbf{b}} ;
    \end{equation}
    the Lemma then holds, by
    relabeling the index of summation on the right-hand side from $\vec{\tilde y}$ to $\vec x$.
\end{proof}
\noindent
These results form the basis of simple subroutines, described in Section~\ref{sec:QFE-subroutines} and defined in Appendix~\ref{apx:subroutines}, to allow us to work more easily with the two forms of modular arithmetic involved in quadratic form expansions.

\subsection{Fixing index bit values}
\label{sec:fixBits}

On occasion, we may wish to assign a fixed value to some particular bit of $\vec x \in \{0,1\}^r$ in a quadratic form expansion.
For instance, in the case of a $Z$-basis measurement, we may wish simply to select for only those terms in which a particular value $z \in \{0,1\}$ for some bit $x_k$ is realised.%
    \footnote{%
        Indeed, for a quadratic form expansion as in Eqn.~\eqref{eqn:QFE-specific} where $A$ is in principal row form, fixing the value of any particular bit $x_k$ for any $1 \le k \le r$, corresponds to computing $[\bra{z}_{p(k)} \otimes I]\ket{\psi}$, up to a normalising factor.
    }
More subtly, when simulating a Hadamard operation, it may become necessary to isolate those terms for which $\langle A\vec x \oplus \vec b\;\!\vert\:\! \psi\rangle \ne 0$ by fixing the value of some bit $x_k$.

For convenience, we suppose that it is the final bit $x_r$ which we wish to fix (to fix any other bit, we may first perform a simple change of variables as described in Lemma~\ref{lemma:QFE-change-of-variables}).
Given a state $\ket{\psi}$ as in Eqn.~\eqref{eqn:QFE-specific}, consider the (possibly subnormalised) vector $\ket{\psi'}$ obtained by selecting only those terms such that $x_r = z$ for some constant $z \in \{0,1\}$\,:
\begin{equation}{}
    \label{eqn:bitFixedQFE}
            \ket{\psi'}
        \;=\;
            \frac{\tau^g}{\sqrt{2^r}} \!\!
            \sum_{\vec x \in \{0,1\}^r} \!\!\!\!
                    i^{\,\vec x\trans \! Q \vec x}\,\delta_{x_r,z}
                \ket{A \vec x \oplus \vec b}.
\end{equation}
To see how this expression for $\ket{\psi'}$ may be simplified,
we may consider a block structure for $Q$,
\begin{equation}
            Q
        \;=\;
            \left[\begin{array}{c|c}
            ~\\
                \mspace{12mu}
                \tilde Q
                \mspace{12mu}
            &
                \vec q
            \\[2ex]
            \hline &
            \\[-2ex] 
                \vec q\trans 
            &
                u
            \end{array}\right].
\end{equation}
Let $\vec y = {[\,x_1\;\;\cdots\;\;x_{r{-}1}\,]\:\!\trans} \in \{0,1\}^{r{-}1}$: setting $x_r = z \in \{0,1\}$, we then have
\begin{equation}
            \vec x\trans\! Q \vec x
        \;=\;
            \vec y\trans \tilde Q \!\: \vec y
            + 2 \:\! z \;\! \vec q\trans \! \vec y + z^2 u
        \;=\;
            \vec y\trans \bigl[ \tilde Q + 2z \,\mathrm{diag}(\vec q\trans) \bigr] \vec y
            + z u \,.
\end{equation}
If we let $Q' = \tilde Q + 2z \,\mathrm{diag}(\vec q\trans)$, let $\vec a \in \{0,1\}^n$ be the final column of $A$, let $A'$ be the matrix consisting of the first $r-1$ columns of $A$ (omitting the final column, $\vec a$), and let $\vec b' = \vec b \oplus z \vec a$, we have
\begin{equation}{}
    \mspace{-18mu}
    \begin{aligned}[b]
            \ket{\psi'}
        \;&=\;
            \frac{\tau^g}{\sqrt{2^r}} \!\!\!
            \sum_{\vec y \in \{0,1\}^{r-1}} \!\!\!\!
                    i^{\,\vec y\trans [ \tilde Q + 2 z \;\! \mathrm{diag}(\vec q\trans) ] \!\: \vec y  \,+\, zu}
                \ket{A' \vec y \oplus z\:\!\vec a \oplus \vec b} \\
                \label{eqn:bitFixedQFE-reduced}
        \;&=\;
            \frac{1}{\sqrt 2} \cdot \biggl[ \frac{\tau^{g'}\!}{\sqrt{2^{r{-}1}}} \!\!\!\!
            \sum_{\vec y \in \{0,1\}^{r\!\!\:-\!\!\:1}} \!\!\!\!\!\!
                    i^{\,\vec y\trans \! Q'  \vec y}
                \ket{A' \vec y \oplus \vec b'} \biggr],
    \end{aligned}
    \mspace{-30mu}
\end{equation}
where $g' := g + 2 \!\: z \!\: u$.
Note that for any $j = p(k)$ for which $1 \le k < r$, row $j$ of $A'$ is a truncated version of row $j$ of $A$, omitting only the final (zero) coefficient, so that $\vec e_{p(k)}\trans A' = \vec e_k\trans$.
Then $A'$ is also in principal row form, with index map also given by $p$ (albeit restricted to inputs $0 \le k \le r{-}1$).

Thus, to simulate fixing $x_r = z$, and in particular to reduce the number of columns involved in $A$, it suffices to compute $A'$, $\vec b'$, $Q'$, and $g'$ as above.
Note that, due to our convention of treating normalisation as being closely connected to the rank, we do not have a general way of representing or otherwise accounting for the additional factor of $\smash{\tfrac{1}{\sqrt 2}}$ on the right-hand side of Eqn.~\eqref{eqn:bitFixedQFE-reduced}.
Note also that the number $r{-}1$ of columns of the matrix $A'$, and thus the dimension of the summation index, may now differ from the integer $r$ stored as the rank.
Any simulation technique which relies on the above analysis, must independently account for these details to ensure that the result is a quadratic form expansion describing a normalised state.

\subsection{Eliminating columns which are entirely zero}
\label{sec:ZeroColumnElim}

In transforming quadratic form expansions, we may temporarily produce an expansion as in Eqn.~\eqref{eqn:bitFixedQFE} in which the matrix $A$ is not full rank.
When this occurs in our analysis --- particularly, in describing the simulation of Hadamard operations and measurements in the $X$- and $Y$\!\!\:-\!\:eigenbases --- one of the columns of $A$ is in fact entirely zero.
This affords us the opportunity to reduce the number of columns of $A$ by one, or possibly even by two, as we describe here.

We depart slightly from the usual assumption that the state $\ket{\psi}$ is represented as in Eqn.~\eqref{eqn:QFE-specific}.
Suppose that $A$ has rank $r$, but has shape $n \times (r{+}1)$ with some column $c$ which is entirely zero; and that
\begin{equation}
        \ket{\psi}
    \;=\;
        \frac{\tau^g}{\sqrt{2^{r+1}}}
        \!\!
        \sum_{\mathbf x \in \{0,1\}^{r\!\!\:+\!\!\:1}} \!\!\!\!
            i^{\,\vec x \trans \! Q \!\: \vec x} \,
            \ket{A \vec x \oplus \vec b} .
\end{equation}
For our simulation techniques, this only occurs when $A$ is \emph{almost} in principal row form, in that it is equipped with a function which is nearly a principal index map $p: \{1,\ldots,r{+}1\} \to \{1,\ldots,n\}$, such that $\vec e_{p(k)}\trans A = \vec e_k\trans$ for all columns $1 \le k \le r\!+\!1$ such that $k \ne c$.
We may simplify our analysis by performing a change of variables, by defining $A^{(1)} = A E$ and $Q^{(1)} = E\trans \!Q E$, for the permutation matrix $E = I \:\!\oplus\:\! {(\vec e_c \:\!{\oplus}\;\! \vec e_{r{+}1})(\vec e_c \:\!{\oplus}\;\! \vec e_{r{+}1})\trans}\!$ which serves to swap columns $c$ and $r{+}1$ of $A$.
If we let $\vec b^{(1)} = \vec b$ to keep the superscript indices in lock-step, then by Lemma~\ref{lemma:QFE-change-of-variables} we have
\begin{equation}
    \label{eqn:zeroColumnElim-Q1} 
    \begin{aligned}[b]
            \ket{\psi}
        \;=\;
            \frac{\tau^g}{\sqrt{2^{r+1}}}
            \!\!
            \sum_{\mathbf z \in \{0,1\}^{r\!\!\:+\!\!\:1}} \!\!\!\!
                i^{\,\vec z \trans \! Q^{(1)} \!\: \vec z} \,
                \ket{A^{(1)} \vec z \oplus \mathbf{b}^{(1)}} .
    \end{aligned}
\end{equation}
If we define $p': \{1,\ldots,r\} \to \{1,\ldots,n\}$ so that $p'(c) = p(r{+}1)$ and $p'(k) = p(k)$ for $k \ne c$, then this is again almost a principal index map for $A^{(1)}$, except in that there is no row in $A^{(1)}$ which is equal to $\vec e_{r{+}1}\trans$ (or indeed which is even non-zero in column $r{+}1$).

Let $A^{(2)}$ be the matrix consisting of the first $r$ columns of $A^{(1)}$ (omitting the final zero column): then $A^{(2)}$ is in principal index form, as it has shape $n \times r$ and as $\vec e_{p'(c)}\trans A^{(2)} = \vec e_c\trans$ for each $1 \le c \le r$.
We may isolate the sum over the index $x_{r\!\!\:+\!\!\:1}$ in Eqn.~\eqref{eqn:zeroColumnElim-Q1} into a scalar factor, as follows.
As $\smash{Q^{(1)}}$ is symmetric, we can define an $r \times r$ symmetric matrix $Q^{(2)}$, a vector $\vec q \in \mathbb Z^r$, and integer $u$ so that
\begin{equation}
\label{eqref:eq10}
        Q^{(1)}
    \;=\;
        \left[\begin{array}{c|c}
        ~\\
            \mspace{12mu}
            Q^{(2)}
            \mspace{12mu}
        &
            \vec q
        \\[2ex]
        \hline &
        \\[-2ex] 
            \vec q\trans 
        &
            u
        \end{array}\right].
\end{equation}
Let $\vec x = {[\,z_1\;\;\cdots\;\;z_r\,]\:\!\trans} \in \{0,1\}^r$ and $y = z_{r{+}1} \in \{0,1\}$.
We may then re-express the exponent of the phase in Eqn.~\eqref{eqn:zeroColumnElim-Q1} as
\begin{equation}{}
    \begin{aligned}[b]
        \vec z\trans Q^{(1)} \;\! \vec z
    \;&=\;
        \vec x\trans 
            Q^{(2)} \!\:
        \vec x
        \,+\,
            2 \:\! y \;\! \vec q\trans \!\!\: \vec x
        \,+\,
            y^2 \:\! u
    \;=\;
        \vec x\trans Q^{(2)} \;\!\vec x
        \,+\,
            y (2 \!\:\vec q\trans \!\!\: \vec x + u) \,.
    \end{aligned}
\end{equation}
Then, if we let $\vec b^{(2)} = \vec b^{(1)}$, we then have
\begin{equation}
    \label{eqn:zeroColumnElim-Q2} 
    \begin{aligned}[b]
            \ket{\psi}
        \;&=\;
            \frac{\tau^g}{\sqrt{2^{r+1}}}
            \!\!
            \sum_{\substack{\mathbf x \in \{0,1\}^r  \\ y \in \{0,1\}}}
            \!\!\!
                    i^{\,\vec x\trans \! Q^{(2)} \!\: \vec x \,+\, y (2 \vec q\trans\! \vec x \!\:+\!\: u)} \,
                \ket{A^{(2)} \mathbf{x} \oplus \mathbf{b}^{(2)}}
        \\&=\;
            \frac{\tau^g}{\sqrt{2^{r+1}}}
            \!\!
            \sum_{\mathbf x \in \{0,1\}^r}
            \!\!\!
                    i^{\,\vec x\trans \! Q^{(2)} \vec x}
                    \Bigl(
                        1 + i^{\;\!2  \!\: \vec q\trans\! \vec x \!\:+\!\: u}
                    \Bigr)
                \ket{A^{(2)} \vec x \oplus \vec b^{(2)}}
    \;.
    \end{aligned}
\end{equation}
To simplify the parenthesised sum, we must consider two cases: one for $u \in \{1,3\}$, and one for $u \in \{0,2\}$.

\begin{itemize}
\item
    Suppose $u = 2d+1$ for $d \in \{0,1\}$.
    Then for various values of $\vec x$, the scalar expression in parentheses in Eqn.~\eqref{eqn:zeroColumnElim-Q2} has the form $1 \pm i \,=\, \sqrt{2} \;\! \tau^{\pm 1}$.
    Specifically, we have
    \begin{equation}{}
    \mspace{-24mu}
    \begin{aligned}[b]
            \biggl(
                \frac{1 + i^{\;\!2\!\: \vec q\trans\! \vec x \,+\, u}}{\sqrt 2}
            \biggr)
        \,&=\,
            \biggl(
                \frac{1 + (-1)^{\vec q\trans\! \vec x + d} \:\!\cdot\:\! i}{\sqrt 2}
            \biggr)
        \;=\;
            \tau^{(-1)^{\vec q\!\!\; \trans\!\!\!\; \vec x \!\:\oplus\!\: d}}
          =\,
            \tau^{1 \;\!-\;\! 2(\vec q\trans\! \vec x \,\oplus\, d)}
        \,=\,
            \tau \cdot i^{-(\vec q\trans\! \vec x \,\oplus\, d)}
            ,
    \end{aligned}
    \mspace{-12mu}
    \end{equation}
    where recall that $\vec q\trans \vec x$ is in principle an integer, but that $\vec q\trans \vec x \oplus d$ is the reduction of $\vec q\trans \vec x + d$ modulo~2.
    We may simplify this further by considering how to represent $\vec q\trans \vec x \oplus \vec d$, as an expression modulo~4.
    To do so, we define a binary operation `$\ast$' for mod~2 matrix multiplication of integer matrices $A$ and $B$,\footnote{%
            \label{fn:mod2opns}%
            It is not difficult to show that this operator is associative, even over the integers.
        }
where $A \ast B$ is the matrix that results when the coefficients of the matrix product $AB$ are projected to $\{0,1\}$:
\begin{equation}
    \label{eqn:ast-operation}
    A \!\!\:\ast\! B = \bigl[ \:\! c_{j,k} \bigr]\,,
    \qquad
    c_{j,k} = \begin{cases} 0, & \text{if $\vec e_j\trans\!\; (AB)\, \vec e_k$ is odd;} \\[.5ex]
                        1, & \text{if $\vec e_j\trans\!\; (AB)\, \vec e_k$ is even.}
          \end{cases}
\end{equation}
This allows us to explicitly denote when a matrix multiplication is to be reduced modulo~$2$, in a context where other arithmetic is \emph{not} being performed modulo~$2$. 
For instance, we have $\vec q\trans \!\!\!\;\ast \vec x \in \{0,1\}$.
We may then expand this into a matrix expression modulo~4, using the fact that $y^2 \equiv 0 \pmod{4}$ for $y$ even, and $y^2 \equiv 1 \pmod{4}$ for $y$ odd:
\begin{equation}
\label{eqn:evalAstMod4}
\begin{aligned}[b]
        \vec q\trans \!\ast \vec x
    \;&\equiv\;
        (\vec q \trans \vec x)^2 \pmod{4}
    \\[1ex]&=\;
        \vec x\trans \!\!\: \vec q \!\; \vec q\trans \!\: \vec x\,.
\end{aligned}
\end{equation}
Then, also using the fact that we have
\begin{equation}
    \label{eqn:parity-formula-01}
    u \oplus v = u + v - 2uv
\end{equation}
for $u,v \in \{0,1\}$, we have
\begin{equation}
        \vec q \trans \! \vec x \oplus d
    \;=\;
        d + (\vec q\trans \!\!\ast\!\!\: \vec x) - 2d (\vec q\trans \!\!\ast\!\!\;  \vec x)
    \;\equiv\;
        d + (1-2d)(\vec x\trans \!\!\: \vec q \;\! \vec q\trans\!\!\: \vec x)
        \pmod{4}.
\end{equation}
Using this, we then have
    \begin{equation}{}
    \mspace{-24mu}
    \begin{aligned}[b]
            \biggl(
                \frac{1 + i^{\;\!2\!\: \vec q\trans\! \vec x \,+\, u}}{\sqrt 2}
            \biggr)
        \,&=\,
            \tau \cdot i^{-d - (1-2d)(\vec x\trans\! \vec q \;\! \vec q\trans\! \vec x)}
        \,=\,
            \tau^{-u+2} \, i^{\;\!(u-2)(\vec x\trans\! \vec q \;\! \vec q\trans\! \vec x)}
            .
    \end{aligned}
    \mspace{-12mu}
    \end{equation}
    Then we may rewrite Eqn.~\eqref{eqn:zeroColumnElim-Q2} to obtain
    \begin{equation}{}
    \label{eqn:zeroColumnElim-Q3a}
    \mspace{-18mu}
    \begin{aligned}[b]
            \ket{\psi}
        \;&=\;
            \frac{\tau^g}{\sqrt{2^{r}}}
            \!\!
            \sum_{\mathbf x \in \{0,1\}^r}
            \!\!\!
                i^{\,\vec x\trans \! Q^{(2)} \vec x}
                \biggl(
                    \frac{1 + i^{\;\!2\!\: \vec u\trans\! \vec x \,+\, q}}{\sqrt 2}
                \biggr)
                \ket{A^{(2)} \mathbf{x} \oplus \mathbf{b}^{(2)}}
        \\[1ex]&=\;
            \frac{\tau^{g-u+2}}{\sqrt{2^{r}}}
            \!\!
            \sum_{\mathbf x \in \{0,1\}^r}
            \!\!\!
                    i^{\,\vec x\trans \! [Q^{(2)} + (u-2)\vec q \vec q\trans] \vec x}
                \ket{A^{(2)} \mathbf{x} \oplus \mathbf{b}^{(2)}}
        \;=\;
            \frac{\tau^{g'}}{\sqrt{2^{r}}}
            \!\!
            \sum_{\mathbf z \in \{0,1\}^r}
            \!\!\!
                    i^{\,\vec z\trans \! Q' \vec z}
                \,
                \ket{A' \mathbf{z} \oplus \mathbf{b}'} ,
    \end{aligned}
    \mspace{-9mu}
    \end{equation}
    where 
$
        Q'
        \,:=\,
            Q^{(2)} + (u-2) \vec q \:\!\vec q\trans
            ,
 $ 
    and $g' := g - u + 2$, with $A' := A^{(2)}$ a matrix in principal row form with index map $p'$.
    
\item
    Suppose $u = 2d$ for $d \in \{0,1\}$.
    Then for various values of $\vec x$, the scalar expression in parentheses in Eqn.~\eqref{eqn:zeroColumnElim-Q2} is either $0$ or $2$, depending on whether $\vec x$ is orthogonal (mod $2$) to $\vec q$.
    Because of this, it is of particular interest whether $\vec q \ne \vec 0$: we may easily show that this is the case for a simulation of stabiliser operations on normalised states.
    \begin{itemize}
    \item 
        If $\vec q = \vec 0$ and $u = 2$, it would follow that in fact  $\ket{\psi} = \vec 0$.
        This vector cannot be represented under the conditions of the expansion matrix $A$ being in principal row form.
    \item
        For $\vec q = \vec 0$ and $u = 0$, $\ket{\psi}$ would instead be super-normalised but otherwise adequately represented by the data $Q^{(2)}$, $A^{(2)}$, and $\vec b^{(2)}$.
        Note that we could in principle perform further operations to maintain the representation with $A$ in principal row form, though the super-normalisation would not be faithfully represented by the value of the rank $r$.
    \end{itemize}
    As our interest is in the case that the state $\ket{\psi}$ represented at the input is normalised, neither of these cases should arise in the course of a stabiliser circuit simulation.%
        \footnote{%
            In practice, for the stabiliser simulation subroutines which we describe later, we can prove in particular that $q_c = 1$; though we do not need to use this particular fact for our simulation procedures.
        }
    We may suppose that a well-defined subroutine simply stops (possibly setting some warning flag) if it discovers that $\vec q = \vec 0$, which it can test in time $\mathcal{O}(1)$ using the sparse data structure for $Q$.
    
    Given that $\ket{\psi}$ is a normalised state,%
    \label{discn:normalised-zero-column-elim}
    then we have $\vec q \ne \vec 0$: using the sparse data structure of $Q^{(1)}$, we may find an index $1 \le \ell \le r$ for which $u_\ell = \smash{Q^{(1)}_{\ell,r\!\!\:{+}\!\!\:1}} = 1$ in $\mathcal{O}(1)$ time.
    By performing appropriate column operations to the row-vector $\vec q\trans$, we may reduce it to $\vec e_\ell\trans$; a further column-swap would allow us to map this row-vector to $\vec e_r\trans$.
    That is to say, if we let
    $R' \,=\, I \!\;-\!\; \vec e_\ell(\vec q - \vec e_\ell)\trans\!$ and $R'' \,=\, I \!\;\oplus\!\; (\vec e_\ell \oplus \vec e_r)(\vec e_\ell \oplus \vec e_r)\trans\!$, we have $\vec q\trans \:\! R' \:\! R'' = \vec e_r\trans$.
    As both $R'$ and $R''$ are invertible as integer matrices, so is $R := R' R''$.
    (Note that in fact $R'' = I$ in the case $\ell = r$.)
    Continuing from Eqn.~\eqref{eqn:zeroColumnElim-Q2}, we have
    \begin{equation}
    \begin{aligned}[b]
            \ket{\psi}
        \;&=\;
            \frac{\tau^g}{\sqrt{2^{r+1}}}
            \!\!
            \sum_{\mathbf x \in \{0,1\}^r}
            \!\!\!
                    i^{\,\vec x\trans \! Q^{(2)} \vec x}
                    \Bigl(
                        1 + (-1)^{\vec q\trans\! \vec x \!\:+\!\: d}
                    \Bigr)
                \ket{A^{(2)} \vec x \oplus \vec b^{(2)}}
        \\&=\;
                \frac{\tau^g}{\sqrt{2^{r+1}}}
            \!\!
            \sum_{\mathbf x \in \{0,1\}^r}
            \!\!\!
                    i^{\,\vec x\trans \! Q^{(2)} \vec x}
                    \Bigl(
                        1 + (-1)^{\vec e_r\trans R^{-1} \vec x \,+\, d}
                    \Bigr)
                \ket{A^{(2)} \mathbf{x} \oplus \mathbf{b}^{(2)}}        
        \\&=\;
                \frac{\tau^g}{\sqrt{2^{r+1}}}
            \!\!
            \sum_{\mathbf v \in \{0,1\}^r}
            \!\!\!
                    i^{\,\vec v\trans \! R\trans\! Q^{(2)} R \;\! \vec v}
                    \Bigl(
                        1 + (-1)^{\vec e_r\trans \vec v \,+\, d}
                    \Bigr)
                \ket{A^{(2)} R \;\! \mathbf{v} \oplus \mathbf{b}^{(2)}},   
        \end{aligned}
    \end{equation}
    where we use Lemma~\ref{lemma:QFE-change-of-variables} for the last equality.\footnote{%
        To use Lemma~\ref{lemma:QFE-change-of-variables} in this case, we implicitly perform the change of variables over the summation index $\vec z \in \{0,1\}^{r+1}$ used in Eqn.~\eqref{eqn:zeroColumnElim-Q2}, to allow the change of variables also to incorporate the parenthesised expression in the sum.
    } 
    Note that the $\{0,1\}$-matrix $A^{(3)} \equiv A^{(2)} R \pmod{2}$ may again not be in principal row form: while the rows $p'(k)$ of $A^{(2)}$ will be unaffected by right-multiplication by $R$ for $k \ne \ell$, we have  $\vec e_{p'(\ell)}\trans A^{(2)} R = \vec e_\ell\trans R' R'' \equiv \vec q\trans R''$, where $\vec q\trans$ may have more than one non-zero coefficient which are merely permuted by the action of $R''$.
    However, if we define the map%
    \vspace*{-2ex}

    \begin{equation}
        p''(k)   \;=\;
        \begin{cases}
            p'(r), 
            &  \text{if $k = \ell$},
        \\[.5ex]
            p'(k),       &   \text{otherwise},
        \end{cases}
    \end{equation}
    then $\smash{\vec e_{p''(\ell)}\trans A^{(3)}} = \smash{\vec e_{p'(r)} A^{(2)} R' R''} = \smash{\vec e_r\trans R' R''} = \smash{\vec e_r\trans R''} = \vec e_\ell\trans$ when $\ell < r$.
    In this sense, $A^{(3)}$ is nearly in principal row form (with index map $p''$), except that it may lack a row with the vector $\vec e_r\trans$.
    However, we will now see that the $r^{th}$ column of $A$ may be eliminated anyway.

    \medskip
    Note that $1 + (-1)^{v_r + d} = 2\;\!\delta_{v_r,d}$: then if we let $Q^{(3)} \equiv R\trans Q^{(2)} R \pmod{4}$ and $\vec b^{(3)} = \vec b^{(2)}$, we have  
    \begin{equation}
    \label{eqn:zeroColumnElim-Q3b}
    \begin{aligned}[b]
            \ket{\psi}
        \;&=\;
            \sqrt{2} \cdot \biggl[\frac{\tau^g}{\sqrt{2^r}}
            \!\!
            \sum_{\mathbf v \in \{0,1\}^r}
            \!\!\!
                    i^{\,\vec v\trans \!Q^{(3)} \vec v}
                    \,\delta_{v_r,d}\,
                \ket{A^{(3)} \mathbf{v} \oplus \mathbf{b}^{(3)}} \biggr].
        \end{aligned}
    \end{equation}
    We may simplify the expression in square brackets on the right-hand side, following the analysis of Section~\ref{sec:fixBits} for expressions of the form of Eqn.~\eqref{eqn:bitFixedQFE}.
    In particular, the factor of $\sqrt 2$ cancels against the factor of $\tfrac{1}{\sqrt 2}$ on the right-hand side of Eqn.~\eqref{eqn:bitFixedQFE-reduced}.
    Applying the appropriate transformations will yield
    \vspace*{-1ex}
    \begin{equation}{}
    \label{eqn:zeroColumnElim-Q4b}
    \mspace{-18mu}
    \begin{aligned}[b]
            \ket{\psi}
        \;&=\;
            \frac{\tau^{g'}\!}{\sqrt{2^{r{-}1}}} \!\!\!\!
            \sum_{\vec x \in \{0,1\}^{r\!\!\:-\!\!\:1}} \!\!\!\!\!\!
                    i^{\,\vec x\trans \! Q^{(4)}  \vec x}
                    \ket{A^{(4)} \vec x \oplus \vec b^{(4)}},
    \end{aligned}
    \mspace{-30mu}
    \end{equation}
    for suitably defined phase exponent $g'$, Gram matrix $Q^{(4)}$, and vector $\vec b^{(4)}$, and where $A^{(4)}$ differs from $A^{(3)}$ by omitting column $r$.
    In particular, $A^{(4)}$ has shape $n \times (r{-}1)$, so that the restriction of $p''$ to $\{1,\ldots,r{-}1\}$ is a principal index map for $A^{(4)}$; then $A' := A^{(4)}$ has rank $r-1$.
\end{itemize}
In each case above, we compute a suitable quadratic form expansion, with a corresponding matrix $A'$ in principal row form, and with an accompanying principal index map.

\subsection{Subroutines}
\label{sec:QFE-subroutines}

The preceding Sections each motivate simple subroutines, to assist in the transformation and maintenance of quadratic form expansions, in which the matrix $A$ is in principal row form.
We summarise these subroutines here, together with their run-times, for use in the simulation procedures for stabiliser circuits.
Below, $s \ge 0$ is an upper bound on the number of non-zero entries in any row of $A$; $t \ge 0$ is an upper bound on the number of non-zero entries in any column of $A$; $w \ge 0$ is an upper bound on the number of non-zero entries in any row/column of $Q$; and subscripted versions of these variables (such as $s_j$, $t_c$, \emph{etc.}) refer to the number of non-zero entries in specific rows or columns of these matrices. 
\begin{description}[font=\mdseries]
\item[%
    $\algname{ReduceGramRowCol}(c)$] ~\\
    \setalgdesc{ReduceGramRowCol}{%
        For an integer $1 \le c \le r$, reduce the coefficient $Q_{c,c}$ modulo~4, and reduce every off-diagonal entry of row and column $c$ modulo~2.}%
    Following Lemma~\ref{lemma:Gram-matrix-mixed-moduli} (Section~\ref{sec:mixed-modulus-arithmetic}):
    \algdesc{ReduceGramRowCol}
    This runs in time $\mathcal{O}(w_c) \subseteq \mathcal{O}(r)$.
\item[%
    $\algname{ReindexSubtColumn}(k,c)$] ~\\
    \setalgdesc{ReindexSubtColumn}{%
        For distinct integers $1 \le c,k \le r$, update a quadratic form expansion by performing a change of variables in which column $c$ of $A$ is subtracted (mod~2) from column $k$.}%
    Following Lemma~\ref{lemma:QFE-change-of-variables} (Section~\ref{sec:mixed-modulus-arithmetic}):
    \algdesc{ReindexSubtColumn}
    This runs in time $\mathcal{O}(t_c + w_c)  \subseteq \mathcal{O}(n)$.
\item[%
    $\algname{ReindexSwapColumns}(k,c)$] ~\\
    \setalgdesc{ReindexSwapColumns}{%
        For integers $1 \le c,k \le r$, update a quadratic form expansion by performing a change of variables corresponding to swapping columns $c$ and $k$ of $A$.}%
    \algdesc{ReindexSwapColumns}
    This runs in time $\mathcal{O}(t_c + t_k + w_c + w_k)  \subseteq \mathcal{O}(n)$.
\item[%
    $\algname{MakePrincipal}(c,j)$] ~\\
    \setalgdesc{MakePrincipal}{%
        For integers $1 \le c \le r$ and $1 \le j \le n$, if $A_{j,c} = 1$, perform appropriate changes of variable to transform row $j$ of $A$ to $\vec e_c\trans$, in order to make $j$ a principal row of $A$.}%
    \algdesc{MakePrincipal}
    If $s_j = 1$ or if $A_{j,c} = 0$, this runs in time $\mathcal{O}(1)$; otherwise this runs in time $\mathcal{O}(s_j t_c + s_j w_c)  \subseteq \mathcal{O}(nr)$.
    
\item[%
    $\algname{ReselectPrincipalRow}(j,c)$] ~\\
    \setalgdesc{ReselectPrincipalRow}{%
        For integers $1 \le c \le r$ and $0 \le j \le n$, attempt to find a row $j_\ast \ne j$ of $A$, such that $A_{j_\ast\!\:,c} \ne 0$, to serve as a new principal row.
        If no such row is found, the stop without modifying the quadratic form expansion.}%
    \algdesc{ReselectPrincipalRow}
    If $j > 0$ is the only row in which column $c$ is non-zero, this halts in time $\mathcal{O}(1)$ without modifying $A$.
    Otherwise it runs in time $\mathcal{O}(s_{j_\ast} t_c + s_{j_\ast} w_c)  \subseteq \mathcal{O}(nr)$, where $s_{j_\ast}$ is the minimum number of non-zero entries in a row $j_\ast \ne j$ for which $A_{j_\ast,c} = 1$.

\item[%
    $\algname{FixFinalBit}(z)$] ~\\
    \setalgdesc{FixFinalBit}{%
        Perform a transformation on the quadratic form expansion, consistent with fixing the value of $x_r = z$, reducing the rank in doing so.}%
    \algdesc{FixFinalBit}
    This runs in time $\mathcal{O}(t_r + w_r)  \subseteq \mathcal{O}(n)$.
    
\item[%
    $\algname{ZeroColumnElim}(c)$] ~\\
    \setalgdesc{ZeroColumnElim}{%
        Eliminate (one or two) redundant columns from the matrix $A$ of a quadratic form expansion, given that $A$ has $r{+}1$ columns but rank $r$, and that column $c$ in particular is entirely zero.
        (This will in general involve other non-trivial changes to $A$.)%
    }
    \algdesc{ZeroColumnElim}
    This runs in time $\mathcal{O}(tw + w^2) \subseteq \mathcal{O}(nr)$.
\end{description}
Despite the fact that these procedures do not take any explicit arguments which contain the data $\mathcal E = (n,r,g,Q,A,\vec b,p)$ representing the quadratic form expansion, each procedure described above should be understood to potentially modify some or all of those parameters.
The implementations of these subroutines are simple, and are described in pseudocode in Appendix~\ref{apx:subroutines}, along with a run-time analysis for each.

\section{Simulating Clifford operations}
\label{sec:simulCliffOpns}

In this Section, we describe how to simulate Clifford operations on quadratic form expansions, using the data structures and subroutines described in Section~\ref{sec:computing-w-QFEs}.
We describe their run-time complexity, in terms which take advantage of bounds on the number of non-zero coefficients in the rows and columns of the Gram matrix $Q$ and the expansion matrix $A$.
We show, independently of any sparsity conditions, that each operation may be simulated in time $\mathcal{O}(nr)$: a summary of the run-times of each procedure is provided in Section~\ref{runtimes}.
We also show that single-qubit $X$-, $Y$-, or $Z$-basis measurements with deterministic outcomes may be similated in $\mathcal{O}(1)$ time.

In the following, we let $0 \le s,w \le r$ be (respectively) upper bounds on the number of non-zero entries in any row of $A$ or row/column of $Q$, $0 \le t \le n-r+1$ be an upper bound\footnote{%
        Because there are $r$ principal rows, each of which is non-zero in exactly one column, each column of $A$ has at least $r-1$ zero coefficients in it when $r > 0$.
        (For $r = 0$, there are no non-zero coefficients in $A$ at all.)
        As a point of interest, we note that as $s,w \le r$, it follows that we may bound $st, wt \le \tfrac{1}{4}n^2 + \tfrac{1}{2}n + \tfrac{1}{4}$.
    }
on the number of non-zero entries in any column of $A$, and subscripted versions of these (such as $s_j$, $t_k$, \emph{etc.}) represent the number of non-zero entries in particular rows or columns.

\subsection{Pauli operations}
\label{sec:simulatePaulis}

We may easily represent the effect of Pauli operations on quadratic form expansions.
To represent an $X_j$ operation, we may simply flip the coefficient $b_j$ of $\vec b$, thereby realising $X_j$ on each term of the quadratic form expansion.
To represent a $Z_j$ operation, we use the fact that $Z_j \ket{\vec z} = (-1)^{z_j} \ket{\vec z} = (-1)^{\vec e_j\trans \vec z}$.
Letting $\vec a_j\trans = \vec e_j\trans A$ represent the $j\textsuperscript{th}$ row of $A$,
and using the fact that $\vec a_j\trans \vec x = \vec x\trans \!\mathrm{diag}(\vec a_j\trans) \vec x$ for $\vec x \in \{0,1\}^r$, we have
\begin{equation}
\begin{aligned}[b]
        Z_j \ket{A\vec x \oplus \vec b}
    \;&=\;
        (-1)^{\vec e_j\trans (A\vec x \oplus \vec b)} \ket{A \vec x \oplus \vec b} \\
    \;&=\;
        (-1)^{\vec a_j\trans\!\!\; \vec x + b_j} \ket{A \vec x \oplus \vec b}
    \;=\;
        (-1)^{\vec x\trans \!\!\; \mathrm{diag}(\vec a_j\trans) \vec x + b_j} \ket{A \vec x \oplus \vec b} \,.
\end{aligned}
\end{equation}
This allows us to simulate $Z_j$ on a quadratic form expansion by
\begin{equation}
\begin{aligned}[b]
        Z_j \ket{\psi}
      \;&=\;
        \frac{\tau^g}{\sqrt{2^r}}
            \!
            \sum_{\mathbf x \in \{0,1\}^r} 
                \!\!\!
                i^{\;\!\vec x\trans \!\!\;Q \vec x \,+\, 2 \vec x\trans \mathrm{diag}(\vec a_j\trans) \vec x + 2b_j}
                \,
                \ket{A \mathbf{x} \oplus \mathbf{b}}
    \;=\;
        \frac{\tau^{g'}}{\sqrt{2^r}}
            \!
            \sum_{\mathbf x \in \{0,1\}^r} 
                \!\!\!
                i^{\;\!\vec x\trans \! Q' \vec x}
                \,
                \ket{A \mathbf{x} \oplus \vec b}
    \;,
\end{aligned}
\end{equation}
for $g' = g + 4b_j$ and $Q' = Q + 2 \;\!\mathrm{diag}(\vec a_j\trans)$.
To represent a $Y = iXZ$ operation, we may simply add $2$ to the exponent of $\tau$ to incorporate the leading imaginary scalar, and then simulate $Z$ followed by $X$.
We summarise these three transformations by the procedures $\algname{SimulateX}(j)$ and $\algname{SimulateZ}(j)$, described in Figure~\ref{alg:Paulis}.
\begin{figure}[t]%
    \small
    ~\!\!\!\!\!\!\hfill
    \begin{minipage}[t]{.286\textwidth}\raggedright
        \rule{\textwidth}{1pt}
        $\algname{SimulateX}(j)$
        \smallskip
        \hrule
        \smallskip
        \begin{itshape}\small 
        Simulate the effect of an $X_j$ gate.%
        \end{itshape}%
        \smallskip
        \hrule
        \medskip
        \smallskip
        Update $b_j \gets b_j \oplus 1$\,. \\[.25ex]
        \rule{\textwidth}{1pt}
    \end{minipage}
    \hfill\hfill
    \begin{minipage}[t]{.35\textwidth}\raggedright
        \rule{\textwidth}{1pt}
        $\algname{SimulateZ}(j)$
        \smallskip
        \hrule
        \smallskip
        \begin{itshape}\small 
        Simulate the effect of a $Z_j$ gate.%
        \end{itshape}
        \smallskip
        \hrule
        \vspace*{-0.5ex}
        \begin{enumerate}[leftmargin=4ex,itemsep=-0.25ex]
        \item
            Update $g \gets g + 4 b_j \!\!\mod{8}$.
        \item
            Update $Q  \gets Q \!\!\;+\!\!\; 2 \;\!\mathrm{diag}(\vec e_j\trans \!\!\; A)\!\!\mod{4}$.
        \end{enumerate}
        \vspace*{-3ex}
        \rule{\textwidth}{1pt}
    \end{minipage}
    \hfill\hfill
    \begin{minipage}[t]{.35\textwidth}
        \rule{\textwidth}{1pt}
        $\algname{SimulateY}(j)$
        \smallskip
        \hrule
        \smallskip
        \begin{itshape}\small
        Simulate the effect of a $Y_{j}$ gate.
        \end{itshape}
        \smallskip
        \hrule
        \vspace*{-0.5ex}
        \begin{enumerate}[leftmargin=4ex,itemsep=-0.25ex]
        \item
            Update $g \gets g + 2 \!\!\mod{8}$.
        \item
            Call $\algname{SimulateZ}(j)$, then $\algname{SimulateX}(j)$.
        \end{enumerate}
        \vspace*{-3ex}
        \rule{\textwidth}{1pt}
    \end{minipage}
    \hfill\!\!\!\!\!\!~
\caption{%
    \label{alg:Paulis}%
    Procedures to simulate Pauli operations.
}
\end{figure}

\begin{lemma}
\label{lemma:simulatePaulis}
    Let $\ket{\psi}$ be represented by a quadratic form expansion as in Eqn.~\eqref{eqn:QFE-specific}, and $1 \le j \le n$.
    We may compute a quadratic form expansion for $X_j \ket{\psi}$ in time $\mathcal{O}(1)$, and for either $Y_j \ket{\psi}$ or $Z_j \ket{\psi}$ in time $\mathcal{O}(s_j) \subseteq \mathcal{O}(r)$.
\end{lemma}
\begin{proof}
    These follow from the run-time bounds for the procedures \algname{SimulateX}, \algname{SimulateY}, and \algname{SimulateZ}.
    The run-time bound of $\mathcal{O}(1)$ for \algname{SimulateX} is trivial, and the bound for \algname{SimulateY} and \algname{SimulateZ} is governed by the time required to modify the diagonal of $Q$ by adding the non-zero coefficients of row $j$ of $A$.
\end{proof}

\subsection{Hadamard operations}

Our techniques to simulate a Hadamard operation on a quadratic form expansion involve modifications to the matrix $A$.%
    \footnote{%
        It would not be difficult to modify our techniques to treat the special case that the qubit acted on is a $Y$-eigenstate, so that only the matrix $Q$ is modified in that case.
        We opt not to do so, to simplify the overall presentation of our analysis. 
    }
The complexity of these modifications to $A$ depend on its principal row form structure.
While the analysis of this operation involves a significant amount of case analysis in principle, we may describe it more simply using the results of Section~\ref{sec:ZeroColumnElim}.

\medskip\noindent
We represent the Hadamard operator using the commonplace formula for its coefficients,
\begin{equation}
        H
    \;=\;
        \frac{1}{\sqrt 2} \!\sum_{k,z \in \{0,1\}\!\!} \!\!
            (-1)^{kz} \,\ket{z}\!\!\!\;\bra{k}.
\end{equation}
This is itself a essentially a quadratic form expansion, in the terminology of Ref.~\cite{dBQuadratic}.
The image of the standard basis state $\ket{A \vec x \oplus \vec b}$ under the operator $\ket{z}\!\!\bra{k}_j \otimes I$ is non-zero, only if $k = \vec e_j\trans\!(A\vec x \oplus \vec b)$.
If this is the case, the state in the $j\textsuperscript{th}$ tensor factor of $\ket{A \vec x \oplus \vec b}$ is simply $\ket{k}$; the effect of the operator $\ket{z}\!\!\bra{k}$ is to replace this with $\ket{z}$.
Let $K_j = I \oplus \vec e_j\vec e_j\trans$ represent the map which acts on vectors and matrices by left-multiplication, to zero out row $j$.
We may then describe the effect of $H_j$ on a standard basis state $\ket{A\vec x \oplus \vec b}$ as
\begin{equation}
\label{eqn:Hadamard-zeroing-row}
\begin{aligned}[b]
        H_j \ket{A \vec x \oplus \vec b}
    \;&=\;
        \frac{1}{\sqrt 2} \!\!\sum_{k,z \in \{0,1\}\!\!} \!\!
            (-1)^{kz} \,\Bigl(\ket{z}\!\!\!\;\bra{k}_j \otimes I\Bigr) \ket{A \vec x \oplus \vec b}
    \\&=\;
        \frac{1}{\sqrt 2} \!\!\sum_{z \in \{0,1\}} \!\!
            (-1)^{(\vec e_j\trans\!(A\vec x \oplus \vec b))z} \ket{K_j(A \vec x \oplus \vec b) + z\vec e_j}
    \\&=\;
        \frac{1}{\sqrt 2} \!\!\sum_{z \in \{0,1\}} \!\!
            (-1)^{(\vec e_j\trans\!A\vec x)z + b_j z}
            \, \ket{(A' \vec x + z\vec e_j) \oplus \vec b'},
\end{aligned}
\end{equation}
where $A' = K_j A$ and $\vec b' = K_j \vec b$.
Applying this representation straightforwardly to a quadratic form expansion as in Eqn.~\eqref{eqn:QFE-specific} yields:
\begin{equation}
\begin{aligned}[b]
        H_j \ket{\psi}
    \;&=\;
        \frac{\tau^g}{\sqrt{2^{r{+}1}}}
        \!\!
        \sum_{\substack{\vec x \in \{0,1\}^r \\ z \in \{0,1\}}}
        \!\!\!\!
            i^{\;\! \vec x\trans \! Q \;\!\vec x \,+\, 2(\vec e_j\trans A \vec x)z + 2b_j z}
            \ket{(A' \vec x + z\vec e_j) \oplus \vec b'}
.
\end{aligned}
\end{equation}
We may condense this expression as follows, incorporating the vector $\vec x$ and bit $z$ into a single summation index $\vec y = {[\,x_1\;\;\cdots\;\;x_r\;\;z\,]}\:\!\trans$.
If we extend $Q$ to an $(r{+}1)\times(r{+}1)$ matrix $Q'$ with an additional row and column which is entirely zero, and let $\vec{\tilde a}\trans = {[\,A_{j,1}\;\;\cdots\;\;A_{j,r}\;\;0\,]}$ be row $j$ of $A$ extended by a further zero coefficient,
we have
\begin{equation}
\begin{aligned}[b]
        \vec x\trans \!Q\:\! \vec x \,+\, 2(\vec e_j\trans A \vec x)z \,+\, 2b_j z
    \;&=\;
        \vec y\trans \!Q' \vec y \,+\, y_{r{+}1} \vec{\tilde a}\trans \vec y \,+\, \vec y\trans \vec{\tilde a} \;\!y_{r{+}1} \,+\, 2b_j y_{r{+}1}
    \\&=\;
        \vec y\trans \!Q' \vec y \,+\, \vec y\trans \!(\vec e_{r{+}1} \vec{\tilde a}\trans + \vec{\tilde a} \;\!\vec e_{r{+}1}\trans) \vec y \,+\, 2b_j y_{r{+}1}
    \;=\;
        \vec y\trans\! Q'' \;\! \vec y\,,
\end{aligned}
\end{equation}
where we define $Q'' := Q' + \vec e_{r{+}1}\!\;\vec{\tilde a}\trans + \vec{\tilde a}\;\!\vec e_{r{+}1}\trans + 2 b_j \vec e_{r{+}1}\vec e_{r{+}1}\trans$\,.
Then, if we let $A'' = {\bigl[\,A'\;{;}\; \vec e_j\,\bigr]}$ be the matrix obtained by adjoining the vector $\vec e_j$ as an additional $(r{+}1)\textsuperscript{st}$ column to the matrix $A'$,
we have
\begin{equation}
\label{eqn:Hadamard-simple}
\begin{aligned}[b]
        H_j \ket{\psi}
    \;&=\;
        \frac{\tau^g}{\sqrt{2^{r{+}1}}}
        \!\!
        \sum_{\vec z \in \{0,1\}^{r{+}1}} 
        \!\!\!\!
            i^{\;\!
                \vec z\trans \! Q'' \vec z}
            \,
        \ket{A'' \vec z \oplus \vec b'}.
\end{aligned}
\end{equation}
Thus we may simulate the Hadamard by expanding the Gram matrix to an $(r{+}1) \times (r{+}1)$ matrix with a new row and column computed from $\vec b$ and from row $j$ of $A$; then zeroing out that row of $A$; then extending $A$ by a new column consisting of the vector $\vec e_j$.

If we were to impose no constraints on the number of columns of the expansion matrix $A''$, this would suffice to produce a representation of $H_j \ket{\psi}$.
However, it may be that $A''$ does not have rank $r{+}1$, so that it is not in principal row form as we require for our analysis of measurements.
We consider the following cases.
\begin{itemize}
\item 
    If $j$ is not a principal row of $A$, then the matrix $A'$ differs from $A$ only in row $j$, and not in any principal row of $A$.
    Extending this matrix to $A''$ only adds a further $0$ coefficient to each of the principal rows, and sets row $j$ to $\vec e_{r{+}1}\trans$.
    Then $A''$ actually is in principal row form in this case, and we may extend $p$ to obtain a principal index map for $A''$ by setting $p(r+1) = j$.
\item
    If $j = p(c)$ for some column $1 \le c \le r$, we may reduce our analysis to the case where $j$ is not a principal row --- provided that we can choose an alternative row to act as a principal row for column $c$.
    We may attempt to do this by invoking $\algname{ReselectPrincipalRow}(j,c)$.
    If afterwards $p(c) \ne j$, then $A$ has been transformed so that $j$ is not a principal row, and may proceed as above.
    
    If instead $\algname{ReselectPrincipalRow}(j,c)$ does not change the value of $p(c)$, it must be the case that $j$ is the only row in which column $c$ is non-zero.
    Then row $j$ of $A''$ would be $\vec e_{r{+}1}\trans$ rather than $\vec e_c\trans$\,, and column $c$ would be entirely zero.
    If we set $p'(r{+}1) = j$, then $A''$ is `nearly' in principal row form, precisely in the sense considered in Section~\ref{sec:ZeroColumnElim}: we have $\smash{\vec e_{p'(k)}\trans A''} = \vec e_k\trans$ for all $1 \le k \le r{+}1$ so long as $k \ne c$, and column $c$ of $A''$ is zero.
    Because Eqn.~\eqref{eqn:Hadamard-simple} expresses a normalised state, we also know that the procedure described on  page~\pageref{discn:normalised-zero-column-elim} will be able to find the non-zero coefficient in the Gram matrix, if this is needed to produce an expansion matrix in principal row form.%
        \footnote{%
            In fact, we can identify the location of one such coefficient analytically.
            Using the definitions following shortly after Eqn.~\eqref{eqn:Hadamard-zeroing-row}, and recalling in this case that $j$ is the principal row for column $c$, we have $Q'' \vec e_{r{+}1} = \vec{\tilde a} + 2b_j \vec e_{r{+}1} = \vec e_c + 2b_j \vec e_{r{+}1}$.
            Then $Q''_{c,r{+}1} = Q''_{r{+}1,c} = 1$; if we were to swap row $c$ and row $r{+}1$ (and similarly for the columns) of $Q''$, we would obtain a matrix $Q'''$ such that $Q'''_{c,r{+}1} = Q'''_{r{+}1,c} = 1$.}
    Therefore, to update the  quadratic form expansion for the operation $\ket{\psi'} = H_j \ket{\psi}$ as in Eqn.~\eqref{eqn:Hadamard-simple}, it then suffices to invoke $\algname{ZeroColumnElim}(c)$ in this case to obtain an equivalent representation in which the corresponding matrix $A$ has fewer columns, and in particular is full rank.
\end{itemize}
Figure~\ref{alg:SimulateH} presents a procedure $\algname{SimulateH}$, which summarises the remarks above.
Steps~1 and~2 serve to check whether $j$ is a principal row, and to attempt to choose an alternative principal row if necessary; Steps~3--6 represent the transformations described in the analysis above for when either $j$ is not a principal row, or when Step~2 succeeds in choosing a new principal row in place of row $j$.
In the latter case, Step~7 increases the rank parameter $r$; otherwise, if $j$ was a principal row which could not be re-selected, we invoke $\algname{ZeroColumnElim}(c)$ as described.

\begin{figure}[t]%
    \centering\small
    \begin{minipage}{.75\textwidth}
        \rule{\textwidth}{1pt}
        $\algname{SimulateH}(j)$
        \smallskip
        \hrule
        \smallskip
        \begin{itshape}\small 
        Simulate the effect of an $H_j$ gate.
        \end{itshape}
        \smallskip
        \hrule
        \begin{enumerate}[leftmargin=4ex,itemsep=.25ex]
        \item
            If $j$ is a principal row of $A$, let $1 \le c \le r$ be such that $j = p(c)$; otherwise let $c = 0$.
        \item
            If $c > 0$, call $\algname{ReselectPrincipalRow}(j,c)$.
            If afterwards $j \ne p(c)$, update $c \gets 0$.
        \item
            Let $\vec{\tilde a} = {[\,A_{j,1}\;\;\cdots\;\;A_{j,r}\;\;0\,}]\:\!\trans \in \{0,1\}^{r{+}1}$.
        \item
            Modify $A$ by updating $A_{j,k} \gets 0$ for each $1 \le k \le r$; 
            then extend $A$ by adjoining the column vector $\vec e_j$ and update $p(r{+}1) \gets j$.
        \item
            Modify $Q$ by extending by one row and one column, where the extra row is set to $\vec{\tilde a}\trans$ and the extra column is set to $\vec{\tilde a}$; then update $Q_{r{+}1,r{+}1} \gets 2b_j$\;.
        \item
            Update $b_j \gets 0$.
        \item
            If $c > 0$, call $\algname{ZeroColumnElim}(c)$; otherwise update $r \gets r + 1$.
        \end{enumerate}
        \vspace*{-3ex}
        \rule{\textwidth}{1pt}
    \end{minipage}
\caption{%
    \label{alg:SimulateH}
    A procedure to simulate a Hadamard operation.}
\end{figure}

\begin{lemma}
\label{lemH}
    Let $\ket{\psi}$ be represented by a quadratic form expansion as in Eqn.~\eqref{eqn:QFE-specific}, and $1 \le j\le n$.
    Then $\algname{SimulateH}(j)$ computes a quadratic form expansion for $H_{j} \ket{\psi}$ in time $\mathcal{O}((s{+}w)(t{+}w)) \subseteq \mathcal{O}(nr)$.
    If $j$ does not correspond to a principal row of $A$, then $\algname{SimulateH}(j)$ instead runs in time $\mathcal{O}(s_j) \subseteq \mathcal{O}(r)$.
\end{lemma}
\begin{proof}
    Step 1 takes $\mathcal{O}(1)$ operations, and Step 2 uses $\algname{ReselectPrincipalRow}$ at most once, taking $\mathcal{O}(s t + s w)$ operations.
    In Step~3, we copy row $j$ of $A$, taking time $\mathcal{O}(s_j)$.
    Step~4 involves adding a single non-zero entry to a new column of $A$ and assigning $p(r{+}1) \gets j$: using the sparse data structure for $A$, this takes time $\mathcal{O}(1)$.
    Step~5 extends $Q$ by a row and a column, each containing at most $s_j + 1$ non-zero entries,  requiring only time $\mathcal{O}(s_j)$ using the sparse data structure for $Q$.
    Step 6 takes $\mathcal{O}(1)$ operations, and finally Step 7 calls $\algname{ZeroColumnElim}$ at most once, which requires $\mathcal{O}(tw + w^2)$ operations.
    The time complexity is then dominated by Steps~2 and~7, which is therefore $\mathcal{O}(st + sw + tw + w^2)$; as these run-time costs are only incurred when $j$ is a principal row of $A$, the time complexity is $\mathcal{O}(s_j)$ when $j$ is not a principal row of $A$.
\end{proof}

\subsection{Diagonal Clifford operations}

Without much difficulty, we can show that the operators $S$ and $\mathrm CZ$ may be simulated on a quadratic form expansion as in Eqn.~\eqref{eqn:QFE-specific} just by modification of the scalar $g$ and the Gram matrix $Q$.

We first consider the $\mathrm CZ$ operation.
On an individual standard basis term $\ket{\vec z}$, we have $\mathrm CZ_{j,k} \ket{\vec z}
    \,=\,
        (-1)^{z_j z_k} \ket{\vec z}
    \,=\,
        (-1)^{(\vec e_j\trans\!\!\; \vec z)(\vec e_k\trans \!\!\; \vec z)} \ket{\vec z}$.
Let $\vec a_j\trans = \vec e_j\trans A$ and $\vec a_k\trans = \vec e_k\trans A$ be the $j\textsuperscript{th}$ and $k\textsuperscript{th}$ rows of $A$.
Then, on a standard basis state $\ket{A \vec x \oplus \vec b}$, we  obtain: 
    \begin{equation}{}
    \mspace{-18mu}
    \begin{aligned}[b]
        \mathrm CZ_{j,k} \ket{A \vec x \oplus \vec b}
    \,&=\,
        (-1)^{[\vec e_j\trans\! (A \vec x \,\oplus\, \vec b)][ \vec e_k\trans (A \vec x \oplus \vec b)]}
        \ket{A \vec x \oplus \vec b}
    \\&=\,
        (-1)^{[\vec a_j\trans\! \vec x \,+\, b_j] [\vec a_k\trans \vec x \,+\, b_k]}
        \ket{A \vec x \oplus \vec b}    
    \\&=\,
        (-1)^{\;\!(\vec a_j\trans \vec x)(\vec a_k\trans \vec x)\,+\, b_k(\vec a_j\trans \vec x) \,+\, b_j(\vec a_k\trans \vec x) \,+\, b_j b_k}
        \ket{A \vec x \oplus \vec b}
    \\&=\,
        i^{\;\!2(\vec a_j\trans \vec x)(\vec a_k\trans \vec x)\,+\, 2b_k(\vec a_j\trans \vec x) \,+\, 2b_j(\vec a_k\trans \vec x) \,+\, 2b_j b_k}
        \ket{A \vec x \oplus \vec b}
        .
    \end{aligned}
    \mspace{-12mu}
    \end{equation}
    Using the fact that $\vec u\trans \vec v = \vec v\trans \vec u$ for vectors $\vec u, \vec v \in \{0,1\}^r$, we have  $2(\vec a_j\trans \vec x)(\vec a_k\trans \vec x) = \vec x\trans\!(\vec a_j \vec a_k\trans + \vec a_k \vec a_j\trans)\vec x$\,;
    and using the fact that $\vec a_j\trans \vec x = \vec x\trans \!\mathrm{diag}(\vec a_j\trans) \vec x$ and $\vec a_k\trans \vec x = \vec x\trans \! \mathrm{diag}(\vec a_k\trans) \vec x$ for $\vec x \in \{0,1\}^r$,
    we then have
    \begin{equation}
    \begin{aligned}[b]
        \mathrm CZ_{j,k} \ket{\psi}
       \,&=\,
            \frac{\tau^g}{\sqrt{2^r}}
            \!\!\!
            \sum_{\mathbf x \in \{0,1\}^r} 
            \!\!\!\!
                i^{\;\!
                    \vec x\trans \! Q \vec x
                    \,+\,
                    \vec x\trans \!(\vec a_j \vec a_k\trans + \vec a_k \vec a_j\trans + 2 \,\mathrm{diag}(b_k \vec a_j\trans + b_j \vec a_k\trans)) \vec x
                    \,+\,
                    2b_j b_k}
                \,
                \ket{A \mathbf{x} \oplus \mathbf{b}}
        \\[1ex]&=\,
            \frac{\tau^{g'}}{\sqrt{2^r}}
            \!\!
            \sum_{\mathbf x \in \{0,1\}^r}
            \!\!\!\!
                i^{\:\!
                    \vec x\trans \!  Q' \!\: \mathbf x}
                \ket{A \mathbf{x} \oplus \mathbf{b}}    \;,
    \end{aligned}
\end{equation}
where $g' := g + 4b_jb_k$ and $Q' := Q \,+\,  \vec a_j \vec a_k\trans + \vec a_k \vec a_j\trans + 2\,\mathrm{diag}(b_k \vec a_j\trans + b_j \vec a_k\trans)$\,.
That is, we can simulate $\mathrm CZ$ by adding terms to the Gram matrix which may be easily computed from rows of $A$ and coefficients of $\vec b$; in particular, given that $Q$ is symmetric, $Q'$ is as well.

To simulate the effect of $S = \ket{0}\!\bra{0} + i \ket{1}\!\bra{1}$ on a basis state $\ket{A \vec x \oplus \vec b}$, the outcome of a matrix multiplication which is evaluated modulo~$2$ in a ket may come to affect an expression evaluated modulo~$4$ in an imaginary exponent.
To do this, we use the same operation $\ast$ which we defined in Eqn.~\eqref{eqn:ast-operation}, apply the formula of Eqn.~\eqref{eqn:evalAstMod4} to express it as a matrix operation modulo~4, and also use the formula of Eqn.~\eqref{eqn:parity-formula-01} for parities of bits.
We may then analyse the $S$ gate similarly to the $\mathrm CZ$ gate.
On an individual standard basis term $\ket{\vec z}$, we have $S_j \ket{\vec z} = i^{\:\!z_j} \ket{\vec z} = i^{\:\!(\vec e_j\trans\!\!\; \vec z)} \ket{\vec z}$\,: then, again letting $\vec a_j\trans = \vec e_j\trans A$, we have
    \begin{equation}{}
    \mspace{-18mu}
    \begin{aligned}[b]
        S_j \ket{A \vec x \oplus \vec b}
    \,=\,
        i^{\,\vec e_j\trans\! (A \ast \vec x \,\oplus\, \vec b)}
        \ket{A \vec x \oplus \vec b}
    \,&=\,
        i^{\;\!(\vec e_j\trans \!\!\: A \ast \vec x) \,\oplus\,  (\vec e_j\trans \vec b)}
        \ket{A \vec x \oplus \vec b}
    \\&=\,
        i^{\;\!(\vec x\trans \!\!\!\: \vec a_j \vec a_j\trans \vec x) \,\oplus\,  b_j}
        \ket{A \vec x \oplus \vec b}
    \\&=\,
        i^{\;\!(\vec x\trans \!\!\!\; \vec a_j \vec a_j\trans  \vec x) \,+\, b_j \,-\, 2 b_j (\vec x\trans \!\!\!\; \vec a_j \vec a_j\trans  \vec x)}
        \ket{A \vec x \oplus \vec b}
        .
    \end{aligned}
    \mspace{-12mu}
    \end{equation}
    We may then describe the representation of $S_j$ on a quadratic form expansion, by
    \begin{equation}
    \begin{aligned}[b]
        S_j \ket{\psi}
       \,&=\,
            \frac{\tau^g}{\sqrt{2^r}}
            \!\!\!
            \sum_{\mathbf x \in \{0,1\}^r} 
            \!\!\!\!
                i^{\;\!
                    \vec x\trans \! Q \vec x
                    \,+\,
                    (1 - 2b_j) \vec x\trans (\vec a_j\vec a_j\trans) \vec x + b_j}
                \,
                \ket{A \mathbf{x} \oplus \mathbf{b}}
        \,=\,
            \frac{\tau^{g'}}{\sqrt{2^r}}
            \!\!
            \sum_{\mathbf x \in \{0,1\}^r}
            \!\!\!\!
                i^{\:\!
                    \vec x\trans \!  Q' \!\: \mathbf x}
                \ket{A \mathbf{x} \oplus \mathbf{b}}    \;,
    \end{aligned}
\end{equation}
where $g' := g + 2b_j$ and $Q' := Q \,+\, (1 - 2b_j) \vec a_j \vec a_j\trans$\,.
We then simulate $S_j$ by again adding symmetric terms to the Gram matrix which may be easily computed from rows of $A$ and coefficients of $\vec b$.

\begin{figure}[t]%
    \small
    ~\hfill
    \begin{minipage}{.375\textwidth}
        \rule{\textwidth}{1pt}
        $\algname{SimulateS}(j)$
        \smallskip
        \hrule
        \smallskip
        \begin{itshape}\small 
        Simulate the effect of an $S_j$ gate.
        \end{itshape}
        \smallskip
        \hrule
        \begin{enumerate}[leftmargin=4ex,itemsep=.25ex]
        \item
            Let $\vec a_j = {[\,A_{j,1}\;\;\cdots\;\;A_{j,r}\,}]\:\!\trans \in \{0,1\}^r$.
        \item
            Update $Q \gets Q + (1 - 2b_j) \vec a_j \vec a_j\trans$.
        \item
            For each $1 \le k \le r$ such that $A_{j,k} \ne 0$: ~\\
                call $\algname{ReduceGramRowCol}(k)$.
        \item
            Update $g \gets g + 2b_j$\,.
        \end{enumerate}
        \vspace*{-2ex}
        \rule{\textwidth}{1pt}
    \end{minipage}
    ~\hfill\hfill~
    \begin{minipage}{.575\textwidth}
        \rule{\textwidth}{1pt}
        $\algname{SimulateCZ}(j,k)$
        \smallskip
        \hrule
        \smallskip
        \begin{itshape}\small
        Simulate the effect of a $\mathrm CZ_{j,k}$ gate.
        \end{itshape}
        \smallskip
        \hrule
        \begin{enumerate}[leftmargin=4ex,itemsep=.25ex]
        \item
            Let $\vec a_j = {[\,A_{j,1}\;\;\cdots\;\;A_{j,r}\,}]\:\!\trans\!,\; 
            \vec a_k = {[\,A_{k,1}\;\;\cdots\;\;A_{k,r}\,}]\:\!\trans \in \{0,1\}^r$.
        \item
            Update $Q \gets Q + \vec a_j\vec a_k\trans + \vec a_k\vec a_j\trans + 2\:\!\mathrm{diag}(b_k\vec a_j\trans + b_j \vec a_k\trans)$.
        \item
            For each $1 \le h \le r$ such that $A_{j,h} \ne 0$ or $A_{k,h} \ne 0$: ~\\
                call $\algname{ReduceGramRowCol}(h)$.
        \item
            Update $g \gets g + 4 b_j b_k$\,.
        \end{enumerate}
        
        \vspace*{-2ex}
        \rule{\textwidth}{1pt}
    \end{minipage}
    \hfill~
\caption{%
    \label{alg:diagonal-Cliffords}%
    Procedures to simulate the $S = \mathrm{diag}(1,i)$ and $\mathrm CZ = \mathrm{diag}(+1,+1,+1,-1)$ operations.
}
\end{figure}

\medskip\noindent
Figure~\ref{alg:diagonal-Cliffords} presents procedures $\algname{SimulateS}$ and $\algname{SimulateCZ}$ which summarise the analysis described above.
\begin{lemma}
\label{lemS}
    Let $\ket{\psi}$ be represented by a quadratic form expansion as in Eqn.~\eqref{eqn:QFE-specific}, and $1 \le j \le n$.
    We may compute a quadratic form expansion for $S_j \ket{\psi}$ in time $\mathcal{O}(s_j^2 + s_j w) \subseteq \mathcal{O}(r^2)$.
\end{lemma}
\begin{proof}
In $\algname{SimulateS}(j)$, the vector $\vec a_j$ can simply be read from row $j$ of $A$ in Step~1 in time $\mathcal{O}(s_j)$.
Note that $\vec a_j \vec a_j\trans$ is non-zero in precisely $s_j^2$ entries.
Step~2 then modifies the Gram matrix $Q$ in time $\mathcal{O}(s_j^2)$.
In Step~3, we call $\algname{ReduceGramRowCol}(k)$ for those indices $0 \le k \le r$ corresponding to non-zero entries of $\vec a_j$.
This has complexity $\mathcal{O}(w_k)$ for a maximum of $s_j$ values of $k$, taking $\mathcal{O}(s_j w)$ time in total.
Step~4 modifies the value of $g$, in constant time; the total complexity of $\mathcal{O}(s_j^2 + s_j w)$ then follows.%
    \footnote{%
            \label{fn:tightenDiagonalCliffordRuntimes}%
            The run-times of $\algname{SimulateS}(j)$ and of $\algname{SimulateCZ}(j,k)$ can be easily be sharpened to $\mathcal{O}(s_j^2)$ and $\mathcal{O}(s_j s_k)$ respectively, by modifying them to only reduce the entries where $Q'$ and $Q$ differ.
            We use the definitions presented in Figure~\ref{alg:diagonal-Cliffords} to simplify the overall presentation of our results.
        }
\end{proof}

\begin{lemma}
\label{lemCZ}
    Let $\ket{\psi}$ be represented by a quadratic form expansion as in Eqn.~\eqref{eqn:QFE-specific}, and $1 \le j,k \le n$. 
    We may then compute a quadratic form expansion for $\mathrm CZ_{j,k} \ket{\psi}$ in time $\mathcal{O}(s_j s_k + s_j w + s_k w) \subseteq \mathcal{O}(r^2)$.
\end{lemma}
\begin{proof}
    In $\algname{SimulateCZ}(j,k)$, the vectors $\vec a_j$ and $\vec a_k$ can simply be read from rows $j$ and $k$ of $A$ in Step~1 in time $\mathcal{O}(s_j + s_k)$.
    Note that $\vec a_j \vec a_k\trans$ and $\vec a_k \vec a_j\trans$ are each non-zero in precisely $s_j s_k$ entries.
    Step~2 then modifies the Gram matrix $Q$ in time $\mathcal{O}(s_j s_k)$.
    In Step~3, we call $\algname{ReduceGramRowCol}(h)$ for those indices $0 \le k \le r$ corresponding to non-zero entries either of $\vec a_j$ or $\vec a_k$.
    This has complexity $\mathcal{O}(w_h)$ for a maximum of $s_j + s_k$ values of $k$, taking $\mathcal{O}(s_j w + s_k w)$ time in total.
    Step~4 modifies the value of $g$, in constant time; the total complexity of $\mathcal{O}(s_j s_k + s_j w + s_k w)$ then follows.%
    \addtocounter{footnote}{-1}%
    \footnotemark
\end{proof}

\subsection{Controlled-NOT operations}

The way in which we simulate $\mathrm CX$ operations, with control qubit $h$ and target qubit $j$ for $h \ne j$, depends on whether or not $j$ corresponds to a principal row of $A$.
In the case that it does not, we may simulate it in a straightforward way, by noting that on standard basis states we have
\begin{equation}
    \mathrm CX_{h,j} \ket{\vec z}
    \;=\;
    \ket{\vec z \oplus z_h \vec e_j}
    \;=\;
    \ket{E_{j,h} \vec z},
\end{equation}
where $E_{j,h} = I + \vec e_j \vec e_h\trans$ acts on vectors via left-multiplication by adding row $h$ into row $j$ (and similarly for matrices).
Thus we have
\begin{equation}
    \begin{aligned}[b]
            \mathrm CX_{h,j} \ket{\psi}
        \,=\,
            \frac{\tau^g}{\sqrt{2^r}}
            \!\!
            \sum_{\mathbf x \in \{0,1\}^r} 
            \!\!\!\!
                i^{\;\! \vec x\trans \! Q \vec x}
                \,
                \ket{E_{j,h} A \mathbf{x} \oplus E_{j,h} \mathbf{b}}
        \,=\,
            \frac{\tau^g}{\sqrt{2^r}}
            \!\!
            \sum_{\mathbf x \in \{0,1\}^r} 
            \!\!\!\!
                i^{\;\! \vec x\trans \! Q \vec x}
                \,
                \ket{A' \mathbf{x} \oplus \mathbf{b}'}
            \;,
    \end{aligned}
\end{equation}
where $A' = E_{j,h} A$ is obtained from $A$ by adding row $h$ into row $j$ modulo 2, and $\vec b' = E_{j,h} \vec b$ differs from $\vec b$ in that $b'_j = b_j \oplus b_h$.

If $j = p(c)$ for some column $c$ of $A$, then $A'$ will not be in principal row form unless row $h$ of $A$ happened to be entirely zero.
However, as $E_{h,j}$ is invertible and $A$ is full rank, then $A'$ will also be full rank; and it will only fail to be in principal row form in that there is no row which contains $\vec e_c\trans$.
Precisely because $A'$ is full rank, column $c$ of $A'$ will not be entirely zero --- in particular, either $A_{j,c} = 1$ (as in the case that $A_{h,c} = 0)$, or row $h$ is itself not a principal row and $A_{h,c} = 1$. 
In either case, there is at least one row $1 \le j_\ast \le n$, for which row $j_\ast$ of $A'$ could be made into a principal row for column $c$ (yielding a matrix $A''$ in principal row form) by suitable column operations, without disturbing the other principal rows.
We may find such a row $j_\ast$\,, and perform the appropriate change of variables, by simply invoking the subroutine $\algname{ReselectPrincipalRow}(0,c)$ --- in particular, setting the first argument to $0$ to allow the possibility of selecting $j_\ast = j$ if this is the only row (or the row with the smallest number of non-zero entries) such that $A_{j_\ast,c} = 1$.

\begin{figure}[t]
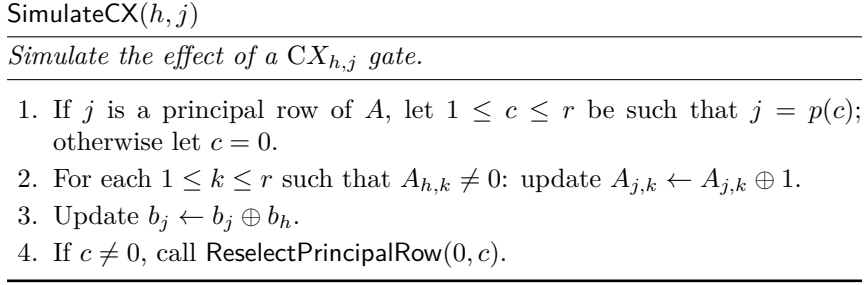
%
    \small
    ~\hfill
    \begin{minipage}[t]{.75\textwidth}
        \rule{\textwidth}{1pt}
        $\algname{SimulateCX}(h,j)$
        \smallskip
        \hrule
        \smallskip
        \begin{itshape}\small
        Simulate the effect of a $\mathrm CX_{h,j}$ gate.
        \end{itshape}
        \smallskip
        \hrule
        \vspace*{-0.5ex}
        \begin{enumerate}[leftmargin=4ex,itemsep=-0.25ex]
        \item
            If $j$ is a principal row of $A$, let $1 \le c \le r$ be such that $j = p(c)$; otherwise let $c = 0$.
        \item
            For each $1 \le k \le r$ such that $A_{h,k} \ne 0$: update $A_{j,k} \gets A_{j,k} \oplus 1$.
        \item
            Update $b_j \gets b_j \oplus b_h$.
        \item
            If $c \ne 0$, call $\algname{ReselectPrincipalRow}(0,c)$.
        \end{enumerate}
        \vspace*{-3ex}
        \rule{\textwidth}{1pt}
    \end{minipage}
    \hfill~
\caption{%
    \label{alg:CNOT}%
    Procedure to simulate a controlled-NOT operation.
}
\end{figure}

Figure~\ref{alg:CNOT} presents a procedure $\algname{SimulateCX}(h,j)$ which summarises the above analysis.
Using the run-time bound for $\algname{ReselectPrincipalRow}$ described in in Lemma~\ref{lemma:runtime-ReselectPrincipalRow}, we may easily show the following:
\begin{lemma}
    Let $\ket{\psi}$ be represented by a quadratic form expansion as in Eqn.~\eqref{eqn:QFE-specific}, and $1 \le h,j \le n$.
    Then $\algname{SimulateCX}(h,j)$ computes a quadratic form expansion for $\mathrm CX_{h,j} \ket{\psi}$ in time $\mathcal{O}(s_h t + s_h w + t + w) \in \mathcal{O}(nr)$.
    If $j$ does not correspond to a principal row of $A$, then $\algname{SimulateCX}(h,j)$ instead runs in time $\mathcal{O}(s_h) \subseteq \mathcal{O}(r)$.    
\end{lemma}
\begin{proof}
    We may perform Step~1 in time $\mathcal{O}(1)$.
    Step~2 involves iterating over the $s_h$ indices $0 \le k \le r$ for which $A_{h,k} \ne 0$, performing $\mathcal{O}(1)$ operations on each iteration; and Step~3 can also be performed in time $\mathcal{O}(1)$.
    In Step~4, which is only invoked if $j$ was a principal row of the expansion matrix $A$ at the input, we invoke $\algname{ReselectPrincipalRow}(0,c)$, which runs in time $\mathcal{O}((s_h+1)t + (s_h+1)w)$, using the fact that the number of non-zero entries in row $j$ at this step is at most $s_h + 1$.
\end{proof}

\subsection{Pauli basis measurements}

As quadratic form expansions emphasise the decomposition of a state in the standard basis, this leads to simpler procedures to simulate $Z$-basis measurements compared to $X$- or $Y$-basis measurements: $Z$-basis measurements may only decrease the number of columns of $A$, whereas $X$- and $Y$-basis measurements have an analysis which is closer to that of the Hadamard operation.
However, these operations all can be done particularly quickly when the qubit being measured is disentangled from the other qubits --- for instance, in those cases where the measurement outcome is in principle deterministic.

We note that, in principle, one may simulate an X- or a Y-basis measurement by performing a suitable single-qubit unitary $U$ (respectively: $U = H$ or $U = S\herm  H$), followed by a Z-measurement, followed by $U\herm$.
If our aim were only to demonstrate that these operations could be simulated in $\mathcal{O}(n^2)$ time, this would suffice.
However, one might wish to simulate these measurement operations without \emph{necessarily} reducing them to Z-basis measurements, if avoiding that reduction could provide a savings in operations.
We find that this is possible (see Lemma~\ref{lemma:SimulateXYMeas} below), using an analysis which builds on the similarity between simulating these measurements and simulating the Hadamard transformation.
For this reason, we treat each of the Pauli measurements on an equal footing.

\subsubsection{Simulating measurements with deterministic outcomes}
\label{sec:simulDeterministicMeas}

As we maintain the expansion matrix $A$ in principal row form, none of the terms in Eqn.~\eqref{eqn:QFE-specific} may cancel.
Then, a $Z$-basis measurement on qubit $j$ is deterministic if qubit $j$ is in the some particular state $\ket{\beta}$ (for a bit $\beta \in \{0,1\}$) in every term of the quadratic form expansion --- and in particular, does not depend on any bit of the summation index $\vec x$.
This occurs if and only if row $j$ of $A$ is entirely zero, which can be determined in time $\mathcal{O}(1)$.
If this is the case, the outcome will be $\beta = b_j$\,, which again can be computed in time $\mathcal{O}(1)$.

Remarkably, we may show that $X$- and $Y$-basis measurements with deterministic outcomes can also be simulated in time $\mathcal{O}(1)$.
Consider a state $\ket{\psi}$, in which either an $X$-basis measurement or a $Y$-basis measurement on some qubit $j$ would have a determinstic outcome.
Such measurements cannot be deterministic if qubit $j$ is in a $Z$-eigenbasis state, so they can be deterministic only if row $j$ of $A$ has a non-zero entry.
Furthermore, the outcome of such a measurement can only be deterministic if there is no entanglement between qubit $j$ and any other qubits.
In particular:
\begin{itemize}
\item 
    There can be no correlations between the outcomes of $Z$-basis measurements on qubit $j$ and any other qubits.
    This implies that it must be possible to reindex the sum by column operations on $A$, so that $\vec e_j$ is a column of $A$: that is, there must be a solution to the set of equations $\vec e_j = A \vec v$ for some $\vec v \in \{0,1\}^r$.
    However, for $A$ in principal row form, row $p(k)$ of $A \vec v$ will be non-zero for any $k$ such that $v_k = 1$.
    It follows that $\vec e_j = A \vec v$ is only solvable if $\vec v$ has exactly one non-zero entry: that is, if $\vec e_j$ is itself a column of $A$.
    Furthermore, as only row $j$ of $A$ is non-zero in column $c$, this implies that row $j$ is a principal row with $j = p(c)$.
    
    We may determine whether this is the case in constant time, by determining whether row $j$ of $A$ has exactly one non-zero entry, finding the column $c$ in which that non-zero entry occurs, and then checking whether column $c$ has exactly one non-zero entry.
\item
    Given that the above holds, the state of qubit $j$ in each term of the quadratic form expansion is $\ket{x_c}$ for some $1 \le c \le r$, and only other variables $x_k$ for $k \ne c$ are involved in the state of the other qubits.
    For qubit $j$ to be unentangled from the others, the relative phases of the terms in the expansion must be a product of the relative phases for the state of qubit $j$, and the relative phases for the state of the other qubits, with no contribution from cross-terms $x_c x_k$.
    By Lemma~\ref{lemma:Gram-matrix-mixed-moduli}, this is equivalent the off-diagonal coefficients in row/column~$c$ of $Q$ all being even.
    
    As our subroutines all evaluate the off-diagonal terms modulo $2$, it suffices to check whether all coefficients in row/column $c$ of $Q$ are zero, except possibly $Q_{c,c}$.
    We may determine this in constant time by checking  whether row $c$ of $Q$ has at most one non-zero coefficient, and (if one such coefficient exists) whether that coefficient is on the diagonal.
\end{itemize}
Together, these conditions suffice to show that qubit $j$ is in one of the states $\ket{\texttt +}$, $\ket{\texttt{+i}}$, $\ket{\texttt -}$, or $\ket{\texttt{-i}}$.
\emph{Which} of these states it is in, is determined by  whether $(-1)^{b_j} Q_{c,c} = 0$, $1$, $2$, or $3$ respectively (corresponding to a relative phase of $+1 = i^0$, or $i = i^1$, or $-1 = i^2$, or $-i = i^3$); this may also be determined in time $\mathcal{O}(1)$.

By construction, much of this analysis generalises to the case where the qubit to be measured is in a single-qubit stabiliser state, unentangled with any others, even if the measurement outcome is \emph{not} deterministic: it will be in an eigenstate of \emph{some} Pauli operator $\{X_j,Y_j,Z_j\}$.
Whether this is the case may be determined in time $\mathcal{O}(1)$; a random outcome and an update to the quadratic form expansion can then be computed in constant time as well.

\subsubsection{Simulating \textit{Z}-basis measurements}

To describe the effect of a measurement with a random outcome $\beta$, we perform a change of variables so that row $j$ is a principal row --- and so that in particular, the value of qubit $j$ in any term in the superposition depends only on the bit $x_r$.
(We may attempt to do so in a way to minimise the amount of work required.) 
We then fix the value of that particular bit, using the subroutine $\algname{FixFinalBit}(\beta\oplus b_j)$ from Section~\ref{sec:fixBits} --- where we take the argument $\beta \oplus b_j$ to over-write the value of $b_j$ with $\beta$.

\begin{figure}[t]%
    \centering\small
    \begin{minipage}{.9125\textwidth}
        \rule{\textwidth}{1pt}
        $\algname{SimulateMeasZ}(j)$
        \smallskip
        \hrule
        \smallskip
        \begin{itshape}\small 
            Simulate a $Z$-basis measurement on qubit $j$, storing the result in a bit $\beta$ produced as output.
        \end{itshape}
        \smallskip
        \hrule
        \begin{enumerate}[leftmargin=4ex,itemsep=.25ex]
        \item
            If row $j$ of $A$ is zero, set $\beta \gets b_j$ and stop; otherwise set $\beta$ to a uniformly random bit-value and proceed.
        \item
            Find $1 \le k \le r$ such that $A_{j,k} \ne 0$,
            minimising the number of non-zero entries in column $k$ of $A$.
        \item
            Call $\algname{ReindexSwapColumns}(k,r)$.
        \item 
            Call $\algname{MakePrincipal}(r,j)$.
        \item
            Call $\algname{FixFinalBit}(\beta\oplus b_j)$.
        \end{enumerate}
        \vspace*{-3ex}
        \rule{\textwidth}{1pt}
    \end{minipage}
\caption{%
    \label{alg:SimulateMeasZ}
    A procedure to simulate a $Z$-basis measurement.
}
\end{figure}

Figure~\ref{alg:SimulateMeasZ} presents an algorithm $\algname{SimulateMeasZ}(j)$, which supplements the analysis of deterministic $Z$-basis measurements with the operations described above.

\begin{lemma}
\label{lemZMeas}
    Let $\ket{\psi}$ be represented by a quadratic form expansion as in Eqn.~\eqref{eqn:QFE-specific}, and $1 \le j\le n$.
    Then $\algname{SimulateMeasZ}(j)$ computes a quadratic form expansion for a post-measurement state arising from performing a $Z$-basis measurement on the state $\ket{\psi}$ in time ${\mathcal{O}(s_j t + s_j w)} \subseteq \mathcal{O}(nr)$ in general, generating a uniformly random outcome if required.
    In particular: if $j$ is a principal row of $A$, $\algname{SimulateMeasZ}(j)$ runs in time ${\mathcal{O}(t {+} w)} \subseteq \mathcal{O}(n)$.
    Furthermore, in the case where the measurement outcome is deterministic, $\algname{SimulateMeasZ}(j)$ returns the measurement outcome in time $\mathcal{O}(1)$ without modifying the input.
\end{lemma}
\begin{proof}
    Step 1 takes $\mathcal{O}(1)$ operations: this is all that is required in the deterministic case.
    Step 2 iterates over the non-zero entries of row $j$ of $A$, performing simple integer comparisons and assignments in each iteration, taking $\mathcal{O}(s_j)$ operations in total.
    Steps~3, 4, and~5 invoke the subroutines $\algname{ReindexSwapColumns}$, $\algname{MakePrincipal}$, and $\algname{FixFinalBit}$, using $\mathcal{O}(t+w)$, $\mathcal{O}(s_j t + s_j w)$ and $\mathcal{O}(t+w)$ operations respectively.
    The total complexity is then dominated by Step~4, which has run-time $\mathcal{O}(s_j t + s_j w)$.
\end{proof}

\subsubsection{\textit{X}- and \textit{Y}-basis measurements with random outcomes}
\label{sec:XY-random-meas}

For the sake of convenience, we define $\ket{\mathsf{x}_\beta} = \smash{\tfrac{1}{\sqrt 2}\bigl(\ket{0} + i^{\:\!2\beta} \ket{1}\bigr)}$ and $\ket{\mathsf{y}_\beta} = \smash{\tfrac{1}{\sqrt 2}\bigl(\ket{0} + i^{\:\!2\beta{+}1}\ket{1}\bigr)}$ as notation for the eigenstates of $X$ and $Y$ respectively, where $\beta \in \{0,1\}$.
To simulate an $X$- or $Y$-basis measurement on qubit $j$ which has a random outcome, we may generate a measurement outcome $\beta \in \{0,1\}$ uniformly at random, and then determine the effect of applying the
projectors
\begin{equation}
        \ket{\mathsf x_\beta}\!\!\bra{\mathsf x_\beta}
    \;=\;
        \frac{1}{2} \! \sum_{k,z \in \{0,1\}\!\!} \!
            (-1)^{\beta(z-\:\!k)} \ket{z}\!\!\bra{k}\,,
    \qquad
        \ket{\mathsf y_\beta}\!\!\bra{\mathsf y_\beta}
    \;=\;
        \frac{1}{2} \! \sum_{k,z \in \{0,1\}\!\!} \!\!\!
            i^{\:\!(2\beta+1)(z-\:\!k)} \ket{z}\!\!\bra{k}
\end{equation}
Let $\ket{\psi^{(X)}_\beta} = \smash{\sqrt 2 \;\!\bigl(\ket{\mathsf x_\beta}\!\!\bra{\mathsf x_w}_j \otimes I\bigr) \ket{\psi}}$ denote the state after applying an $X$-basis measurement on qubit $j$ of  $\ket{\psi}$, given that the outcome $\beta \in \{0,1\}$ is not deterministic (where the factor of $\sqrt 2$ is to renormalise the state); similarly, let $\ket{\psi^{(Y)}_\beta} = \smash{\sqrt 2 \;\!\bigl(\ket{\mathsf y_\beta}\!\!\bra{\mathsf y_\beta}_j \otimes I\bigr) \ket{\psi}}$ denote the state after applying an $Y$-basis measurement on qubit $j$ of $\ket{\psi}$, given that the outcome $\beta \in \{0,1\}$ is not deterministic.
Following a similar analysis as for the Hadamard: if we let $K_j = I \oplus \vec e_j \vec e_j\trans$, and define $A' = K_j A$ and $\vec b'  = K_j \vec b$, we may show
\begin{subequations}%
\label{eqn:QFE-XYmeas-1}%
\begin{align}
        \ket{\psi^{(X)}_\beta}
    \;&=\;
        \frac{\tau^g}{\sqrt{2^{r{+}1}}}
        \!\!
        \sum_{\substack{\vec x \in \{0,1\}^r \\ z \in \{0,1\}}}
        \!\!\!\!
            i^{\;\! \vec x\trans \! Q \;\!\vec x}\,(-1)^{\beta(z\,-\,\vec e_j\trans A \vec x \,-\, b_j)}\,
            \ket{(A' \vec x + z\vec e_j) \oplus \vec b'} ,
    \\[1ex]
        \ket{\psi^{(Y)}_\beta}
    \;&=\;
        \frac{\tau^g}{\sqrt{2^{r{+}1}}}
        \!\!
        \sum_{\substack{\vec x \in \{0,1\}^r \\ z \in \{0,1\}}}
        \!\!\!\!
            i^{\;\! \vec x\trans \! Q \;\!\vec x
                    \,+\, (2\beta+1)(z \,-\, \vec e_j\trans\!(A\ast \vec x) \,-\, b_j \,+\, 2 \vec e_j\trans\!(A\ast \vec x) b_j)}\,
            \ket{(A' \vec x + z\vec e_j) \oplus \vec b'} ,
\end{align}%
\end{subequations}%
using the operation $\ast$ as defined in Eqn.~\eqref{eqn:ast-operation}, and applying the formula $u \oplus v = u + v - 2uv$ to coefficient $j$ of $(A\!\ast\!\vec x \oplus \vec b)$ in the imaginary exponent.
Let $\vec a\trans = \vec e_j\trans\! A$: we may then simplify the imaginary exponents in Eqns.~\eqref{eqn:QFE-XYmeas-1}, and in particular factor out constant terms which only contribute global phase factors, to obtain
\begin{subequations}%
\allowdisplaybreaks
\label{eqn:QFE-XYmeas-2}%
\begin{align}
        \ket{\psi^{(X)}_\beta }
    \;&=\;
        \frac{\tau^{g-4 \beta b_j}}{\sqrt{2^{r{+}1}}}
        \!\!
        \sum_{\substack{\vec x \in \{0,1\}^r \\ z \in \{0,1\}}}
        \!\!\!\!
            i^{\;\! \vec x\trans \! Q \;\!\vec x \,+\, 2 \beta  z\,-\,2 \beta \vec  a\trans\! \vec x}\,
            \ket{(A' \vec x + z\vec e_j) \oplus \vec b'} ,
    \\[1ex]
        \ket{\psi^{(Y)}_\beta}
    \;&=\;
        \frac{\tau^{g\;\!-\;\!(4\beta+2)b_j}}{\sqrt{2^{r{+}1}}}
        \!\!
        \sum_{\substack{\vec x \in \{0,1\}^r \\ z \in \{0,1\}}}
        \!\!\!\!
            i^{\;\! \vec x\trans \! Q \;\!\vec x
                    \,+\, (2\beta+1)z \,+\, (2\beta+1)(2b_j-1)(\vec a\trans\! \ast \;\!\vec x)}\,
            \ket{(A' \vec x + z\vec e_j) \oplus \vec b'} .
\end{align}%
\end{subequations}%
We may again condense the indices of summation by defining $\vec y = {[\,x_1\;\;\cdots\;\;x_r\;\;z\,]\:\!\trans}$, and letting $A'' = {\bigl[\,A' \;{;}\; \vec e_j\,\bigr]}$ be the matrix obtained by adjoining the vector $\vec e_j$ as an additional $(r{+}1)\textsuperscript{st}$ column to the matrix $A'$.
Let $\vec{\tilde a}\trans = {[\,A_{j,1}\;\;\cdots\;\;A_{j,r}\;\;0\,]}$ be an extension of $\vec a\trans$ by a further zero coefficient, and let $Q'$ be an $(r{+}1)\times(r{+}1)$ matrix which extends $Q$ with an additional row and column which is entirely zero.
We then define Gram matrices
\begin{subequations}%
\begin{align}
        Q^{(X)}
    \;&=\;
        Q' \,+\, 2\beta \;\!\mathrm{diag}(\vec{\tilde a}\trans \!+ \vec e_{r{+}1}\trans)    \,,
    \\[1ex]
        Q^{(Y)}
    \;&=\;
        Q'
        \,+\,
        (2b_j+2\beta-1) \vec{\tilde a}\;\!\vec{\tilde a}\trans
        \,+\,
        (2\beta +1)\vec e_{r{+}1}\vec e_{r{+}1}\trans
        \,,
    \end{align}
\end{subequations}
so that we may express the imaginary exponents in Eqns.~\eqref{eqn:QFE-XYmeas-2} as
\begin{subequations}%
\begin{align}
        \vec x\trans \!Q\:\! \vec x \,+\, 2\beta z - & 2\beta \vec a\trans\!\vec x \nonumber 
    \\&=\;
        \vec y\trans \!Q' \vec y \,+\, 2\beta y_{r{+}1} - 2\beta \vec{\tilde a}\trans\! \vec y
    \notag\\&\equiv\;
        \vec y\trans \!Q' \vec y \,+\, \vec y\trans \!\bigl(2\beta\:\!\mathrm{diag}(\vec e_{r{+}1}\trans) + 2\beta\:\!\mathrm{diag}(\vec{\tilde a}\trans)\bigr) \vec y  \mod{4}
        \mspace{-9mu}
    \notag\\&=\;
        \vec y\trans\! Q^{(X)} \;\! \vec y\,,
    \\[2ex]\mspace{-9mu}
        \vec x\trans \!Q\:\! \vec x \,+\, (2\beta\!+\!1)z + & (2\beta\!+\!1)(2b_j\!-\!1)(\vec a\trans\!\!\ast\!\!\; \vec x) \nonumber
    \\&=\;
        \vec y\trans \!Q' \vec y \,+\, (2\beta\!+\!1) y_{r{+}1} + (4\beta b_j + 2b_j - 2\beta  - 1)(\vec{\tilde a}\trans\!\! \ast\!\!\; \vec y)
    \notag\\&\equiv\;
        \vec y\trans \!Q' \vec y \,+\, (2\beta+1) y_{r{+}1} + ( 2b_j + 2\beta - 1) (\vec y\trans\!\vec{\tilde a}\;\!\vec{\tilde a}\trans\! \vec y) \mod{4}
        \mspace{-9mu}
    \notag\\&=\;
        \vec y\trans\!Q^{(Y)} \vec y\;.
    \end{align}%
\end{subequations}
Then, if we let $g^{(X)} := g - 4 \beta b_j$ and $g^{(Y)} := g - (4\beta +2)b_j$, we may re-express Eqns.~\eqref{eqn:QFE-XYmeas-2} as
\begin{equation}
\label{eqn:QFE-XYmeas-3}%
{}\mspace{-9mu}
        \ket{\psi^{(X)}_w}
    \,=\,
        \frac{\tau^{g^{(X)}}}{\sqrt{2^{r{+}1}}}
        \!\!\!\!
        \sum_{\vec y \in \{0,1\}^{r{+}1}\!\!}
        \!\!\!\!\!\!
            i^{\;\! \vec y\trans \! Q^{(X)} \vec y}
            \ket{A'' \vec y \oplus \vec b'} ,
    \qquad
        \ket{\psi^{(Y)}_w}
    \,=\,
        \frac{\tau^{g^{(Y)}}}{\sqrt{2^{r{+}1}}}
        \!\!\!\!
        \sum_{\vec y \in \{0,1\}^{r{+}1}\!\!}
        \!\!\!\!\!\!
            i^{\;\! \vec y\trans \! Q^{(Y)} \vec y}
            \ket{A'' \vec y \oplus \vec b'} .
    \mspace{-6mu}
\end{equation}%

As with simulating the Hadamard, $A''$ may not be in principal row form.
Furthermore, as $A''$ is constructed in the same way as in the analysis of the Hadamard operation, we may treat this possibility in exactly the same way:
\begin{itemize}
\item 
    If $j$ is not a principal row of $A$, we may apply the above analysis without modifications, and extend $p$ to obtain a principal index map for $A''$ by setting $p(r{+}1) = j$.
\item
    If $j = p(c)$ for some column $1 \le c \le r$, we may attempt to reduce to the case where $j$ is not a principal row, by invoking $\algname{ReselectPrincipalRow}(j,c)$.
    If afterwards $p(c) \ne j$, then the expansion $A$ has been transformed so that $j$ is not a principal row, and may proceed as above.
    Otherwise, if $\algname{ReselectPrincipalRow}(j,c)$ does not change the value of $p(c)$, it must be the case that $j$ is the only row in which some column $c$ of $A$ is non-zero, row $j$ of $A''$ would be $\vec e_{r{+}1}\trans$ rather than $\vec e_c\trans$\,, and column $c$ would be entirely zero.
    We may then setting $p'(r+1) = j$, and invoke $\algname{ZeroColumnElim}(c)$ to transform the quadratic form expansion into a form where the corresponding matrix $A'$ is in principal row form.
\end{itemize}

Figure~\ref{alg:SimulateMeasXY} presents  algorithms $\algname{SimulateMeasY}(j)$ and $\algname{SimulateMeasX}(j)$, which supplement the analysis of deterministic measurements in Section~\ref{sec:simulDeterministicMeas} with the operations described in the analysis above.
\begin{figure}[!t]%
    \centering\small
    ~\!\!\!\!\!\!
    \begin{minipage}{.4875\textwidth}
        \rule{\textwidth}{1pt}
        $\algname{SimulateMeasX}(j)$
        \smallskip
        \hrule
        \smallskip
        \begin{itshape}\small 
            Simulate an $X$-basis measurement on qubit $j$, storing the result in a bit $\beta$ produced as output.%
        \end{itshape}
        \smallskip
        \hrule
        \begin{enumerate}[leftmargin=4ex,itemsep=.25ex]
        \item
            If $j$ is a principal row of $A$, let $1 \le c \le r$ be such that $j = p(c)$; otherwise let $c = 0$.
        \item
            If $c > 0$, call $\algname{ReselectPrincipalRow}(j,c)$.
            If afterwards $j \ne p(c)$, update $c \gets 0$.
        \item            
            If $c = 0$, or if row $c$ of $Q$ has a non-zero entry off of the diagonal, select $\beta \in \{0,1\}$ uniformly at random and proceed.
            Otherwise:
            \begin{itemize}[topsep=0ex, leftmargin=4ex]
            \item
                If $Q_{c,c} = 2u$ for $u \!\in\! \{0,1\}$, set $\beta \!=\! u$ and stop;
            \item
                Otherwise, select $\beta \in \{0,1\}$ uniformly at random, set $Q_{c,c} \gets 2\beta$, and stop.
            \end{itemize}
        \item
            Let $\vec{\tilde a} = {[\,A_{j,1}\;\;\cdots\;\;A_{j,r}\;\;0\,}]\:\!\trans \in \{0,1\}^{r{+}1}$.
        \item
            Modify $A$ by updating $A_{j,k} \gets 0$ for each $1 \le k \le r$; 
            then extend $A$ by adjoining the column vector $\vec e_j$ and update $p(r{+}1) \gets j$.
        \item
            Modify $Q$ by extending by one row and one column, initially set to zero.
        \item 
            Update ${Q \gets Q + 2\beta \;\! \mathrm{diag}(\vec{\tilde a}\trans \!\!+\!\!\; \vec e_{r{+}1}\trans)}$.
        \item
            For each $1 \le k \le r$ such that $\tilde a_k \ne 0$: ~\\
                reduce $Q_{k,k}$ modulo $4$.
        \item
            Update $b_j \gets 0$.
        \item
            If $c > 0$, call $\algname{ZeroColumnElim}(c)$; \\ otherwise update $r \gets r + 1$.
        \end{enumerate}
        \vspace*{-3ex}
        \rule{\textwidth}{1pt}
    \end{minipage}
    \hfill\hfill
        \begin{minipage}{.5\textwidth}
        \rule{\textwidth}{1pt}
        $\algname{SimulateMeasY}(j)$
        \smallskip
        \hrule
        \smallskip
        \begin{itshape}\small 
            Simulate a $Y$-basis measurement on qubit $j$, storing the result in a bit $\beta$ produced as output.%
        \end{itshape}
        \smallskip
        \hrule
        \begin{enumerate}[leftmargin=4ex,itemsep=.25ex]
        \item
            If $j$ is a principal row of $A$, let $1 \le c \le r$ be such that $j = p(c)$; otherwise let $c = 0$.
        \item
            If $c > 0$, call $\algname{ReselectPrincipalRow}(j,c)$.
            If afterwards $j \ne p(c)$, update $c \gets 0$.
        \item            
            If $c = 0$, or if row $c$ of $Q$ has a non-zero entry off of the diagonal, select $\beta \in \{0,1\}$ uniformly at random and proceed.
            Otherwise:
            \begin{itemize}[topsep=0ex, leftmargin=4ex]
            \item
                If $Q_{c,c} \!=\! 2u{+}1$ for $u \!\in\! \{0,1\}$, set $\beta \!=\! u {\!\:\oplus\!\:} b_j$ and stop;
            \item
                Otherwise, select $\beta \in \{0,1\}$ uniformly at random, set $b_j \gets 0$ and $Q_{c,c} \gets 2\beta+1$, and stop.
            \end{itemize}
        \item
            Let $\vec{\tilde a} = {[\,A_{j,1}\;\;\cdots\;\;A_{j,r}\;\;0\,}]\:\!\trans \in \{0,1\}^{r{+}1}$.
        \item
            Modify $A$ by updating $A_{j,k} \gets 0$ for each $1 \le k \le r$; 
            then extend $A$ by adjoining the column vector $\vec e_j$ and update $p(r{+}1) \gets j$.
        \item
            Modify $Q$ by extending by one row and one column, initially set to zero.
        \item
            Update ${Q \gets Q + (2b_j {+} 2\beta {-} 1)\vec{\tilde a}\:\!\vec{\tilde a}\trans\! + (2\beta{+}1) \vec e_{r\!\!\;{+}\!\!\;1}\vec e_{r\!\!\;{+}\!\!\;1}\trans}$\;.
        \item
            For each $1 \le k \le r$ such that $\tilde a_k \ne 0$: ~\\
                call $\algname{ReduceGramRowCol}(k)$.
        \item
            Update $b_j \gets 0$.
        \item
            If $c > 0$, call $\algname{ZeroColumnElim}(c)$; \\ otherwise update $r \gets r + 1$.
        \end{enumerate}
        \vspace*{-3ex}
        \rule{\textwidth}{1pt}
    \end{minipage}
    \!\!\!\!\!\!~
\caption{%
    \label{alg:SimulateMeasX}%
    \label{alg:SimulateMeasY}%
    \label{alg:SimulateMeasXY}%
    Procedures to modify a quadratic form expansion, to represent the effects of $X$-basis or $Y$-basis measurements.}
\end{figure}
The treatment of deterministic measurements is in each case captured by the first three steps, which determines whether $j$ is a principal row $j = p(c)$, attempts to reselect it if so (thereby testing whether the column $c$ has more than one non-zero entry), and then tests whether the conditions for qubit $j$ to be unentangled from the others.
Apart from Step~3, and the modifications made to the Gram matrix $Q$, both procedures are essentially identical to each other and to $\algname{SimulateH}(j)$ as presented in Figure~\ref{alg:SimulateH}.

\begin{lemma}
\label{lemma:SimulateXYMeas}
    Let $\ket{\psi}$ be represented by a quadratic form expansion as in Eqn.~\eqref{eqn:QFE-specific}, and $1 \le j\le n$.
    Then:
    \begin{itemize}
    \item
        $\algname{SimulateMeasX}(j)$ computes a quadratic form expansion for a post-measurement state arising from performing a $X$-basis measurement on the state $\ket{\psi}$ in time ${\mathcal{O}(s_j + tw + w^2)} \subseteq \mathcal{O}(nr)$ in general, generating a uniformly random outcome if required; this may be tightened to $\mathcal{O}(s_j + 1)$ in case $j$ is not a principal row of $A$.
    \item
        $\algname{SimulateMeasY}(j)$ computes a quadratic form expansion for a post-measurement state arising from performing a $Y$-basis measurement on the state $\ket{\psi}$ in time ${\mathcal{O}(s_j^2 + s_jw + tw + w^2)} \subseteq \mathcal{O}(nr)$ in general, generating a uniformly random outcome if required; this may be tightened to $\mathcal{O}(s_j^2 + s_j w + 1)$ in case $j$ is not a principal row of $A$.
    \item
        If $\ket{\psi}$ is an eigenstate of the Pauli $X$, $Y$ or $Z$ operators, both procedures terminate in time $\mathcal{O}(1)$, and do not modify their input if the outcome is deterministic.
    \end{itemize}
\end{lemma}
\begin{proof}
    Step 1 takes $\mathcal{O}(1)$ operations, and if $j$ is not a principal row, so does Step~2.
    If $j = p(c)$ is a principal row, Step~2 instead invokes $\algname{ReselectPrincipalRow}(j,c)$, which takes time $\mathcal{O}(t_c)$.
    Note that in the case that qubit $j$ is not entangled with any other qubit $\ket{\psi}$, we have $t_c = 1$.
    If this succeeds in choosing a different principal row, then Step~3 again takes $\mathcal{O}(1)$ operations.
    Otherwise, Step~3 tests for the special case that qubit $j$ is an eigenstate of $X_j$ or $Y_j$ by testing whether $Q$ has any non-zero off-diagonal elements in row $c$.
    If not, it performs special operations to allow for a fast run-time in this special case:
    \begin{itemize}
    \item
        $\algname{SimulateMeasX}(j)$ checks whether the outcome is deterministic, which occurs when $Q_{c,c} = 2u$ for $u \in \{0,1\}$.
        If so, it sets $\beta = u$ as the measurement outcome; otherwise it selects a random measurement outcome, and updates $Q_{c,c} \gets 2\beta$ to represent the corresponding post-measurement state.
    \item
        $\algname{SimulateMeasY}(j)$ checks whether the outcome is deterministic, which occurs when $Q_{c,c} = 2u+1$ for $u \in \{0,1\}$.
        If so, it sets $\beta = u$ as the measurement outcome; otherwise it selects a random measurement outcome, and updates $Q_{c,c} \gets 2\beta+1$ to represent the corresponding post-measurement state.
    \end{itemize}
    In each case above, the procedures then stop, having taken $\mathcal{O}(1)$ time.
    If instead $Q$ does have non-zero off-diagonal elements, both procedures instead select a random measurement outcome $\beta$, and the procedure continues.
    
    In Step~4, we initialise a vector $\vec{\tilde a}$ from row $j$ of $A$, taking time $\mathcal{O}(s_j)$; Step~5 then modifies row $j$ of $A$ by setting all of its entries to $0$, except for an entry $1$ in a new $(r+1)\textsuperscript{st}$ column, which may be done at the same time as Step~4.
    In Step~6, both procedures extend $Q$ by an additional row and column of zeroes, which may in principle be done without any modification of the sparse data structure storing $Q$.
    Then Steps~7 and~8 realise different updates to the Gram matrix according to the measurement type:
    \begin{itemize}
    \item
        $\algname{SimulateMeasX}(j)$ adds up to $s_j + 1$ entries to the diagonal, and then reduces them modulo~4, taking time $\mathcal{O}(s_j)$;
    \item
        $\algname{SimulateMeasX}(j)$ adds two terms to $Q$, together having precisely $s_j^2 + 1$ non-zero entries.
        It then invokes $\algname{ReduceGramRowCol}(k)$ on all rows $1 \le k \le r$ in which those terms have non-zero entries, taking time $\mathcal{O}(s_j w)$.
    \end{itemize}
    After this, Step~9 takes time $\mathcal{O}(1)$.
    If row $j$ was not initially a principal row of $A$, or if Step~2 succeeded in selecting a new principal row, then Step~10 simply updates $r$ in time $\mathcal{O}(1)$; otherwise, column $c$ is now zero and must be removed using $\algname{ZeroColumnElim}(c)$, taking time $\mathcal{O}(tw + w^2)$.
    
    In the above, if $j$ was not a principal column, the two procedures $\algname{SimulateMeasX}(j)$ and $\algname{SimulateMeasX}(j)$ respectively take time $\mathcal{O}(s_j + 1)$ and $\mathcal{O}(s_j^2 + s_j w + 1)$.
    Note that, if row $j$ of the initial expansion matrix is entirely zero, this is $\mathcal{O}(1)$ in both cases, matching the run-time for the case that qubit $j$ is initially in an $X$- or $Y$-basis state.
    If instead $j$ is a principal row of the original expansion matrix $A$, and qubit $j$ is entangled with some other qubits in $\ket{\psi}$, the run-times of these subroutines takes an additional amount of time which is domiated by Step~10, with respective totals of $\mathcal{O}(s_j + tw + w^2)$ and $\mathcal{O}(s_j^2 + s_j w + t w + w^2)$.
\end{proof}

\subsection{Summary of simulation complexities of stabiliser operations}
\label{runtimes}

Ignoring refinements which are possible when the expansion matrix $A$ or Gram matrix $Q$ are sparse, we may summarise the run-time bounds for the subroutines presented above as follows:

\vspace*{-1ex}
\paragraph{Pauli operations ---\!\!\!}
    Pauli $X$ operations may be simulated in $\mathcal{O}(1)$ time, and both $Y$ and $Z$ operations may be simulated in $\mathcal{O}(r)$ time, governed by the number of non-zero entries in the row of the expansion matrix $A$ corresponding to the qubit which these operators act upon.

\vspace*{-1ex}
\paragraph{Diagonal operations ---\!\!\!}
    The diagonal $S$ and $\mathrm CZ$ operations may both be simulated in $\mathcal{O}(r^2)$ time, governed by the number of non-zero entries in the row(s) of the expansion matrix $A$ for the qubit(s) which these operations act upon.

\vspace*{-1ex}
\paragraph{Hadamard and controlled-NOT operations ---\!\!\!}
    The Hadamard operation may be simulated in time $\mathcal{O}(nr)$, or in time $\mathcal{O}(r)$ when the qubit it acts on does not correspond to a principal row of the expansion matrix $A$.
    Similarly, the controlled-NOT operation may be simulated in time $\mathcal{O}(nr)$, or in time $\mathcal{O}(r)$ when its target qubit does not correspond to a principal row of $A$.
    In both cases, the $\mathcal{O}(nr)$ bound when acting on qubits corresponding to principal rows, arises from performing a change of variables to maintain the expansion matrix in principal row form; the factor of $n$ in particular arises from an $n - r + 1$ bound on the number of non-zero entries in each column.
    
\vspace*{-1ex}
\paragraph{Pauli observable measurements ---\!\!\!}
The result of measuring a qubit $j$ which is unentangled from any others can be computed in time $\mathcal{O}(1)$.
In particular, any single-qubit measurement with a deterministic outcome can be simulated in constant time.
Apart from this special case, each of these measurement operations can be simulated in time $\mathcal{O}(nr)$ in general, dominated by the time required to maintain the principal row form; $X$- and $Y$-basis measurements can be performed more efficiently (in time $\mathcal{O}(r)$ and $\mathcal{O}(r^2)$ respectively) when performed on a qubit which \emph{does not} correspond to a principal row of $A$, and $Z$-basis measurements can be performed more efficiently (in time $\mathcal{O}(n)$) when performed on a qubit which \emph{does} correspond to a principal row of $A$.

\bigskip\noindent
In each case, $r$ is bounded above by $n$ as a result of maintaining the principal row form.
If we abandon this requirement, for example to allow controlled-NOT gates to be universally simulatable in time $\Theta(r)$, then we must do more work when performing a Hadamard or $X$- or $Y$-basis measurement  (potentially up to $\mathcal{O}(n^3)$ to perform Gaussian elimination).

\section{Simulation complexity of composite procedures}
\label{sec:complexity-procedures}

We now prove a number of results regarding the complexity of simulating procedures consisting of multiple stabiliser operations.

\subsection{Simulating stabiliser circuits in general}
\label{sec:generalCircuitSim}

The above Lemmata allow us to show the following result, which is the most general summary regarding the effectiveness of our techniques for weak simulation of stabiliser circuits:
\begin{thm}
\label{thm:weakSimulationAsymptotic}
A stabiliser circuit consisting of $M$ stabiliser operations from the set of operations described in Section~\ref{sec:elementaryStabiliserOpns}, whose initial state is expressed as a quadratic form expansion with an expansion matrix $A$ in principal row form, can be weakly simulated in $\mathcal{O}(Mn^2)$ operations.
\end{thm}
\begin{proof}
    As each of the $M$ stabiliser operations in the circuit can be simulated in time $\mathcal{O}(nr) \subseteq \mathcal{O}(n^2)$ by Lemmas~\ref{lemma:simulatePaulis}--\ref{lemma:SimulateXYMeas}, the Theorem follows.
\end{proof}

\noindent
This asymptotic result in principle matches the worst-case performance of each of the stabiliser formalism~\cite{Aaronson2004}, graph-state-based representations~\cite{Anders2006}, and the phase-sensitive Clifford simulator of Bravyi~\emph{et al}~\cite{Bravyi2019}. 
While each of those other techniques have faster asymptotic performance for certain operations, the worst-case performance described in Theorem~\ref{thm:weakSimulationAsymptotic} also obscures faster performance of different operations under various conditions of sparsity.

In common with any weak simulation technique, the result above may easily be `upgraded' in certain cases to a result for \textit{strong} simulation (in which we are interested in the probability of obtaining certain outcomes for a subset of the measurement operations).
Specifically: suppose that we are interested in the probability that some $k$ of the measurements yield an outcome $\boldsymbol \beta \in \{0,1\}^k$.
Let us say that a measurement on a qubit $j$ is `of interest' if it is one of the measurements whose outcomes we consider.
We may recursively define the `history' of such a measurement --- representing a sort of `causal past' of the measurement --- as:
\begin{enumerate}[label=\parit{\roman*}]
\item
    any operation which is performed on the measured qubit $j$, prior to that measurement of interest, or
\item
    an operation (possibly classically controlled) which acts on some other qubit $j'$, before an operation which also acts on $j'$ and is also in the history of the measurement of interest on $j$;
\item
    a measurement operation whose outcome is used as a classical control, for an operation in the history of the measurement of interest on $j$. 
\end{enumerate}
If every measurement in the history of some measurement of interest, is itself a measurement of interest, we may compute the probability of the outcome $\boldsymbol \beta$ by sequentially computing the probability that each measurement of interest produces a particular outcome.
Without imposing this constraint, the strong simulation problem may in some cases be tractible, but in general is as difficult as strong simulation of a Hadamard+Toffoli circuit, which is $\textbf{\texttt\#P}$-hard~\cite[Theorem~2]{JozsaClassical}.

\subsection{Improved simulation of multiple `terminal' measurements}
\label{sec:parallelMeasurements}

Quadratic form expansions allow for more efficient techniques to strongly simulate `terminal' measurements: that is, measurements on distinct qubits, which are the last measurements of interest to be performed on a state (so that the post-measurement state is not needed to compute any other outcomes).
For the sake of simplicity, we may consider the case where all of the measurements to be performed are $Z$-basis measurements, by reducing $X$-basis and $Y$-basis measurements to $Z$-basis measurements preceded (and followed) by appropriate single-qubit unitaries.
We may consider improvements in the run-time for strong simulation of the measurement stage itself, by reduction to solving a system of equations involving the expansion matrix $A$ (or a sub-matrix of $A$) for the state just prior to measurement.

\begin{thm}
\label{thm:strongSubsetMeasurement}
Let $\ket{\psi}$ be represented by a quadratic form expansion as in Eqn.~\eqref{eqn:QFE-specific}, on which we perform $0 \le k \le n$ measurements in the $Z$-basis on distinct qubits.
We may determine the probability of the outcomes of those measurements yielding a particular string $\boldsymbol \beta \in \{0,1\}^k$ in time $\mathcal{O}(k^2 n)$.
\end{thm}

\begin{proof}
Consider the sub-matrix $A'$ of $A$, obtained by restricting to the rows of $A$ corresponding to the qubits whose outcomes we are interested in.
This gives rise to a system of $k$ equations in $r$ unknowns $A' \vec x' = (\boldsymbol \beta \oplus \vec b)$, which may be solved by Gaussian elimination in time $\mathcal{O}(k^2 r) \subseteq \mathcal{O}(k^2 n)$.
If this system of equations is unsatisfiable, then $\boldsymbol \beta$ is not a possible measurement outcome.
Otherwise, the probability with which $\boldsymbol \beta$ occurs is given by $2^{-r'}$, where $r' \ge 0$ is the rank of $A'$ (\emph{i.e.},~the number of bits required to determine the outcome $\boldsymbol \beta$ among the possible outcomes on the measured qubits).
\end{proof}

\noindent
We would often expect better performance than $\mathcal{O}(k^2 n)$ for strong simulation of the measurements in practise.
An elementary observation is that Gaussian elimination in fact takes time $\mathcal{O}(k^2 r)$, for the rank parameter $r$ which is itself merely bounded by $n$; and in fact this may be easily sharpened to ${\mathcal{O}(\kappa \!\: k  \!\: r)}$, where $\kappa = {\min\,\{k,r\}}$, as $\kappa$ is a bound on the rank of $A'$.
Furthermore, in the case $k \approx n$, we would expect the problem to become simpler due to the involvement of principal rows corresponding to some of the measured qubits.
The $k$ qubits whose measurement outcomes we are interested in, will include some number $0 \le k_p \le k$ of qubits corresponding to principal rows of the expansion matrix $A$.
The presence of such rows may be used to speed up the computation of a reduced row-echelon form for $A'$, as the corresponding rows of $A'$ have only one non-zero entry to account for in row reduction.%
    \footnote{%
        One may alternatively consider a pre-processing stage in which one uses back-substitution of the values of the indices $x_j$ for such principal rows, fixing them to the corresponding entry in $\boldsymbol \beta \oplus \vec b$, reducing the number of rows and columns of the matrix $A'$ to be considered in doing so.
    }
The complexity in this case would then be $\mathcal{O}(k_p k + \tilde \kappa \!\;\tilde k \!\;\tilde r)$, where $\tilde k = k - k_p$, $\tilde r = r - k_p$, and $\tilde \kappa = \min\,\{\tilde k, \tilde r\}$.
When $k_p = k$ (\emph{i.e.},~all of the measured qubits correspond to principal rows of $A$) or when $k_p = r \le k$ (\emph{i.e.},~when all of the principal rows of $A$ correspond to measured qubits), this leads to $\tilde \kappa = 0$, providing a potentially significant drop in run-time complexity.
In particular:

\begin{thm}
\label{thm:strongTotalMeasurement}
Let $\ket{\psi}$ be represented by a quadratic form expansion as in Eqn.~\eqref{eqn:QFE-specific}, on which we measure all but $0 \le \ell \le n$ of the qubits in the $Z$-basis.
We may determine the probability of the outcomes of those measurements yielding a particular string $\boldsymbol \beta \in \{0,1\}^{n-\ell}$ in time $\mathcal{O}(n^2 + \ell^2 n)$.
\end{thm}
\begin{proof}
    Following the analysis above, take $k = n - \ell$ and $k_p \ge r - \ell$, so that $\tilde k = k - k_p \le n - r$, and $\tilde r = r - k_p \le \ell.$
    Then $\tilde \kappa \le \ell$ as well.
    By definition, we have $k_p k \le (n-\ell)r = nr - \ell r$; then strong simulation of ${n - \ell}$ qubits can be done in time ${\mathcal{O}(k_p k + \tilde \kappa \;\! \tilde k \;\! \tilde r)}
    = \mathcal{O}(nr - \ell r + \ell^2 n - \ell^2 r) \subseteq \mathcal{O}(n^2 + \ell^2 n)$.
\end{proof}

\noindent
Thus, for circuits with a final round of $k$ parallel $Z$-basis measurements where $k \in \mathcal{O}(1)$ or $k = n - \mathcal{O}(\sqrt{n})$, we may perform a strong simulation of these measurements with an amortised cost of $\mathcal{O}(n)$ per qubit, whether or not the outcome is deterministic.%
\footnote{%
    In this setting of total final measurements, our results may be compared to those of Guan and Regan~\cite{GuanRegan2019}: we remark on this comparison in Section~\ref{disc}.
    }
In this spirit, we add a further observation on performing a round of near-total measurement, for weak simulation:
\begin{thm}
\label{thm:weakTotalMeasurement}
Let $\ket{\psi}$ be represented by a quadratic form expansion as in Eqn.~\eqref{eqn:QFE-specific}, on which we measure all but $0 \le \ell \le n$ of the qubits in the $Z$-basis.
Then these measurements may be weakly simulated in time $\mathcal{O}((\ell{+}1) n^2)$.
\end{thm}
\begin{proof}
    For each of the $k_p$ columns $c$, for which the principal row $j = p(c)$ corresponds to a measured qubit, we swap column $c$ with one of the last $k_p$ columns in order, using $\algname{ReindexSwapColumns}$.
    Doing this for $k_p$ such columns will take time $\mathcal{O}(k_p n)$ in total.
    Having done so, we may simulate the measurements on the corresponding principal rows by selecting measurement outcomes at random and performing $\algname{FixFinalBit}$ a total of $k_p$ times, with run-time cost $\mathcal{O}(k_p n)$ again.
    This leaves us with a quadratic form expansion with $\tilde r = r - k_p$ columns: we may simulate the final $\tilde k = k - k_p$ measurements by repeatedly calling $\algname{SimulateMeasZ}$, taking time ${\mathcal{O}(\tilde k \;\! \tilde r \;\! n)}$.
    The total run-time of this procedure is then ${\mathcal{O}(k_p n + \tilde k \;\! \tilde r \;\! n)}$.
    If $k = n - \ell$ and $k_p \ge r - \ell$, we then have $\tilde k \le n - r$ and $\tilde r \le \ell$, so that $k_p n \le n^2$ and $\tilde k \;\! \tilde r \;\! n \le \ell n^2$.
\end{proof}
\noindent
Thus, for $\ell \in \mathcal{O}(1)$, we have an amortised cost of $\mathcal{O}(n)$ for weak simulation of $n - \ell$ parallel $Z$-basis measurements, regardless of whether the outcomes are deterministic; more generally, for any $0 < \ell \ll n$, we have an amortised cost of $\mathcal{O}(\ell n)$ for parallel $Z$-basis measurements.

\subsection{Simulation of stabiliser measurements}
\label{stabmeas}

A particularly useful feature of our techniques is that any single-qubit measurement whose outcome is in principle deterministic, can be simulated in constant time.
This is potentially significant for an important use-case of stabiliser circuit simulation: prototyping and simulating error correction procedures under the influence of a Pauli noise model.

The most prominent error correction procedures are stabiliser codes, whose syndrome measurement procedures correspond to measuring \emph{multi-qubit} Pauli operators as observables, \emph{e.g.},~measuring ${Z {\otimes} Z {\otimes} Z {\otimes} Z}$ as an observable.
In practice, such observable measurements would be performed indirectly, by interacting a `syndrome qubit' with some $k > 1$ `code' qubits using Clifford operations, and then measuring the data qubit.
These syndrome qubits are initially prepared independently of the other qubits in the computation.
Using a quadratic form expansion representation, these syndrome qubits would therefore correspond to rows and columns of $A$ and $Q$ which have a constant number of non-zero elements.
By taking account of this fact, we may show that the syndrome measurements could be simulated in time $\mathcal{O}(kn)$ on an arbitrary state of the rest of the system (regardless of whether the other rows and columns of $A$ and $Q$ are sparse):
\begin{thm}
\label{thm:stabiliserMeasurement}
    Let $\ket{\psi}$ be represented by a quadratic form expansion as in Eqn.~\eqref{eqn:QFE-specific}, which contains at least one `syndrome' qubit, disentangled from the others, which is set aside for the purpose of facilitating multi-qubit measurements.
    Suppose that $\ket{\psi}$ is a $\pm 1$-eigenvector of some multi-qubit Pauli operator $P$ acting on $1 \le k \le n$ qubits.
    Then the outcome of a $P$-measurement on $\ket{\psi}$ can be computed in time $\mathcal{O}(kn)$.
\end{thm}
\noindent
We prove this below.
In fact, we will demonstrate something stronger: \parit{i}~that a $P$-measurement may be simulated in $\mathcal{O}(kn)$ time, but also possibly in $\mathcal{O}(kr)$ time or $\mathcal{O}(k)$ time in situations which would be easy to identify; and \parit{ii}~that this remains true for simulations of certain fault-tolerant procedures to realise a $P$-measurement.
In a stabiliser code such as surface or colour codes~\cite{Fowler2012, landahl2011, Terhal2015} --- in which the stabilisers to be measured are `local', in the sense that $k$ is bounded by some constant --- the above run-time bounds may be tightened further to $\mathcal{O}(n)$, $\mathcal{O}(r)$, and $\mathcal{O}(1)$.

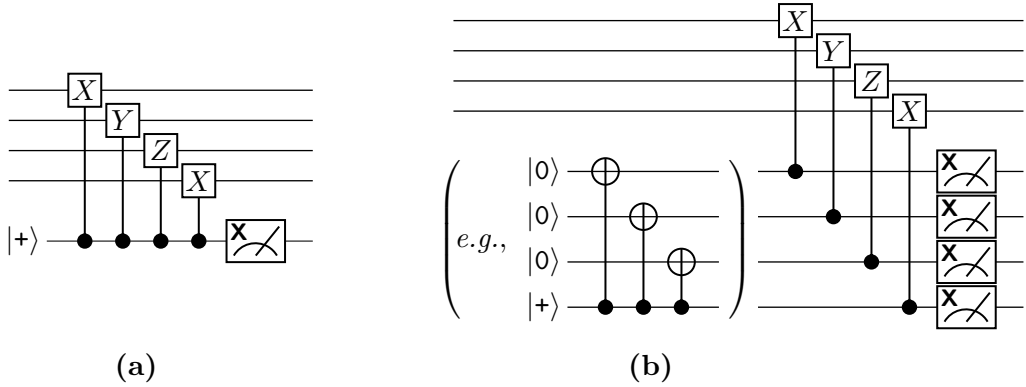
\begin{figure}[t]
    \begin{align*}
        \begin{gathered}
        \begin{tikzpicture}
            \foreach \j in {0,...,3} {%
                \coordinate (x\j) at (0,{-\j*.4});
                \draw [line width=0.5pt]  (x\j) -- ++(4,0);
            }
            \coordinate (a) at (0,-2);
            \draw ($(a) + (.5,0)$) -- ++(3.5,0);
            \node at ($(a) + (.1875,0)$) {$\lvert \texttt+ \rangle$};
            \foreach \j in {0,...,3} {%
                \filldraw [line width=.75pt] ($(x\j) + ({1+\j/2},5pt)$) -- ($(a) + ({1+\j/2},0)$) circle (2.5pt);
            }
            \node [line width=.75pt,draw=black,fill=white,inner ysep=2.5pt, inner xsep=-8pt,minimum width=1.125em] 
                    at ($(x0) + (1.0,0)$) {$X$};
            \node [line width=.75pt,draw=black,fill=white,inner ysep=2.5pt, inner xsep=-8pt,minimum width=1.125em] 
                    at ($(x1) + (1.5,0)$) {$Y$};
            \node [line width=.75pt,draw=black,fill=white,inner ysep=2.5pt, inner xsep=-8pt,minimum width=1.125em] 
                    at ($(x2) + (2.0,0)$) {$Z$};
            \node [line width=.75pt,draw=black,fill=white,inner ysep=2.5pt, inner xsep=-8pt,minimum width=1.125em] 
                    at ($(x3) + (2.5,0)$) {$X$};
        \node [
            draw=black, fill=white, inner sep=2pt, label distance=-5mm, minimum height=1.45em, minimum width=2em, line width=.75pt
        ] (meas) at ($(a) + (3.25,0)$) {};
        \draw [line width=.75pt] ($(meas.south) + (-.75em,1mm)$) arc (150:30:.85em);
        \draw  [line width=.75pt] ($(meas.south) + (0,1mm)$) -- ++(.8em,1em);
        \node [
            anchor=north west, inner sep=1.5pt, font=\sffamily\bfseries\footnotesize
        ] at (meas.north west) {X};
        \end{tikzpicture}
        \end{gathered}
    &\qquad&
        \begin{gathered}
        \begin{tikzpicture}
            \foreach \j in {0,...,3} {%
                \coordinate (x\j) at (0,{-\j*.4});
                \draw [line width=0.5pt] (x\j) -- ++(7.5,0);
            }
            \foreach \j in {0,...,3} {%
                \coordinate (a\j) at (0,{-2-\j*.6});
                \draw [line width=0.5pt] ($(a\j) + (1.5,0)$) -- ++(6,0);
                \ifnum\j<3
                    \node at ($(a\j) + (1.1875,0)$) {$\lvert \texttt0 \rangle$};
                \else
                    \node at ($(a\j) + (1.1875,0)$) {$\lvert \texttt+ \rangle$};
                \fi
            }
            \foreach \j in {0,1,2} {%
                \filldraw [line width=.75pt] ($(a\j) + ({2+\j/2},5pt)$) -- ($(a3) + ({2+\j/2},0)$) circle (2.5pt);
                \draw [line width=.75pt] ($(a\j) + ({2+\j/2},0)$) circle (5pt);
            }
            \foreach \j in {0,...,3} {%
                \filldraw [line width=.75pt] ($(x\j) + ({4.5+\j/2},5pt)$) -- ($(a\j) + ({4.5+\j/2},0)$) circle (2.5pt);
            }
            \node [line width=.75pt,draw=black,fill=white,inner ysep=2.5pt, inner xsep=-8pt,minimum width=1.125em] 
                    at ($(x0) + (4.5,0)$) {$X$};
            \node [line width=.75pt,draw=black,fill=white,inner ysep=2.5pt, inner xsep=-8pt,minimum width=1.125em] 
                    at ($(x1) + (5.0,0)$) {$Y$};
            \node [line width=.75pt,draw=black,fill=white,inner ysep=2.5pt, inner xsep=-8pt,minimum width=1.125em] 
                    at ($(x2) + (5.5,0)$) {$Z$};
            \node [line width=.75pt,draw=black,fill=white,inner ysep=2.5pt, inner xsep=-8pt,minimum width=1.125em] 
                    at ($(x3) + (6.0,0)$) {$X$};
                \filldraw [white] ($(a0) + (3.5,0.25)$) -- ++(.5,0) -- ($(a3) + (4,-0.25)$) -- ++(-.5,0) -- cycle;
                \node at ($(a0)!0.5!(a3) + (0.25,0)$)
                    {$\left(\begin{matrix} \\[18mm] \end{matrix}\right.$\!\!\emph{e.g.},};
                \node at ($(a0)!0.5!(a3) + (3.75,0)$)
                    {$\left.\begin{matrix} \\[18mm] \end{matrix}\right)$};
        \foreach \j in {0,1,2,3} {
        \node [
            draw=black, fill=white, inner sep=2pt, label distance=-5mm, minimum height=1.45em, minimum width=2em, line width=.75pt
        ] (meas) at ($(a\j) + (6.75,0)$) {};
        \draw [line width=.75pt] ($(meas.south) + (-.75em,1mm)$) arc (150:30:.85em);
        \draw  [line width=.75pt] ($(meas.south) + (0,1mm)$) -- ++(.8em,1em);
        \node [
            anchor=north west, inner sep=1.5pt, font=\sffamily\bfseries\footnotesize
        ] at (meas.north west) {X};
        }
        \end{tikzpicture}
        \end{gathered}
    \\[-.5ex]
        \textbf{(a)}\mspace{96mu} && \textbf{(b)}\mspace{216mu}
    \end{align*}
    \vspace*{-4ex}
    \caption{%
        \label{fig:syndromeMeasurement}
            Procedures to perform a syndrome measurement, using auxiliary qubits to measure a Pauli observable --- in this case, an operator $X {\otimes} Y {\otimes} Z {\otimes} X$.
            Under a Pauli noise model, any particular error operations which accrue on the upper four `data' qubits result in a deterministic syndrome bit --- though the procedure to obtain that bit may involve non-deterministic processes.
            These procedures may be simulated in time $\mathcal{O}(n)$ in each case.
            \textbf{(a)}~A~simple procedure for eigenvalue estimation of the Pauli operator to be measured.
            The syndrome qubit, which is initially prepared in the state $\lvert \texttt+ \rangle$, would be simulated with a principal row of an expansion matrix $A$, which in particular has exactly one non-zero entry.
            The controlled-Pauli operations can then be simulated variously in time $\mathcal{O}(1)$, $\mathcal{O}(r)$, or $\mathcal{O}(n)$, and the final $X$-basis measurement can be realised in time $\mathcal{O}(1)$.
            This procedure may be elaborated with a flag qubit \cite{Chao2018} to make it fault-tolerant, adding only a constant amount of work for the simulation.
            \textbf{(b)}~A~procedure to measure the syndrome bit using Shor-style syndrome extraction \cite{Shorfault, DiVincenzo2007}.
            Multiple syndrome qubits are initially prepared in a `cat' state, involving a single principle row.
            For example, the preparation of this state could be simulated using non-principal rows representing qubits initially prepared in the state $\lvert \texttt0 \rangle$ for all but one syndrome qubit; these qubits are entangled with a final qubit prepared in the state $\lvert \texttt+ \rangle$, corresponding to a principal row.
            This preparation time can be simulated in time $\mathcal{O}(1)$, as they involve only a constant number of qubits independent of the larger system.
            We interact these syndrome qubits with the data qubits to realise a distributed eigenvalue estimation procedure: each such interaction may be simulated in time $\mathcal{O}(1)$, $\mathcal{O}(r)$, or $\mathcal{O}(n)$.
            All but the last of the $X$-basis measurements involve a non-principal row with only one non-zero element, and so can be simulated in time $\mathcal{O}(1)$; the final syndrome measurement produces the syndrome bit, and is therefore a deterministic measurement also simulated in time $\mathcal{O}(1)$.
            If the simulation of the cat state involves Pauli noise, this noise may be countered with further parity checks~\cite{Shorfault,NandC} or distillation \cite{Bennett1996, Nigmatullin2016, Dur2007}, each round of which may also be simulated in time $\mathcal{O}(1)$ using the same techniques as above.
        }
\end{figure}

Figure~\ref{fig:syndromeMeasurement} demonstrates some typical presentations of syndrome measurement procedures, for a Pauli observable $X{\otimes}Y{\otimes}Z{\otimes}X$ selected as an example.
This measurement is simulated using controlled-$X$, controlled-$Y$, and controlled-$Z$ operations to realise phase-kicks onto one or more syndrome qubits.
(We may decompose a controlled-$Y$ operation as $\mathrm CY = \bigl[\begin{smallmatrix} I & 0 \\ 0 & Y \end{smallmatrix}] = \bigl[\begin{smallmatrix} I & 0 \\ 0 & i I \end{smallmatrix}] \bigl[\begin{smallmatrix} I & 0 \\ 0 & X \end{smallmatrix}] 
\bigl[\begin{smallmatrix} I & 0 \\ 0 & Z \end{smallmatrix}] = (S \otimes I) \, \mathrm CX \, \mathrm CZ$ for the purposes of simulation.)
These operations involve the syndrome qubits as controls: simulating these operations therefore do not affect the number of non-zero elements in the rows of $A$ which correspond to these syndrome qubits, which thus remain $\mathcal{O}(1)$ throughout.
Using the analysis of the simulation runtimes of these operations in terms of sparsity parameters, we may then bound the time to simulate an operation $\mathrm CZ_{a,j}$ by $\mathcal{O}(s_j {+} w) \subseteq \mathcal{O}(r)$, and to simulate an operation $\mathrm CX_{a,j}$ by $\mathcal{O}(t) \subseteq \mathcal{O}(n)$.
(These operations may in fact be simulatable in time $\mathcal{O}(1)$, depending on whether row $j$ of $A$ happens to be a principal row in the case of $\mathrm CZ_{a,j}$ operations, or whether row $j$ happens \emph{not} to be a principal row in the case of $\mathrm CX_{a,j}$.)
Because $S_a$ can be simulated in time $\mathcal{O}(1)$ in this case, $\mathrm CY_{a,j}$ operations can then be simulated in time $\mathcal{O}(s_j + t) \subseteq \mathcal{O}(n)$.
Furthermore, the measurements on these qubits each have either a deterministic outcome, or may be simulated without any reselection of principal rows.
As a result, each of the measurements of the syndrome qubits may be simulated in time $\mathcal{O}(1)$.
The syndrome measurement procedure as a whole then takes time $\mathcal{O}(kn)$, dominated by the two-qubit operations involved; and this bound may be loose, depending on the particular Pauli operation being measured and the qubits on which they are being measured.

\begin{proof}[Proof of Theorem~\ref{thm:stabiliserMeasurement}]
    We prove the result for the non-fault-tolerant technique to realise $P$-measurements by a phase kick with a single syndrome qubit (as illustrated in Figure~\ref{fig:syndromeMeasurement}a): similar remarks may be applied to other techniques.
    Let $a$ be the dedicated `syndrome' qubit.\footnote{%
        While it is not ideal programming practise, it would be enough to have some qubit $a$ which is known to be disentangled from the rest, whether or not it is set aside to simulate multi-qubit measurements.
        The initial state of this qubit $a$ can be read and stored in time $\mathcal{O}(1)$, and restored in time $\mathcal{O}(1)$ after the simulated measurement.}
    In $\mathcal{O}(1)$ time, we may simulate a procedure which prepares $a$ in the state $\ket{\texttt+}$.
    Each of the two-qubit measurements involved in simulating the measurement of $P$ may be simulated in time $\mathcal{O}(n)$ by the analysis above; and the final measurement on $a$ itself may be simulated in time $\mathcal{O}(1)$, by hypothesis that $\ket{\psi}$ is an eigenvector of $P$.
    The total run-time is then dominated by the time required to simulate the $k$ two-qubit operations, which is $\mathcal{O}(kn)$.
\end{proof}

In the case that $k$ is bounded above by a constant, the above allows us to match the asymptotic complexity obtained by Gidney~\cite{gidney2021stim}, of $\mathcal{O}(n)$ to simulate syndrome measurement  under a Pauli noise model.
Our techniques allow us to do so without needing to record the circuit being simulated (in order to simulate the stabiliser measurement in the Heisenburg picture).

We note that for sufficiently large values of $k$, our techniques may be expected to be slower in practise than those of Ref.~\cite{gidney2021stim}, for simulating deterministic Pauli stabiliser measurements \emph{as such}.
This is a result of the fact that our techniques rely on two-qubit unitaries to represent a \emph{particular physical realisation} of a $k$-qubit Pauli observable mearsurement.
We note that such a realisation is necessary for a sufficiently detailed simulation of a physical procedure to realise such $P$-measurements.
The techniques described by Gidney~\cite{gidney2021stim} would also yield $\mathcal{O}(kn)$ run-times to simulate certain procedures for stabiliser measurements, such as that in Figure~\ref{fig:syndromeMeasurement}a; for any which produce non-deterministic measurement outcomes, such as that represented in Figure~\ref{fig:syndromeMeasurement}b, the run-time required by those techniques appear to scale as $\mathcal{O}(kn^2)$.
In this sense, our techniques seem to provide an advantage for simulating proedures which make frequent syndrome measurements in the presence of Pauli noise.

It is plausible that qubits in an error correcting code will have enough structure in their entanglement relations to allow for a relatively sparse representation, allowing for faster simulation than is represented by the unconditional bound of $\mathcal{O}(n)$.
Because of the importance of syndrome measurement to simulation of error corrected architectures, we conjecture that this is an application for which our techniques will be particularly well-suited.

\section{Discussion}
\label{disc}

In this paper, we have developed on the notion of a quadratic form expansion, as described in Ref.~\cite{dBQuadratic} and informed by presentation of similar path-sum like representations of stabiliser states and stabiliser circuits~\cite{CalderbankGoodQEC, dehaene, JBVClassical, VDNClassical, Amy_2019}.
Procedures to efficiently simulate one- or two-qubit stabiliser circuits on $n$-qubit stabiliser states using such a representation, are implicit in many of these works.
We have presented \emph{explicit} procedures to simulate such operations in time $\mathcal{O}(n^2)$, matching the worst-case asymptotic complexity achievable under the stabiliser formalism~\cite{Aaronson2004} among others~\cite{Anders2006,Bravyi2019}.
We obtain this result by considering quadratic form expansions which are subject to certain constraints: notably, involving an `expansion matrix' $A$ in principal row form.
Furthermore, the bound of $\mathcal{O}(n^2)$ is in some cases loose.
As we describe in Section~\ref{runtimes}, when the state has $2^r$ standard basis components for $r \ll n$, the worst-case complexity to simulate a stabiliser operation is $\mathcal{O}(nr)$.
For each stabiliser operation, we present still more refined bounds on simulation complexity, when $A$ or the quadratic form can be represented by sparse data structures.

We briefly consider the way in which our techniques fail to extend to Clifford+T circuits (\emph{i.e.},~including a $T = \sqrt{S} = \mathrm{diag}(1,\tau)$ gate for $\tau = \sqrt{i}$), or Clifford-plus-controlled-S circuits.
\label{discn:Clifford+T}%
Consider a representation of states similar to Eqn~\eqref{eqn:QFE-specific}, in which the relative phases are expressed as powers of $\tau$ instead of powers of $i$.
This would complicate the analysis of the simulation of the Hadamard operation (or more precisely the procedure analogous to \algname{ZeroColumnElim}), requiring the analysis of a new case in which some entry of the Gram matrix $Q$ represents an odd power of $\tau$, to extend the analysis of Eqn.~\eqref{eqn:zeroColumnElim-Q2}.
There is no algebraic `coincidence' regarding $1 + \tau$, which is analogous to the equality  $1 + i = \sqrt{2}\cdot\tau$, which would allow us to reduce the rank $r$ to simplify a quadratic form expansion involving relative phases which are powers of $\tau$; more sophisticated (and computationally demanding) techniques would be required.
In the case of controlled-$S$ gates, any attempt to simulate them directly on a state as in Eqn.~\eqref{eqn:QFE-specific} would require that we abandon the constraint that $Q$ is symmetric; however, this complicates the analysis of Eqns.\eqref{eqref:eq10}--\eqref{eqn:zeroColumnElim-Q2} in \algname{ZeroColumnElim} as well, as well as any analysis involving a change of variables (as Lemma~\ref{lemma:QFE-change-of-variables} relies on the symmetry of $Q$).
The existence of such obstacles, is no surprise: the ability to efficiently simulate either of these circuit classes, suffices to simulate arbitrary quantum computations with bounded error~\cite{NandC, Dawson2005}.
However, extensions of our techniques in this direction could yield modest reductions to the simulation complexity of more complex quantum procedures.

We note that the sparsity of the expansion matrix $A$ and the Gram matrix $Q$ (at least in certain rows and columns) is crucial to the usefulness of our techniques for any given application.
For a random stabiliser circuit --- in which the probability of any one of a Hadamard gate, $X$~measurement, or $Y$~measurement is bounded below by a constant --- one may expect the value of the rank parameter $r$ above to approach $\tfrac{1}{2}n$, and then vary only slightly from this value on average.
In this case of $r \in \Theta(n)$, the bound $\mathcal{O}(nr) = \mathcal{O}(n^2)$ for our techniques to simulate various unitary gates, compares poorly to other simulation techniques~\cite{Aaronson2004,Anders2006}.
However, as $r$ increases, the number of principal rows of $A$ --- each of which has exactly one non-zero entry --- also increases.
By selecting certain qubits to correspond to principal rows of $A$, one may in some cases achieve significant improvements in simulation time for certain procedures, as we describe in Section~\ref{stabmeas}.
For structured stabiliser circuits as opposed to randomly generated ones, this may yield significant advantages for simulation.
We expect this may be the case for simulations of stabiliser circuits which realise operations involving certain error correction codes.

The run-times of our techniques are broadly similar to those described by Anders and Briegel~\cite{Anders2006}, who describe stabiliser states as the image of graph states $\ket{G}$~\cite{Hein2004, Hein2006} under a tensor product of local Clifford operations (essentially products of $H$ and $S$).
Note that our techniques represent the state $\ket{\texttt +}^{\otimes n}$ using a rank of $r = n$, and expansion matrix $A = I_n$, and a Gram matrix $Q = 0$.
From the way that the Gram matrix $Q$ updates under $\mathrm CZ$ operations with our techniques, it follows that a quadratic form expansion for a graph state $\ket{G}$ is essentially to set $r = n$, $A = I_n$, and to set $Q$ to the adjacency matrix of $G$.
The degree parameters of Ref.~\cite{Anders2006} then coincide with the sparsity parameters $w$ for the Gram matrix $Q$.
Quadratic form expansions of the sort of Eqn.~\eqref{eqn:QFE-specific} may then be considered the image of a graph state under additional $S$ operations, followed by a `CNOT circuit' --- which is to say, a linear isometry of the form $\ket{\vec x} \mapsto \ket{A\vec x}$.
The principal advantage of our techniques over those of Ref.~\cite{Anders2006} is in the use of the that isometry, represented by the expansion matrix $A$, to represent correlations between $Z$-basis measurement outcomes.
This provides us with improved simulation of parallel $Z$-basis measurements as in Section~\ref{sec:parallelMeasurements}, and a way to try to systematically reduce the complexity of operations in structured circuits as suggested by the results of Section~\ref{stabmeas}.

Our techniques have certain advantages and disadvantages compared to simulation with the stabiliser formalism, and its elaborations in Refs.~\cite{Aaronson2004,Bravyi2019,gidney2021stim}.
Those techniques uniformly require $\mathcal{O}(n)$ time to simulate unitary operations such as $S$, $H$, $\mathrm CZ$, and $\mathrm CX$, whereas in the case $r \in \Theta(n)$, our techniques frequently require $\Theta(n^2)$.
Our techniques are no worse than the stabiliser formalism for simulating measurements, as stabiliser-based techniques require time $\mathcal{O}(n^2)$ in general to perform a weak simulation of a single measurement.
More notably, in the case of strong simulation of a few qubits (or of nearly all of the qubits) in the $Z$-basis, our techniques provide asymptotic improvements over the stabiliser formalism for $Z$-basis measurements with random outcomes.
Furthermore, whereas Gidney's refinement~\cite{gidney2021stim} enables $\mathcal{O}(n)$-time simulation of multi-qubit measurements with deterministic outcomes, our techniques can simulate single-qubit measurements with deterministic outcomes in time $\mathcal{O}(1)$.
As we describe in Section~\ref{stabmeas}, our techniques also extends to $\mathcal{O}(n)$-time simulation of fault-tolerant syndrome measurement of `local' Pauli operators $P$ under Pauli noise models including ones which yield non-deterministic measurement outcomes as part of the process.
Thus, our techniques may prove more practical for simulations in which such measurements are frequent.
We speculate that the setting of operations on error-corrected qubits may also have enough additional structure to allow simulation using sparse data structures: this may provide further opportunities to improve the simulation complexity of these circuits, using our techniques.

Quadratic form expansions could possibly be regarded as \emph{less general} than the representation for stabiliser states presented by Bravyi~\emph{et~al.}~\cite{Bravyi2019}.
The sense in which this is the case is as follows.
Ref.~\cite{Bravyi2019} describes a representation of stabiliser states in the form
\begin{equation}
    \ket{\psi} = \omega \;\! U_C \:\! U_H \ket{\vec b},
\end{equation}
where $\omega \in \mathbb{C}$ is a scalar, $\vec b \in \{0,1\}^n$, $U_H$ is a unitary realisable by Hadamard operations applied to some of the qubits, and $U_C$ is an operator such that $U_C \ket{\texttt0}^{\otimes n} = \ket{\texttt0}^{\otimes n}$.
In our work, the choice of $0 \le r \le n$ and an initial expansion matrix of only principal rows and zero rows corresponds to $U_H$; then the choice of Gram matrix and expansion matrix can be described by an operator $U_C$.
By using a representation for $U_C$ using the stabiliser formalism, Ref.~\cite{Bravyi2019} is able to simulate one- and two-qubit Clifford operations in time $\mathcal{O}(n)$. 
Our techniques are better suited to cases where the expansion matrix $A$ or Gram matrix $Q$ are expected to be close to $\mathcal{O}(\sqrt n)$-sparse, so that unitary gates may be simulated in time $\mathcal{O}(n)$, and where measurements are frequent enough that the techniques of Sections~\ref{sec:parallelMeasurements} and~\ref{stabmeas} provide an advantage over the $\mathcal{O}(n^2)$-time simulation time using the techniques of Ref.~\cite{Bravyi2019}.

Guan and Regan~\cite{GuanRegan2019} describe results in strong circuit simulation, which also makes use of what we call a quadratic form expansion, albeit allowing a summation index $\mathbf y \in \{0,1\}^h$ for $h$ possibly much larger than $n$, and using a different approach to that used in our analysis to navigate mixed-modulus arithmetic in the quadratic form.
Motivated by the connection between computing binary matrix rank and strong simulation, they describe techniques to relate the problem of evaluating $p = \lvert \bra{00 \cdots 0} C \ket{00\cdots0} \rvert^2$ for a Clifford circuit $C$ (consisting of only $\mathrm{CZ}$ gates, $S$ gates, and Hadamard gates) to transformations of the quadratic form.
This allows them to compute such probabilities $p$ in time $\mathcal{O}(M + n + M_{\!\!\;h}^{\;\!\omega})$, where $n$ the number of qubits on which $C$ acts, $M$ is the number of Clifford gates in $C$, and $M_{\!\!\;h}$ is specifically the number of Hadamard gates (which contribute to the length $h$ of the summation index $\vec y$).
These results rely on allowing $h$ to grow without an upper bound of $n$.
We may contrast this run-time for strong simulation of a  `total' measurement, to that of Theorem~\ref{thm:strongTotalMeasurement}, which yields an upper bound of $\mathcal{O}(Mn^2)$ to evaluate $p$ (dominated by the time to construct a quadratic form expansion for the pre-measurement state).
The results of Ref.~\cite{GuanRegan2019} are advantageous in the setting that $M_{\!\!\;h}^{\;\!\omega} \ll M n^2$.
If we consider the case $M_h / M \in \Theta(1)$ in which the Hadamard gates make up a significant fraction of the total number of gates in $C$, this is equivalent to a bound of $M \in \mathcal{O}(n^{2/(\omega-1)})$ on the circuit size.
For the known bound $\omega < 2.3729$, this implies that the techniques of Ref.~\cite{GuanRegan2019} scale asymptotically similarly to ours (or better) to compute $p$ for circuits of size about $\mathcal{O}(n^{1.4568})$.

There is potential for our techniques involving quadratic form expansions, to work well in conjunction with techniques for quantum circuit minimisation or verification.
For instance, consider a Clifford circuit expressing an $n$-qubit unitary $U$.
By computing a quadratic form expansion for the $2n$-qubit `Choi' state ${(U \otimes I) \ket{\Phi_n}}$ where $\ket{\Phi_n} = {\tfrac{1}{\sqrt {2^n}} \sum_{\vec x} \ket{\vec x}\ket{\vec x}}$, we may determine whether or not $U$ is in fact the identity operation.
This observation may also be easily extended to circuits with measurement operations, provided that we consider the transformation realised for specific measurement outcomes.
In this respect, our techniques may be used as an alternative realisation of a subset of the techniques of Amy~\cite{Amy_2019} to verify quantum circuit equalities.
At the same time, our representation of stabiliser states have a certain similarity to a representation of operations by Heyfron and Campbell~\cite{Heyfron2017}, who use a symmetric \emph{order 3} tensor $S$ where our representation uses a symmetric Gram matrix $Q$, to represent certain operations from the third level of the Clifford hierarchy~\cite{gottesman1999, Zeng2008}.
Perhaps by reconsidering the results of Ref.~\cite{Heyfron2017} through a similar careful management of mixed-modulus arithmetic, our techniques may be extended to provide further advances in circuit simplification, for circuits involving controlled-controlled-$Z$, controlled-$S$, or $T$ gates.
The fact that our techniques represent the global phase of a simulated state, means that they may be also used in conjunction with techniques described in Ref.~\cite{Bravyi2019} for extending beyond simulation of stabiliser circuits.
Of course, Ref.~\cite{Bravyi2019} also provides simulation techniques which extend the stabiliser formalism and track the global phase; it remains to be seen whether there exists applications for which we may use the structure of the circuit to be simulated, to more effectively simulate them using sparse data structures.

Finally, we suggest that quadratic form expansions may be useful pedagogically, for discussions about stabiliser circuits in an educational setting.
Most students who first learn of quantum computation will be familiar with the idea of expanding a state in terms of standard basis components, and performing computation \emph{within} a standard basis vector, both of which are clear features of the quadratic form expansion representation.
Entanglement, in the form of states such as $\tfrac{1}{\sqrt 2} \sum_{x} \ket{x}\ket{x}$, is also represented explicitly by quadratic form expansions: and while students of physics may find quantum correlations to be adequately represented by correlations of $Z$ and $X$ observables, not all students of quantum computation may find it as easy to reason about observables in this way.
This issue also applies more generally to the stabiliser formalism: while the usefulness of the stabiliser formalism is beyond doubt, it is perhaps more conceptually sophisticated than necessary for a first encounter with circuit simulation.
In contrast, our techniques may be simplified to be more approachable at an introductory level --- for instance, by abandoning the constraint of keeping $A$ in principal row form, and considering simulation techniques requiring $\mathcal{O}(n^3)$ time using Gaussian elimination to maintain the constraint $r = \mathrm{rank}(A)$.
This would yield techniques which may be more suitable to teach the simulatability of stabiliser circuits in a first encounter.
At the same time, it is possible to describe the limits of extending these techniques to circuits with $T$ or controlled-$S$ gates, by showing how specific steps fail when the relative phase is described as a power of $\tau = \sqrt{i}$, or when the Gram matrix $Q$ is instead allowed to be replaced with a matrix which is not symmetric.
This is a possible avenue to describing efficiently simulatable quantum computations, in a way which may be more approachable at an introductory level.

\section*{Acknowledgements}

This work was partly undertaken whilst N.\,de\,B.\ was the Oxford--Tencent Fellow, and S.\,J.\,H.\ was working part-time on the EPSRC-funded NQIT project, in the Department of Computer Science at the University of Oxford.
The authors thank Matthew Amy and Zen Harper for their helpful feedback on a draft of this article; Alec Edgington for checking the pseudocode and spotting a number of bugs whilst undertaking a code implementation \cite{simplex} of the techniques described in this article; Alex Kerzner for contacting us to report some minor bugs in an earlier version of the article; and the anonymous reviewers at \textit{Quantum} for their careful review and helpful comments.

\newpage
\appendix

\section{Purely number-theoretic Lemmata on Gram matrices modulo~4}




We call a function $\mathbf Q: \mathbb Z^r \to \mathbb Z$ a \emph{quadratic form} if $\mathbf Q(\vec x)$ is an integer multivariate polynomial in which each term has degree $2$, \emph{i.e.},~if can be expressed as
$
        \mathbf Q(\mathbf x)
    =
        \sum_{j \le k} 
        u_{j,k} \;\! x_j \;\! x_k
$  
for some coefficients $u_{j,k} \in \mathbb Z$\,.
When working over a field such as $\mathbb R$ or $\mathbb Z_p$ for prime $p > 2$, it is common to consider a symmetric matrix $Q$ such that $Q_{j,j} = u_{j,j}$, $Q_{j,k} = \tfrac{1}{2}u_{j,k}$ for $j < k$, and $Q_{j,k} = \tfrac{1}{2}u_{k,j}$ for $j > k$.
The matrix $Q$ is then called the \emph{Gram matrix} of $\mathbf Q$.
While $\mathbb Z$ does not have a multiplicative inverse for $2$ (so that not all quadratic forms over $\mathbb Z$ can be represented in this way), our techniques are concerned with quadratic forms $\mathbf Q$ over $\mathbb Z$, which do happen to arise from a `Gram matrix' $Q$ in this way.
In this case, we may represent $\mathbf Q(\vec x) = \vec x\trans \!Q \vec x$.

As we are exclusively interested in quadratic forms evaluated for vectors $\vec x \in \{0,1\}^r$, and evaluated as an imaginary exponent $i^{\:\!\mathbf Q(\vec x)}$ to define a relative phase, we may show a few simple but very helpful results about such Gram matrices to support our analysis.
In particular, the fact that $Q$ is symmetric means that when considering $\mathbf Q(\vec x) = \vec x\trans \! Q \:\! \vec x$ mod~4, our analysis in fact benefits from properties which we would normally only expect when evaluating expressions modulo~2:

\begin{lemma}
    \label{lemma:consistency-mod-2}
    Let $\vec x, \vec{\tilde x} \in \mathbb Z^r$ such that $\vec x \equiv \vec{\tilde x} \pmod{2}$.
    Then for any symmetric $r \times r$ matrix $Q$ over the integers, $\vec x\trans \! Q \:\! \vec x \,\equiv\, \vec{\tilde x}\trans\!Q \:\! \vec{\tilde x} \pmod{4}$.
\end{lemma}
\begin{proof}
    Let $\vec v \in \mathbb Z^r$ be such that $\vec{\tilde x}  - \vec x = 2 \vec v$.
    Then we have
    \begin{equation}
    \begin{aligned}[b]
            \vec{\tilde x}\trans\! Q \:\! \vec{\tilde x}
        \,\,=\;
            (\vec x + 2 \vec v)\trans\! Q \:\! (\vec x + 2 \vec v)
        \;&=\;
            \vec x\trans\! Q \vec x 
            \,+\, 2 \vec v\trans\! Q \:\! \vec x \,+\, 2 \vec x\trans\! Q \:\! \vec v
            \,+\, 4 \vec v\trans\! Q \:\! \vec v
        \\[1ex]&=\;
            \vec x\trans\! Q \:\! \vec x 
            \,+\, 4 (\vec v\trans\! Q \:\! \vec x 
            + \vec v\trans\! Q \:\! \vec v)
        \;\equiv\;
            \vec x\trans\! Q \:\! \vec x \pmod{4}.
    \end{aligned}
    \tag*{\qedhere}
    \end{equation}
\end{proof} 

\noindent
This phenomenon in which equivalence mod~2 automatically lifts to equivalence mod~4, extends to a limited extent even to the Gram matrix $Q$ itself:

\begin{lemma}
    \label{lemma:congruence-of-Gram-matrices-mod-4}
    For $r \times r$ symmetric matrices $Q$ and $Q'$ over $\mathbb Z$, we have $\vec x\trans \!Q \:\!\vec x \equiv \vec x\trans \!Q' \:\!\vec x \pmod{4}$ for all $\vec x \in \{0,1\}^r$ if and only if $Q' - Q = 2\Gamma$ for some matrix $\Gamma$ over $\mathbb Z$ whose diagonal entries are all even.
\end{lemma}
\begin{proof}
    Let $Q$ and $Q'$ be symmetric $r \times r$ matrices over $\mathbb Z$; and let $\mathrel{\equiv_4}$ stand for the relation of equivalence modulo $4$.
    \begin{itemize}
    \item
        Suppose that $Q' - Q = 2\Gamma$ for an integer matrix $\Gamma$ with an even diagonal.
        Then we have
        \vspace*{-1ex}
        \begin{equation}
        \begin{aligned}[b]
                \vec x\trans Q' \;\! \vec x
            \;\;&=\;\:\:
                \sum_{k=1}^r x_k^2 (Q_{k,k} + 2\Gamma_{k,k})
                \;\;+\!\!
                \sum_{1 \le j < k \le r} \!\!\!
                    2 x_j (Q_{j,k} + 2 \Gamma_{j,k}) x_k
            \\&\mathrel{\equiv_4}\; 
                \sum_{k=1}^r x_k^2 Q_{k,k}
                \;\;+\!\!
                \sum_{1 \le j < k \le r} \!\!\!
                    2 \;\! x_j Q_{j,k} x_k
            \;\;=\;\;
                \vec x\trans Q \;\!\vec x
        \end{aligned}
        \end{equation}~\\[-5ex]

        \noindent
        for all $\vec x \in \{0,1\}^r$.
    \item
        Suppose that $\vec x\trans Q\;\!\vec x \equiv \vec x \trans Q'\;\!\vec x \pmod{4}$ for all $\vec x \in \{0,1\}^r$.
        Then in particular, $Q'_{k,k} = \vec e_k\trans Q'\;\!\vec e_k \mathrel{\equiv_4} \vec e_k\trans Q\;\!\vec e_k = Q_{k,k}$; and 
        \begin{equation}
            \begin{aligned}[b]
                    2 Q'_{j,k}
                \;& =\;
                    (\vec e_j{+}\vec e_k)\trans Q' (\vec e_j{+}\vec e_k) - Q'_{j,j} - Q'_{k,k} \\
                \;& \mathrel{\equiv_4}\;
                    (\vec e_j{+}\vec e_k)\trans Q (\vec e_j{+}\vec e_k) - Q_{j,j} - Q_{k,k}
                \;=\;
                    2 Q_{j,k}
                \;.
            \end{aligned}
        \end{equation}
        It follows that $2(Q' - Q) \equiv 0 \pmod{4}$, so that $\Delta = Q' - Q$ is a symmetric matrix with only even coefficients and a diagonal which is entirely zero; if $\Gamma = \tfrac{1}{2}\Delta$, then $\Gamma$ is an integer matrix with a diagonal which is entirely even.
        \qedhere
    \end{itemize}
\end{proof}
\noindent 
(We note that the above proof does not rely in a significant way on $\vec x$ having only coefficients in $\{0,1\}$: if we consider $\vec x \in \mathbb Z^r$ rather than $\vec x \in \{0,1\}^r$, the corresponding equivalence also holds.)

\newpage
\section{Subroutines to transform quadratic form expansions}
\label{apx:subroutines}

In Section~\ref{sec:computing-w-QFEs}, we describe some techniques to transform quadratic form expansions, and declare subroutines to incorporate these techniques.
Here, we describe in pseudocode how these subroutines may be implemented, and describe their run-time bounds in the setting outlined in Section~\ref{sec:data-structures}.

\subsection{Reducing Gram matrices}

\begin{figure}[t]%
    \small
    ~\hfill
    \begin{minipage}[t]{.675\textwidth}
        \rule{\textwidth}{1pt}
        $\algname{ReduceGramRowCol}(c)$
        \smallskip
        \hrule
        \smallskip
        \begin{itshape}\small
            \algdesc{ReduceGramRowCol}%
        \end{itshape}
        \smallskip
        \hrule
        \medskip
            For each $1 \le k \le r$ such that $Q_{c,k} \ne 0$:
            \begin{itemize}[leftmargin=6ex,topsep=0ex]
            \item
                If $k = c$, reduce $Q_{k,k}$ mod~4;
                otherwise reduce both $Q_{c,k}$ and $Q_{k,c}$ mod~2.
            \end{itemize}
        \vspace*{-1ex}
        \rule{\textwidth}{1pt}
    \end{minipage}
    \hfill~
    \caption{%
        \label{alg:ReduceGramRowCol}%
        A procedure to simplify the Gram matrix $Q$ in a specific row or column, suitable for when operations have been performed on that row or column.
    }
\end{figure}
Figure~\ref{alg:ReduceGramRowCol} defines a simple procedure $\algname{ReduceGramRowCol}(c)$, taking an argument $1 \le c \le r$ indexing a row/column of the Gram matrix $Q$.
It reduces its off-diagonal entries of that row and column modulo~2, and its diagonal entry modulo~4.
This subroutine will be suitable to help decrease the number of non-zero entries in that row and column of $Q$, after some operation has been performed which affects row/column $c$ of $Q$.
Following Lemma~\ref{lemma:congruence-of-Gram-matrices-mod-4}, this does not change the state which is represented by the quadratic form expansion.
We may easily show the following:
\begin{lemma}
    For an integer $1 \le c \le r$, and a Gram matrix $Q$ which has at most $w_c \ge 1$ non-zero entries in row/column $c$, $\algname{ReduceGramRowCol}(c)$ runs in time $\mathcal{O}(w_c)$.
\end{lemma}

\subsection{Changes of variables}

Figure~\ref{alg:ReindexAlgs} defines two subroutines, $\algname{ReindexSubtColumn}(k,c)$ and $\algname{ReindexSwapColumns}(k,c)$.
These both  realise a change in the representation of a quadratic form expansion, as described in Lemma~\ref{lemma:QFE-change-of-variables}, by performing column operations on two distinct columns $c$ and $k$ of the expansion matrix $A$.
We define them both in such a way that they do nothing if $c = k$; otherwise,
\begin{itemize}
\item
    $\algname{ReindexSubtColumn}(k,c)$ performs a change of variables which has the effect of subtracting column $c$ of $A$, from column $k$;
\item
    $\algname{ReindexSwapColumns}(k,c)$ performs a change of variables which has the effect of interchanging columns $c$ and $k$ of $A$.
\end{itemize}

\begin{figure}[t]%
    \small\centering
    ~\hfill
    \begin{minipage}[t]{.4875\textwidth}
        \rule{\textwidth}{1pt}
        $\algname{ReindexSubtColumn}(k,c)$
        \smallskip
        \hrule
        \smallskip
        \begin{itshape}\small 
            \algdesc{ReindexSubtColumn}%
        \end{itshape}
        \smallskip
        \hrule
        \begin{enumerate}[leftmargin=4ex,itemsep=.25ex]
        \item
            If $c = k$, then stop; otherwise proceed.
        \item
            For each $1 \le j \le n$ such that $A_{j,c} \ne 0$: \\
                update $A_{j,k} \gets A_{j,k} \oplus 1$.
        \item
            For each $1 \le h \le r$ such that $Q_{h,c} \ne 0$:  \\
                update $Q_{h,k} \gets Q_{h,k} - Q_{h,c}$.
        \item
            For each $1 \le h \le r$ such that $Q_{c,h} \ne 0$:  \\
                update $Q_{k,h} \gets Q_{k,h} - Q_{c,h}$.
        \item
            Call $\algname{ReduceGramRowCol}(k)$.
        \end{enumerate}
        \vspace*{-2ex}
        \rule{\textwidth}{1pt}
    \end{minipage}
    \hfill~\hfill
    \begin{minipage}[t]{.4875\textwidth}
        \rule{\textwidth}{1pt}
        $\algname{ReindexSwapColumns}(k,c)$
        \smallskip
        \hrule
        \smallskip
        \begin{itshape}\small 
            \algdesc{ReindexSwapColumns}%
        \end{itshape}
        \smallskip
        \hrule
        \begin{enumerate}[leftmargin=4ex,itemsep=.25ex]
        \item
            If $c = k$, then stop; otherwise, proceed.
        \item
            For each $1 \le j \le n$ such that $A_{j,c} \ne 0$ or $A_{j,k} \ne 0$: \\
                swap the values of $A_{j,c}$ and $A_{j,k}$.
        \item
            For each $1 \le j \le r$ such that $Q_{j,c} \ne 0$ or $Q_{j,k} \ne 0$: \\
                swap the values of $Q_{j,c}$ and $Q_{j,k}$.
        \item
            For each $1 \le j \le r$ such that $Q_{c,j} \ne 0$ or $Q_{k,j} \ne 0$: \\
                swap the values of $Q_{c,j}$ and $Q_{k,j}$.
        \item
            Swap the values of $p(k)$ and $p(c)$.
        \end{enumerate}
        \vspace*{-2ex}
        \rule{\textwidth}{1pt}
    \end{minipage}
    \hfill~
    \caption{%
        \label{alg:ReindexSubtColumn}%
        \label{alg:ReindexSwapColumns}%
        \label{alg:ReindexAlgs}%
        Procedures to change the representation of a quadratic form expansion, by simple column operations on $A$, and corresponding row / column operations on the Gram matrix $Q$.
    }
\end{figure}

\begin{lemma}
    \label{lemma:ReindexAlgs-runtime}
    Let $t_c \ge 0$ (respectively, $t_k$) be the number of non-zero entries in column $c$ (respectively, column $k$) of $A$; and similarly let $w_c \ge 0$ (respectively, $w_k$) be the number of non-zero entries in row/column $c$ (respectively, row/column $k$) of $Q$.
    Then:
    \begin{enumerate}[label=(\alph*)]
    \item
        The procedure $\algname{ReindexSubtColumn}(k,c)$ and realises the change of index described in Lemma~\ref{lemma:QFE-change-of-variables} for $E = I \oplus \vec e_c \vec e_k\trans\!$ in time $\mathcal{O}(t_c + w_c)$.
    \item
        The procedure $\algname{ReindexSwapColumns}(k,c)$ and realises the change of index described in Lemma~\ref{lemma:QFE-change-of-variables} for $E = I \oplus (\vec e_c \!\oplus\!\!\: \vec e_k)(\vec e_c \!\oplus\!\!\: \vec e_k)\trans\!$ in time $\mathcal{O}(t_c + t_k + w_c + w_k)$.
    \end{enumerate}
\end{lemma}
\begin{proof}
    Both of these procedures test whether $c = k$ in time $\mathcal{O}(1)$ in Step~1.
    In Step~2, they then iterate through $\mathcal{O}(t_c)$ values of $j$ in Step~2 (in the case of \algname{ReindexSubtColumn}) or $\mathcal{O}(t_c + t_k)$ values of $j$ (in the case of \algname{ReindexSwapColumns}), for which either one or two columns contain a non-zero value in row $j$, performing operations requiring $\mathcal{O}(1)$ time in each iteration.
    Similarly, in Steps~3 and~4, each procedure iterates respectively through $\mathcal{O}(w_c)$ values of $h$ or $\mathcal{O}(w_c + w_k)$ values of $h$ (for which either one or two rows/columns of $Q$ contain a non-zero value in row or column $h$), and perform operations requiring $\mathcal{O}(1)$ time for each such $h$.
    The effects of these procedures follow from the operations performed in each case, noting in particular that the transformations on $Q$ may be realised as a transformation of the columns followed by a transformation of the rows.
    Finally, in Step~5 of \algname{ReindexSubtColumn}, we reduce the Gram matrix $Q$, performing $\mathcal{O}(w_c)$ further operations; Step~5 of \algname{ReindexSwapColumns} instead interchanges $p(c)$ and $p(k)$, taking time $\mathcal{O}(1)$.
    The asymptotic run-times described for both procedures then follow.
\end{proof}

\subsection{Maintaining principal row forms}

Figure~\ref{alg:PrincipalRowAlgs} defines two subroutines, $\algname{MakePrincipal}$ and $\algname{ReselectPrincipalRow}$, which may be used to help maintain the matrix $A$ in principal row form during other procedures which may significantly affect the rank of $A$.
(These procedures are realised using $\algname{ReindexSubtColumn}$ and  $\algname{ReindexSwapColumns}$ as subroutines, performing the appropriate changes transformations on the Gram matrix $Q$ in doing so.)

The procedure $\algname{MakePrincipal}(c,j)$ takes a column-index $1 \le c \le r$ and a row-index $1 \le j \le n$ as arguments, and attempts to perform a change of variables which would make $j$ a principal row with a single non-zero coefficient in column $c$.
It does so in such a way that it does not affect the principal rows $h = p(k)$ for the other columns $1 \le k \le r$ with $k \ne c$, by only transforming $A$ if column $c$ already has a $1$ in row $j$\,.
If this is the case, it performs suitable column operations to clear the other entries in row $j$ of $A$, and then updates $p(c) \gets j$\,.
\begin{lemma}
    \label{lemma:runtime-MakPrincipal}
    For $1 \le j \le n$ a row index of $A$, and $1 \le c \le r$ be a column index of $A$, let $s_j \ge 0$ be the number of non-zero entries in row $j$ of $A$, $t_c \ge 0$ be the number of non-zero entries in column $c$ of $A$, and $w_c$ be the number of non-zero entries in row/column $c$ of $Q$.
    If $A_{j,c} = 1$, then $\algname{MakePrincipal}(c,j)$ terminates in time $\mathcal{O}(s_j t_c + s_j w_c)$.
    If $s_j = 1$, or if instead $A_{j,c} = 0$, then it terminates in time $\mathcal{O}(1)$.
\end{lemma}
\begin{proof}
    Step~1 takes time $\mathcal{O}(1)$ to determine whether $A_{j,c} = 1$.
    Assuming that the procedure does not terminate at that stage, $s_j, t_c \ge 1$.
    Step~2 then iterates through $s_j - 1$ column indices $k \ne c$, subtracting column $c$ of $A$ from column $k$ (and performing $\mathcal{O}(t_c + w_c)$ operations) by calling $\algname{ReindexSubtColumn}(k,c)$ on each iteration.
    Among other changes, this realises an update $A_{j,k} \gets 0$ for each such $k$.
    Note that as the other principal rows $h$ have $A_{h,c} = 0$ by definition, these rows are unaffected by these column operations.
    Step~2 dominates the run-time if $w_j > 1$, performing $\mathcal{O}(s_j(t_c + w_c))$ operations; if $s_j = 1$, it instead performs $\mathcal{O}(1)$ operations.
    Step~3 also performs $\mathcal{O}(1)$ operations, simply updating the value of $p(c)$.
\end{proof}

    \begin{figure}[t]%
    \small
    ~\hfill
    \begin{minipage}[t]{.475\textwidth}
        \rule{\textwidth}{1pt}
        $\algname{MakePrincipal}(c,j)$
        \smallskip
        \hrule
        \smallskip
        \begin{itshape}\small 
            \algdesc{MakePrincipal}%
        \end{itshape}
        \smallskip
        \hrule
        \begin{enumerate}[leftmargin=4ex,itemsep=.75ex]
        \item
            If $A_{j,c} = 0$, then stop; otherwise proceed.
        \item
            For each $1 \le k \le r$ such that $A_{j,k} \ne 0$:
            \begin{itemize}[topsep=0ex]
            \item
                if $k \ne c$, call \\ $\algname{ReindexSubtColumn}(k,c)$.
            \end{itemize}
        \item
            Update $p(c) \gets j$.
        \end{enumerate}
        \vspace*{-1.5ex}
        \rule{\textwidth}{1pt}
    \end{minipage}
    \hfill\hfill
    \begin{minipage}[t]{.4625\textwidth}
        \rule{\textwidth}{1pt}
        $\algname{ReselectPrincipalRow}(j,c)$
        \smallskip
        \hrule
        \smallskip
        \begin{itshape}\small
            \algdesc{ReselectPrincipalRow}%
        \end{itshape}
        \smallskip
        \hrule
        \begin{enumerate}[leftmargin=4ex,itemsep=.25ex]
        \item
            Find a row index $1 \le j_\ast \le n$ such that $A_{j_\ast,c} \ne 0$ and $j_\ast \ne j$, for which row $j_\ast$ of $A$ has the fewest non-zero entries.
            (Let $j_\ast = 0$ if no such rows exist.)
        \item
            If $j_\ast \ne 0$, call $\algname{MakePrincipal}(c,j_\ast)$.
        \end{enumerate}
        
        \vspace*{-2ex}
        \rule{\textwidth}{1pt}
    \end{minipage}
    \hfill~
\caption{%
    \label{alg:MakePrincipal}%
    \label{alg:ReselectPrincipalRow}%
    \label{alg:PrincipalRowAlgs}
    Procedures to change the representation of a quadratic form expansion by adding one column of $A$ into another, and to select a new principal row for column $k$ (performing the appropriate column-reductions to establish the new row as a principal row if such a row is found).
}
\end{figure}

Using the procedure $\algname{MakePrincipal}$ as a subroutine, we define a procedure $\algname{ReselectPrincipalRow}(j,c)$, which attempts to select a new principal row for the setting where $j = p(c)$ is a principal row, corresponding to a qubit which is to be subject to an operation such as a Hadamard gate.
This procedure is used when a simulated operation would significantly disrupt a principal row for some column $c$, which we may be able to avoid by selecting an alternative principal row $j_\ast$\,.
(This is not possible when row $j$ is the only row in which column $c$ contains a non-zero entry, in which case the matrix $A$ and the value of $p(c)$ will be unchanged.)
To reduce the amount of work that is required to establish a new principal row, we select the new principal row by minimising the number of non-zero entries which must be cleared to make it a principal row.

\begin{lemma}
    \label{lemma:runtime-ReselectPrincipalRow}
    Let $1 \le c \le r$ be a column index for an $n \times r$ matrix $A$ with a principal index map $p$, and either $j = 0$ or $j = p(c)$.
    Let $t_c \ge 0$ (respectively, $w_c$) be the number of non-zero entries in column $c$ of $A$ (respectively, in row/column $c$ of $Q$).
    \begin{itemize}
    \item
        If $j > 0$ is the only row in which column $c$ is non-zero, then $\algname{ReselectPrincipalRow}(j,c)$ does not modify the quadratic form expansion, and halts in time $\mathcal{O}(1)$.
    \item
        If column $c$ has more than one non-zero entry, let $s_{j_\ast}$ be the minimum number of non-zero entries in a row $j_\ast \ne j$ for which $A_{j_\ast, c} = 1$.
        Then $\algname{ReselectPrincipalRow}(j,c)$ finds such a row $j_\ast$, and modifies $A$ and $p$ so that $A$ is in principal row form with $p(c) = j_\ast$, halting in time $\mathcal{O}(t_c)$ if $s_{j_\ast} = 1$, and in time $\mathcal{O}(s_{j_\ast} t_c + s_{j_\ast} w_c)$ otherwise.
    \end{itemize}
\end{lemma}
\begin{proof}
Step~1 iterates through $\mathcal{O}(t_c)$ row-indices $h$, performing simple integer comparisons and assignments in each step on the basis of the number of non-zero entries in each row $h$ of $A$ (assumed to be accessible in $\mathcal{O}(1)$ time through the data structure to store $A$).
In particular, if there is only one such row $h$, this takes time $\mathcal{O}(1)$.
If $h = j > 0$ is the only row index for which $A_{h,c} = 1$, we set $j_\ast = 0$.
Otherwise, we call $\algname{MakePrincipal}(c,j_\ast)$ for the row-index $j_\ast$ of the row with the minimum number of non-zero entries.
The run-time of this step is then $\mathcal{O}(s_{j_\ast} t_c + s_{j_\ast} w_c)$ by Lemma~\ref{lemma:runtime-MakPrincipal}; in particular, if $s_{j_\ast} = 1$, this step takes time $\mathcal{O}(1)$.
The run-time analysis follows.
\end{proof}

\subsection{Fixing the value of the final index}

Figure~\ref{alg:FixFinalBit} defines the subroutine $\algname{FixFinalBit}(z)$.
This procedure realises the transformations to the quadratic form expansion described in Section~\ref{sec:fixBits}, to represent either a simplification or transformation of a quadratic form expansion in which the final bit of the summation index $\vec x \in \{0,1\}^r$ is set to some constant $z \in \{0,1\}$.

We note that in the analysis of Section~\ref{sec:fixBits}, the right-hand side of Eqn.~\eqref{eqn:bitFixedQFE-reduced} features an extra factor of $\smash{\tfrac{1}{\sqrt 2}}$ in addition to the quadratic form expansion described in square brackets.
In some instances, this scalar factor may represent the fact that a measurement which fixes the value of $x_r$, yields the particular outcome $z$ only with probability $\tfrac{1}{2}$.
In other instances, this factor may be canceled out by some other scalar factors which we may apply to $\ket{\psi'}$.
In each case, this factor of $\smash{\tfrac{1}{\sqrt 2}}$ can be accounted for in an appropriate way.
However, it is left to the programmer to ensure that this is done properly: $\algname{FixFinalBit}(z)$ does not attempt to represent or otherwise account for that scalar factor.

    \begin{figure}[t]%
    \small
    ~\hfill
    \begin{minipage}[t]{.85\textwidth}
        \rule{\textwidth}{1pt}
        $\algname{FixFinalBit}(z)$
        \smallskip
        \hrule
        \smallskip
        \begin{itshape}\small 
            \algdesc{FixFinalBit}%
        \end{itshape}
        \smallskip
        \hrule
        \begin{enumerate}[leftmargin=4ex,itemsep=.25ex]
        \item
            Let $\vec a = {[\,A_{1,r}\;\;\cdots\;\;A_{n,r}\,]\:\!\trans\!} \in \{0,1\}^n$, let $\vec q = {[\,Q_{1,r}\;\;\cdots\;\;Q_{r{-}1,r}\,]\:\!\trans} \in \mathbb Z^{r-1}\!$, and let $u = Q_{r,r}$ .
        \item
            Modify $A$ by removing column $r$, and modify $Q$ by removing column/row $r$.
        \item
            Update $Q \gets Q + 2\:\!z\:\!\mathrm{diag}(\vec q\trans) \!\!\mod{4}$, update $\vec b \gets \vec b \oplus z\:\!\vec a$, and update $g \gets g + 2zu$.
        \item
            Update $r \gets r-1$.
        \end{enumerate}
        \vspace*{-2ex}
        \rule{\textwidth}{1pt}
    \end{minipage}
    \hfill~
\caption{%
    \label{alg:FixFinalBit}%
    Procedure to simplify the representation of a quadratic form expansion, corresponding to explicitly fixing the value of the final bit of the summation index to a given value $z$.
}
\end{figure}
\begin{lemma}
    Let $t_r \ge 0$ (respectfully, $w_r$) be the number of non-zero entries in the final column of $A$ (respectfully, the final row/column of $Q$).
    Then for a bit $z \in \{0,1\}$, the procedure $\algname{FixFinalBit}(z)$ transforms the quadratic form expansion consistent with introducing a factor $\delta_{x_r,z}$ as in Eqn.~\eqref{eqn:bitFixedQFE} and reducing it to the form in square brackets in the right-hand side of Eqn.~\eqref{eqn:bitFixedQFE-reduced}, in time $\mathcal{O}(t_r + w_r)$.
\end{lemma}
\begin{proof}
    Step~1 initialises a vector $\vec a \in \{0,1\}^n$ in time $\mathcal{O}(t_r)$ by reading column $r$ of $A$, and initialises a vector $\vec q \in \{0,1\}^{r-1}$ and coefficient $u \in \{0,1\}$ in time $\mathcal{O}(w_r)$ by reading column $r$ of $Q$.
    Step~2 modifies $A$ and $Q$ by removing column $r$ (and in the case of $Q$ also row $r$); this might be done by at the same time as initialising $\vec a$ and $\vec q$ in Step~1, by appropriate operations on the sparse data structures representing $A$ and $Q$. 
    In Step~3, the update $Q \gets Q + 2 \;\! z \: \;\! \mathrm{diag}(\vec q\trans)$ may be performed in time $\mathcal{O}(w_r)$, the update $\vec b \gets \vec b \oplus z \;\!\vec a$ may be performed in time $\mathcal{O}(t_r)$, and the update $g \gets g + 2\:\!z\:\!u$ may be performed in time $\mathcal{O}(1)$; and similarly for the update $r \gets r - 1$ in Step~4.
    The run-time is thus $\mathcal{O}(t_r + w_r)$.
\end{proof}

\subsection{Eliminating columns which are entirely zero}

Figure~\ref{alg:ZeroColumnElim} defines the subroutine $\algname{ZeroColumnElim}(c)$, to simplify a quadratic form expansion as described in Section~\ref{sec:ZeroColumnElim} by eliminating one or more columns from $A$, when it has a rank $r$ which is one less than the number of its columns.

\begin{figure}[t]%
    \centering\small
    \begin{minipage}{.85\textwidth}
        \rule{\textwidth}{1pt}
        $\algname{ZeroColumnElim}(c)$
        \smallskip
        \hrule
        \smallskip
        \begin{itshape}\small 
            \algdesc{ZeroColumnElim}%
        \end{itshape}
        \smallskip
        \hrule
        \begin{enumerate}[leftmargin=4ex,itemsep=.25ex]
        \item
            Call $\algname{ReindexSwapColumns}(c,r{+}1)$.
        \item
            Let $\vec q  = {[\,Q_{1,r{+}1}\;\;\cdots\;\;Q_{r,r{+}1}\,]\:\!\trans} \in \mathbb Z^r$, and let $u = Q_{r{+}1,r{+}1} \pmod{4}$.
        \item
            Modify $A$ by removing column $r{+}1$, and modify $Q$ by removing column $r{+}1$ and row $r{+}1$.
        \item
            If $u$ is odd,
                update $Q \gets Q + (u-2)\:\!\vec q \:\! \vec q\trans \!\!\mod{4}$ and $g \gets g - u + 2$.
            Otherwise, if $u$ is even:
            \begin{enumerate}[leftmargin=4ex,topsep=0ex,itemsep=.25ex,label=\alph*.]
            \item
                If $\vec q = \vec 0$, then stop; otherwise, find an index $1 \le \ell \le r$ such that $q_{\ell} = 1$.
            \item
                For each $1 \le k \le r$ such that $q_k \ne 0$:
                \begin{itemize}[topsep=0ex]
                \item
                    If $k \ne \ell$, call $\algname{ReindexSubtColumn}(k,\ell)$.
                \end{itemize}
            \item
                Call $\algname{ReindexSwapColumns}(r,\ell)$.
            \item
                Call $\algname{FixFinalBit}(u/2)$.
            \end{enumerate}
        \end{enumerate}
        \vspace*{-3ex}
        \rule{\textwidth}{1pt}
    \end{minipage}
\caption{%
    \label{alg:ZeroColumnElim}
    A procedure to eliminate redundant columns from the matrix $A$ of a quadratic form expansion, in the context of a Hadamard operation on a qubit $j$ which takes $A$ out of principal row form.
}
\end{figure}

\begin{lemma}
    Let $1 \le c \le r$ be an index for a column of $A$.
    Let $s,t,w \ge 0$ respectively be an upper bound on the number of non-zero entries in any row of $A$, any column of $A$, and any row/column of $Q$.
    Suppose that 
    \begin{equation}
        \ket{\psi}
    \;=\;
        \frac{\tau^g}{\sqrt{2^{r+1}}}
        \!\!
        \sum_{\mathbf x \in \{0,1\}^{r\!\!\:+\!\!\:1}} \!\!\!\!
            i^{\,\vec x \trans \! Q \!\: \vec x} \,
            \ket{A \vec x \oplus \vec b} ,
    \end{equation}
    but that $A$ has rank only $r$, and that in particular column $c$ of $A$ is entirely zero; and that $p$ is nearly a principal index map for $A$, except again in that column $c$ of $A$ is entirely zero.
    Then the procedure $\algname{ZeroColumnElim}(c)$ transforms the quadratic form expansion to a form consistent with Eqn.~\eqref{eqn:QFE-specific}, in which $A$ is in principal row form with a corresponding principal index map $p$, in time $\mathcal{O}(tw + w^2)$.
\end{lemma}

\begin{proof}
    Step~1 interchanges column $c$ and column $r{+}1$ of $A$, requiring time $\mathcal{O}(2t + 2w)$ to do so.
    Step~2 initialises the vector $\vec q$ and the scalar $u$, which may be done in time $\mathcal{O}(w)$.
    Step~3 modifies $A$ and $Q$ by removing column $r{+}1$ (and in the case of $Q$ also row $r{+}1$); no explicit modifications may be necessary to the sparse data structure of $A$ (as the column being removed has no non-zero entries), and the modification to $Q$ might be done by appropriate operations on the sparse data structure representing $Q$ while initialising $\vec q$ in Step~2.
    
    If $u$ is odd, then in Step~4 we modify the value of $g$, and up to $w$ diagonal entries of $Q$, requiring time $\mathcal{O}(w^2)$; the total run-time is then $\mathcal{O}(t + w^2)$.
    If $u$ is even, we instead perform a number of further operations, which depend on a column index $1 \le \ell \le r$ for which $q_\ell = 1$.
    Step~4a attempts to find such an $\ell$, taking  $\mathcal{O}(1)$ time.
    We stop the procedure if no such $\ell$ exists. Otherwise, we invoke $\algname{ReindexSubtColumn}(k,\ell)$
    for up to $w - 1$ values of $1 \le k \le r$ in Step~4b (requiring $\mathcal{O}(wt + w^2)$ operations in total), invoke $\algname{ReindexSwapColumns}(r,\ell)$ in Step~4c (involving $\mathcal{O}(2t + 2w)$ operations), and call $\algname{FixFinalBit}(q/2)$ in Step~4d (involving $\mathcal{O}(t + w)$ operations).
    The total run-time in this case is then $\mathcal{O}(tw + w^2)$.
\end{proof}

\newpage

\begin{thebibliography}{10}
\providecommand{\url}[1]{#1}
\csname url@samestyle\endcsname
\providecommand{\newblock}{\relax}
\providecommand{\bibinfo}[2]{#2}
\providecommand{\BIBentrySTDinterwordspacing}{\spaceskip=0pt\relax}
\providecommand{\BIBentryALTinterwordstretchfactor}{4}
\providecommand{\BIBentryALTinterwordspacing}{\spaceskip=\fontdimen2\font plus
\BIBentryALTinterwordstretchfactor\fontdimen3\font minus
  \fontdimen4\font\relax}
\providecommand{\BIBforeignlanguage}[2]{{%
\expandafter\ifx\csname l@#1\endcsname\relax
\typeout{** WARNING: IEEEtran.bst: No hyphenation pattern has been}%
\typeout{** loaded for the language `#1'. Using the pattern for}%
\typeout{** the default language instead.}%
\else
\language=\csname l@#1\endcsname
\fi
#2}}
\providecommand{\BIBdecl}{\relax}
\BIBdecl

\bibitem{Aaronson2004}
\BIBentryALTinterwordspacing
S.~Aaronson and D.~Gottesman, ``Improved simulation of stabilizer circuits,''
  \emph{Physical Review A}, vol.~70, no.~5, nov 2004. [Online]. Available:
  \url{https://doi.org/10.1103/physreva.70.052328}
\BIBentrySTDinterwordspacing

\bibitem{Anders2006}
\BIBentryALTinterwordspacing
S.~Anders and H.~J. Briegel, ``Fast simulation of stabilizer circuits using a
  graph-state representation,'' \emph{Physical Review A}, vol.~73, no.~2, Feb
  2006. [Online]. Available: \url{http://doi.org/10.1103/PhysRevA.73.022334}
\BIBentrySTDinterwordspacing

\bibitem{Bravyi2016}
\BIBentryALTinterwordspacing
S.~Bravyi, G.~Smith, and J.~A. Smolin, ``Trading classical and quantum
  computational resources,'' \emph{Physical Review X}, vol.~6, no.~2, Jun 2016.
  [Online]. Available: \url{http://doi.org/10.1103/PhysRevX.6.021043}
\BIBentrySTDinterwordspacing

\bibitem{gidney2021stim}
\BIBentryALTinterwordspacing
C.~Gidney, ``Stim: a fast stabilizer circuit simulator,'' \emph{Quantum},
  vol.~5, p. 497, jul 2021. [Online]. Available:
  \url{https://doi.org/10.22331/q-2021-07-06-497}
\BIBentrySTDinterwordspacing

\bibitem{ShorFactoring}
\BIBentryALTinterwordspacing
P.~Shor, ``Algorithms for quantum computation: discrete logarithms and
  factoring,'' pp. 124--134, 1994. [Online]. Available:
  \url{https://doi.org/10.1109/SFCS.1994.365700}
\BIBentrySTDinterwordspacing

\bibitem{GroverSearch}
\BIBentryALTinterwordspacing
L.~K. Grover, ``A fast quantum mechanical algorithm for database search,'' in
  \emph{Proceedings of the Twenty-Eighth Annual ACM Symposium on Theory of
  Computing}, ser. STOC '96.\hskip 1em plus 0.5em minus 0.4em\relax New York,
  NY, USA: Association for Computing Machinery, 1996, p. 212–219. [Online].
  Available: \url{https://doi.org/10.1145/237814.237866}
\BIBentrySTDinterwordspacing

\bibitem{GK}
\BIBentryALTinterwordspacing
D.~{Gottesman}, ``{The Heisenberg Representation of Quantum Computers},''
  \emph{arXiv e-prints}, Jul 1998. [Online]. Available:
  \url{https://doi.org/10.48550/ARXIV.QUANT-PH/9807006}
\BIBentrySTDinterwordspacing

\bibitem{Devitt2013}
\BIBentryALTinterwordspacing
S.~J. Devitt, W.~J. Munro, and K.~Nemoto, ``Quantum error correction for
  beginners,'' \emph{Reports on Progress in Physics}, vol.~76, no.~7, p.
  076001, Jun 2013. [Online]. Available:
  \url{http://doi.org/10.1088/0034-4885/76/7/076001}
\BIBentrySTDinterwordspacing

\bibitem{Terhal2015}
\BIBentryALTinterwordspacing
B.~M. Terhal, ``Quantum error correction for quantum memories,'' \emph{Reviews
  of Modern Physics}, vol.~87, no.~2, p. 307–346, Apr 2015. [Online].
  Available: \url{http://doi.org/10.1103/RevModPhys.87.307}
\BIBentrySTDinterwordspacing

\bibitem{Roffe2019}
\BIBentryALTinterwordspacing
J.~Roffe, ``Quantum error correction: an introductory guide,''
  \emph{Contemporary Physics}, vol.~60, no.~3, p. 226–245, Jul 2019.
  [Online]. Available: \url{http://doi.org/10.1080/00107514.2019.1667078}
\BIBentrySTDinterwordspacing

\bibitem{Bravyi2019}
\BIBentryALTinterwordspacing
S.~Bravyi, D.~Browne, P.~Calpin, E.~Campbell, D.~Gosset, and M.~Howard,
  ``Simulation of quantum circuits by low-rank stabilizer decompositions,''
  \emph{Quantum}, vol.~3, p. 181, Sep 2019. [Online]. Available:
  \url{http://doi.org/10.22331/q-2019-09-02-181}
\BIBentrySTDinterwordspacing

\bibitem{dBQuadratic}
\BIBentryALTinterwordspacing
N.~de~Beaudrap, V.~Danos, E.~Kashefi, and M.~Roetteler, ``Quadratic form
  expansions for unitaries,'' in \emph{Theory of Quantum Computation,
  Communication, and Cryptography}, Y.~Kawano and M.~Mosca, Eds.\hskip 1em plus
  0.5em minus 0.4em\relax Berlin, Heidelberg: Springer Berlin Heidelberg, 2008,
  pp. 29--46. [Online]. Available:
  \url{https://doi.org/10.1007/978-3-540-89304-2_4}
\BIBentrySTDinterwordspacing

\bibitem{CalderbankGoodQEC}
\BIBentryALTinterwordspacing
A.~R. Calderbank and P.~W. Shor, ``Good quantum error-correcting codes exist,''
  \emph{Physical Review A}, vol.~54, no.~2, p. 1098–1105, Aug 1996. [Online].
  Available: \url{http://doi.org/10.1103/PhysRevA.54.1098}
\BIBentrySTDinterwordspacing

\bibitem{dehaene}
\BIBentryALTinterwordspacing
J.~{Dehaene} and B.~{de Moor}, ``{Clifford group, stabilizer states, and linear
  and quadratic operations over GF(2)},'' \emph{Physical Review A}, vol.~68,
  no.~4, p. 042318, Oct 2003. [Online]. Available:
  \url{https://doi.org/10.1103/physreva.68.042318}
\BIBentrySTDinterwordspacing

\bibitem{VDNClassical}
\BIBentryALTinterwordspacing
M.~Van Den~Nest, ``Classical simulation of quantum computation, the
  gottesman-knill theorem, and slightly beyond,'' \emph{Quantum Info. Comput.},
  vol.~10, no.~3, Mar 2010. [Online]. Available:
  \url{https://doi.org/10.26421/QIC10.3-4-6}
\BIBentrySTDinterwordspacing

\bibitem{JBVClassical}
\BIBentryALTinterwordspacing
J.~Bermejo-Vega and M.~Van Den~Nest, ``Classical simulations of abelian-group
  normalizer circuits with intermediate measurements,'' \emph{Quantum
  Information and Computation}, vol.~14, no. 3\&4, pp. 181--0216, March 2014.
  [Online]. Available: \url{https://doi.org/10.26421/QIC14.3-4-1}
\BIBentrySTDinterwordspacing

\bibitem{Amy_2019}
\BIBentryALTinterwordspacing
M.~Amy, ``Towards large-scale functional verification of universal quantum
  circuits,'' \emph{Electronic Proceedings in Theoretical Computer Science},
  vol. 287, p. 1–21, Jan 2019. [Online]. Available:
  \url{http://doi.org/10.4204/EPTCS.287.1}
\BIBentrySTDinterwordspacing

\bibitem{Gross2006}
\BIBentryALTinterwordspacing
D.~Gross, ``Hudson’s theorem for finite-dimensional quantum systems,''
  \emph{Journal of Mathematical Physics}, vol.~47, no.~12, p. 122107, Dec 2006.
  [Online]. Available: \url{http://doi.org/10.1063/1.2393152}
\BIBentrySTDinterwordspacing

\bibitem{beaudrap2019quantum}
\BIBentryALTinterwordspacing
N.~de~Beaudrap and S.~Herbert, ``Quantum linear network coding for entanglement
  distribution in restricted architectures,'' \emph{Quantum}, vol.~4, p. 356,
  nov 2020. [Online]. Available:
  \url{https://doi.org/10.22331/q-2020-11-01-356}
\BIBentrySTDinterwordspacing

\bibitem{GuanRegan2019}
\BIBentryALTinterwordspacing
C.~Guan and K.~W. Regan, ``Stabilizer circuits, quadratic forms, and computing
  matrix rank,'' 2019. [Online]. Available:
  \url{https://doi.org/10.48550/arxiv.1904.00101}
\BIBentrySTDinterwordspacing

\bibitem{NandC}
\BIBentryALTinterwordspacing
M.~A. Nielsen and I.~L. Chuang, \emph{Quantum Computation and Quantum
  Information: 10th Anniversary Edition}, 10th~ed.\hskip 1em plus 0.5em minus
  0.4em\relax USA: Cambridge University Press, 2011. [Online]. Available:
  \url{https://doi.org/10.1017/CBO9780511976667}
\BIBentrySTDinterwordspacing

\bibitem{JozsaClassical}
\BIBentryALTinterwordspacing
R.~Jozsa and M.~Van Den~Nest, ``Classical simulation complexity of extended
  clifford circuits,'' \emph{Quantum Info. Comput.}, vol.~14, no. 7\&8, p.
  633–648, May 2014. [Online]. Available:
  \url{https://doi.org/10.48550/arxiv.1305.6190}
\BIBentrySTDinterwordspacing

\bibitem{Bravyi2016b}
\BIBentryALTinterwordspacing
S.~Bravyi and D.~Gosset, ``Improved classical simulation of quantum circuits
  dominated by clifford gates,'' \emph{Physical Review Letters}, vol. 116,
  no.~25, Jun 2016. [Online]. Available:
  \url{http://doi.org/10.1103/PhysRevLett.116.250501}
\BIBentrySTDinterwordspacing

\bibitem{Fowler2012}
\BIBentryALTinterwordspacing
A.~G. Fowler, M.~Mariantoni, J.~M. Martinis, and A.~N. Cleland, ``Surface
  codes: Towards practical large-scale quantum computation,'' \emph{Physical
  Review A}, vol.~86, no.~3, Sep 2012. [Online]. Available:
  \url{http://doi.org/10.1103/PhysRevA.86.032324}
\BIBentrySTDinterwordspacing

\bibitem{landahl2011}
\BIBentryALTinterwordspacing
A.~J. Landahl, J.~T. Anderson, and P.~R. Rice, ``Fault-tolerant quantum
  computing with color codes,'' 2011. [Online]. Available:
  \url{https://doi.org/10.48550/arxiv.1108.5738}
\BIBentrySTDinterwordspacing

\bibitem{Chao2018}
\BIBentryALTinterwordspacing
R.~Chao and B.~W. Reichardt, ``Quantum error correction with only two extra
  qubits,'' \emph{Physical Review Letters}, vol. 121, no.~5, Aug 2018.
  [Online]. Available: \url{http://doi.org/10.1103/PhysRevLett.121.050502}
\BIBentrySTDinterwordspacing

\bibitem{Shorfault}
\BIBentryALTinterwordspacing
P.~W. Shor, ``Fault-tolerant quantum computation,'' in \emph{Proceedings of the
  37th Annual Symposium on Foundations of Computer Science}, ser. FOCS
  '96.\hskip 1em plus 0.5em minus 0.4em\relax USA: IEEE Computer Society, 1996,
  p.~56. [Online]. Available: \url{https://doi.org/10.1109/SFCS.1996.548464}
\BIBentrySTDinterwordspacing

\bibitem{DiVincenzo2007}
\BIBentryALTinterwordspacing
D.~P. DiVincenzo and P.~Aliferis, ``Effective fault-tolerant quantum
  computation with slow measurements,'' \emph{Physical Review Letters},
  vol.~98, no.~2, Jan 2007. [Online]. Available:
  \url{http://doi.org/10.1103/PhysRevLett.98.020501}
\BIBentrySTDinterwordspacing

\bibitem{Bennett1996}
\BIBentryALTinterwordspacing
C.~H. Bennett, G.~Brassard, S.~Popescu, B.~Schumacher, J.~A. Smolin, and W.~K.
  Wootters, ``Purification of noisy entanglement and faithful teleportation via
  noisy channels,'' \emph{Phys. Rev. Lett.}, vol.~76, pp. 722--725, Jan 1996.
  [Online]. Available: \url{https://doi.org/10.1103/physrevlett.76.722}
\BIBentrySTDinterwordspacing

\bibitem{Nigmatullin2016}
\BIBentryALTinterwordspacing
R.~Nigmatullin, C.~J. Ballance, N.~de~Beaudrap, and S.~C. Benjamin, ``Minimally
  complex ion traps as modules for quantum communication and computing,''
  \emph{New Journal of Physics}, vol.~18, no.~10, p. 103028, 2016. [Online].
  Available: \url{https://doi.org/10.1088/1367-2630/18/10/103028}
\BIBentrySTDinterwordspacing

\bibitem{Dur2007}
\BIBentryALTinterwordspacing
W.~Dür and H.~J. Briegel, ``Entanglement purification and quantum error
  correction,'' \emph{Reports on Progress in Physics}, vol.~70, no.~8, p.
  1381–1424, Jul 2007. [Online]. Available:
  \url{http://doi.org/10.1088/0034-4885/70/8/R03}
\BIBentrySTDinterwordspacing

\bibitem{Dawson2005}
\BIBentryALTinterwordspacing
C.~M. Dawson, A.~P. Hines, D.~Mortimer, H.~L. Haselgrove, M.~A. Nielsen, and
  T.~J. Osborne, ``Quantum computing and polynomial equations over the finite
  field {Z2},'' \emph{Quantum Info. Comput.}, vol.~5, no.~2, p. 102–112, Mar.
  2005. [Online]. Available:
  \url{https://doi.org/10.48550/arxiv.quant-ph/0408129}
\BIBentrySTDinterwordspacing

\bibitem{Hein2004}
\BIBentryALTinterwordspacing
M.~Hein, J.~Eisert, and H.~J. Briegel, ``Multiparty entanglement in graph
  states,'' \emph{Physical Review A}, vol.~69, no.~6, Jun 2004. [Online].
  Available: \url{http://doi.org/10.1103/PhysRevA.69.062311}
\BIBentrySTDinterwordspacing

\bibitem{Hein2006}
\BIBentryALTinterwordspacing
M.~Hein, W.~Dür, J.~Eisert, R.~Raussendorf, M.~Nest, and H.~Briegel,
  ``Entanglement in graph states and its applications,'' \emph{Quantum
  Computers, Algorithms and Chaos}, vol. 162, 03 2006. [Online]. Available:
  \url{https://doi.org/10.3254/978-1-61499-018-5-115}
\BIBentrySTDinterwordspacing

\bibitem{Heyfron2017}
\BIBentryALTinterwordspacing
L.~E. Heyfron and E.~T. Campbell, ``An efficient quantum compiler that reduces
  {T} count,'' \emph{Quantum Science and Technology}, vol.~4, no.~1, p. 015004,
  sep 2018. [Online]. Available: \url{https://doi.org/10.1088/2058-9565/aad604}
\BIBentrySTDinterwordspacing

\bibitem{gottesman1999}
\BIBentryALTinterwordspacing
D.~{Gottesman} and I.~L. {Chuang}, ``Demonstrating the viability of universal
  quantum computation using teleportation and single-qubit operations,''
  \emph{Nature}, vol. 402, no. 6760, pp. 390--393, 1999. [Online]. Available:
  \url{https://doi.org/10.1038/46503}
\BIBentrySTDinterwordspacing

\bibitem{Zeng2008}
\BIBentryALTinterwordspacing
B.~Zeng, X.~Chen, and I.~L. Chuang, ``Semi-clifford operations, structure of
  ${\mathcal{c}}_{k}$ hierarchy, and gate complexity for fault-tolerant quantum
  computation,'' \emph{Phys. Rev. A}, vol.~77, p. 042313, Apr 2008. [Online].
  Available: \url{https://doi.org/10.1103/PhysRevA.77.042313}
\BIBentrySTDinterwordspacing

\bibitem{simplex}
\BIBentryALTinterwordspacing
A.~Edgington, ``Simplex: a fast simulator for {Clifford} circuits.'' [Online].
  Available: \url{https://github.com/CQCL/simplex/releases/tag/v1.4.0}
\BIBentrySTDinterwordspacing

\end{thebibliography}

\end{document}